\documentclass[acmsmall,screen]{acmart}

\setcopyright{none}

\usepackage{amsfonts}
\usepackage{mathrsfs}
\usepackage{amsmath}
\usepackage{mathtools}
\usepackage{mathpartir}
\usepackage{bbding}
\usepackage{color}
\usepackage[capitalize,nameinlink]{cleveref}
\usepackage{xparse}
\usepackage{ifthen}
\usepackage{thmtools}
\usepackage{thm-restate}
\usepackage{subcaption}
\usepackage{tikz}
\usepackage{complexity-mod}

\usepackage{listings}

  \lstdefinestyle{tinyc}{
    basicstyle=\scriptsize\ttfamily,
    keywordstyle=\color{blue}
  }

  \lstdefinestyle{normalc}{
    basicstyle=\ttfamily,
    numbers=none,
    keywordstyle=\color{blue}
  }

  \lstdefinestyle{inlinec}{
    basicstyle=\ttfamily
  }

\usetikzlibrary{backgrounds}
\tikzstyle{every picture}+=[remember picture]

\usepackage{todonotes}
\newcommand{\RM}[1]{}
\newcommand{\gz}[1]{}
\newcommand{\pb}[1]{}
\newcommand{\rst}[1]{}

\newcommand{\N}{\mathbb{N}}

\DeclareDocumentCommand{\DCPS}{O{}}{\mathsf{DCPS}\ifthenelse{\equal{#1}{}}{}{[#1]}}
\DeclareDocumentCommand{\SRP}{O{}}{\mathsf{SRP}\ifthenelse{\equal{#1}{}}{}{[#1]}  }
\DeclareDocumentCommand{\HP}{O{}}{\mathsf{HP}\ifthenelse{\equal{#1}{}}{}{[#1]}} %

\newclass{\TWOEXPSPACE}{2EXPSPACE}

\usepackage{tikz}
\usetikzlibrary{petri}
\tikzstyle{every place}=[minimum size=5mm]
\tikzstyle{every transition}=[minimum size=5mm]
\usetikzlibrary{backgrounds,calc}

 \usetikzlibrary{ shapes, fit, arrows}
 \usetikzlibrary{decorations, snakes}
\DeclareMathSymbol{\mdot}{\mathord}{symbols}{"01}
\usepackage{enumerate}
\usepackage{lscape}
\def\qed{\rule{0.4em}{1.4ex}}

\def\@envspa{\hspace{0.3em}}
\def\@sa{\hspace{-0.2em}}
\def\@sb{\hspace{0.5em}}
\def\@sc{\hspace{-0.1em}}
\mathchardef\mhyphen="2D

{\bfseries\upshape}{\rm}
{\bfseries\upshape}{\itshape}
{\bfseries\upshape}{\itshape}

\def\set#1{{\{ #1 \}}}
\def\tuple#1{{\langle #1 \rangle }}
\def\norm#1{{|\!| #1 |\!| }}
\def\multi#1{{[\![ #1 ]\!]}}
\def\emulti{\mathbf{0}}
\def\nats{{\mathbb{N}}}
\def\integs{{\mathbb{Z}}}

\def\parikh{{\mathsf{Parikh}}}

\def\mmap{\mathbf{m}}
\def\nmap{\mathbf{n}}
\def\kmap{\mathbf{k}}
\def\bfz{\mathbf{0}}

\def\card#1{\lvert {#1} \rvert}

\def \subword{\sqsubseteq}
\newcommand{\dclosure}[1]{#1\mathord{\downarrow}}

\newcommand{\multiset}[1]{{\mathbb{M}[ #1 ]}}

\def\prod{\mathcal{P}}

\newcommand{\cC}{\mathcal{C}}
\newcommand{\cA}{\mathcal{A}}
\newcommand{\cB}{\mathcal{B}}
\newcommand{\cT}{\mathcal{T}}

\DeclareDocumentCommand{\langof}{O{} m}{%
  \mathsf{L}_{#1}(#2)%
  }
\DeclareDocumentCommand{\autstep}{O{}}{%
        \xrightarrow{#1}%
        }

\DeclareDocumentCommand{\autsteps}{O{}}{%
        \xrightarrow{#1}^*%
      }

\DeclareDocumentCommand{\powerset}{O{} m}{
  \mathbb{P}_{#1}(#2)
}

\mathchardef\mhyphen="2D
\DeclareDocumentCommand{\EXPTIME}{O{}}{%
  \ifthenelse{\equal{#1}{}}{%
    \mathsf{EXPTIME}%
  }{%
    {#1}\mhyphen\mathsf{EXPTIME}%
  }%
}

\DeclareDocumentCommand{\NTERM}{O{}}{\mathsf{NTERM}\ifthenelse{
\equal{#1}{}}{}{[#1]}  }
\DeclareDocumentCommand{\UBOUND}{O{}}{\mathsf{UBOUND}\ifthenelse{
\equal{#1}{}}{}{[#1]}  }
\DeclareDocumentCommand{\REACH}{O{}}{\mathsf{REACH}\ifthenelse{
\equal{#1}{}}{}{[#1]}  }
\DeclareDocumentCommand{\FAIRNTERM}{O{}}{\mathsf{FNTERM}
\ifthenelse{
\equal{#1}{}}{}{[#1]}  }
\DeclareDocumentCommand{\STARV}{O{}}{\mathsf{STARV}\ifthenelse{
\equal{#1}{}}{}{[#1]}  }
\DeclareDocumentCommand{\SREACH}{O{}}{\mathsf{SREACH}\ifthenelse{
\equal{#1}{}}{}{[#1]}  }

\def\Deltac{{\Delta_{\mathsf{c}}}}
\def\Deltai{{\Delta_{\mathsf{i}}}}
\def\Deltar{{\Delta_{\mathsf{r}}}}
\def\Deltau{{\Delta_{\mathsf{u}}}}
\def\Deltat{{\Delta_{\mathsf{t}}}}

\def \cP{\mathcal{P}}

\def \DCFS{\mathsf{DCFS}}
\def \seqA{S_{\cA}}

\def \bbq{\Omega}
\def \bbp{\Phi}
\def\balloonQ{{\bbq}}
\def\balloonP{{\bbp}}

\def \barq{\sigma}
\def \barp{\pi}
\def \barz{\bar{\mathbf{0}}}

\newcommand{\ver}{\mathsf{ver}}

\newcommand{\vmap}{\mathbf{v}}
\def\bmap{\mathbf{b}}
\newcommand{\cM}{\mathcal{M}}

\newcommand{\wmap}{\mathbf{w}}

\def\sup{\mathsf{supp}}

\def\bs{\mathsf{bs}}
\def\cS{\mathcal{S}}
\def\typeseq{\mathsf{typeseq}}
\def \min{\mathsf{min}}
\def \max{\mathsf{max}}
\def \check{\mathsf{check}}

\def \col{\mathsf{color}}
\newcommand{\ltr}[1]{\mathsf{#1}}
\def \lgreen{\mathsf{green}}
\def \lred{\mathsf{red}}

\DeclareDocumentCommand\PDA{}{\mathsf{PDA}}
\DeclareDocumentCommand\VASS{}{\mathsf{VASS}}
\DeclareDocumentCommand\VASSB{}{\mathsf{VASSB}}

\newcommand{\op}{\mathit{op}}
\newcommand{\newb}{\mathsf{inflate}}
\newcommand{\deflateb}{\mathsf{deflate}}
\newcommand{\burstb}{\mathsf{burst}}

\title[Context-Bounded Liveness-Verification for Multithreaded Shared-Memory Programs]{Context-Bounded Verification of Liveness Properties for Multithreaded Shared-Memory Programs}

\newcommand{\OurInstitution}{Max Planck Institute for Software Systems (MPI-SWS)}
\newcommand{\OurStreet}{Paul-Ehrlich-Stra{\ss}e, Building G26}
\newcommand{\OurCity}{Kaiserslautern}
\newcommand{\OurPostcode}{67663}
\newcommand{\OurCountry}{Germany}

\author{Pascal Baumann}
\orcid{0000-0002-9371-0807}             %
\affiliation{
  \institution{\OurInstitution}            %
  \streetaddress{\OurStreet}
  \city{\OurCity}
  \postcode{\OurPostcode}
  \country{\OurCountry}                    %
}
\email{pbaumann@mpi-sws.org}          %

\author{Rupak Majumdar}
\orcid{0000-0003-2136-0542}             %
\affiliation{
  \institution{\OurInstitution}            %
  \streetaddress{\OurStreet}
  \city{\OurCity}
  \postcode{\OurPostcode}
  \country{\OurCountry}                    %
}

\email{rupak@mpi-sws.org}          %

\author{Ramanathan S. Thinniyam}
\orcid{0000-0002-9926-0931}             %
\affiliation{
  \institution{\OurInstitution}            %
  \streetaddress{\OurStreet}
  \city{\OurCity}
  \postcode{\OurPostcode}
  \country{\OurCountry}                    %
}
\email{thinniyam@mpi-sws.org}          %

\author{Georg Zetzsche}
\orcid{0000-0002-6421-4388}             %
\affiliation{
  \institution{\OurInstitution}            %
  \streetaddress{\OurStreet}
  \city{\OurCity}
  \postcode{\OurPostcode}
  \country{\OurCountry}                    %
}
\email{georg@mpi-sws.org}          %

\begin{abstract}
We study context-bounded verification of liveness properties of multi-threaded, shared-memory programs,
where each thread can spawn additional threads.
Our main result shows that context-bounded fair termination is decidable for the model;
context-bounded implies that each spawned thread can be context switched a fixed constant number of times.
Our proof is technical, since fair termination requires reasoning about the composition of
unboundedly many threads each with unboundedly large stacks.
In fact, techniques for related problems, which depend crucially on replacing the pushdown threads
with finite-state threads, are not applicable.
Instead, we introduce an extension of vector addition systems with states (VASS), called VASS with balloons (VASSB), as an intermediate model;
it is an infinite-state model of independent interest.
A VASSB allows tokens that are themselves markings (balloons).
We show that context bounded fair termination reduces to fair
termination for VASSB.
We show the latter problem is decidable by showing a series of reductions: from fair termination to configuration
reachability for VASSB and thence to the reachability problem for VASS.
For a lower bound, fair termination is known to be non-elementary already in the special case where threads run to 
completion (no context switches).

We also show that the simpler problem of context-bounded
termination is $\TWOEXPSPACE$-complete, matching the complexity bound---and indeed 
the techniques---for safety verification.
Additionally, we show the related problem of \emph{fair starvation}, which checks if some thread can
be starved along a fair run, is also decidable in the context-bounded case.
The decidability employs an intricate reduction from fair starvation
to fair termination.
Like fair termination, this problem is also non-elementary.
\end{abstract}

\begin{document}

\maketitle

\sloppy %

\section{Introduction}
\label{sec:intro}

We study decision problems related to liveness verification of shared-memory multithreaded programs.
In a shared-memory multithreaded program, a number of \emph{threads}
execute concurrently.
Each thread executes possibly recursive sequential code, and can spawn new 
threads for concurrent execution.
The threads communicate through shared global variables that they can read and write.
The execution of the program is guided by a non-deterministic \emph{scheduler} that
picks one of the spawned threads to execute in each time step.
If the scheduler replaces the currently executing thread with a different one,
we say the current active thread is \emph{context switched}.

Shared-memory multithreaded programming is ubiquitous and static verification of
safety or liveness properties of such programs is a cornerstone of formal verification
research.
Indeed, there is a vast research literature on the problem---from a foundational understanding 
of the computability and complexity of (subclasses of) models,
to program logics, and to efficient tools for analysis of real systems.

In this paper, we focus on \emph{decidability} issues for \emph{liveness} verification for multithreaded
shared memory programs with the ability to spawn threads.
Liveness properties, intuitively, specify that ``something good'' happens when a program executes.
A simple example of a liveness property is \emph{termination}: the property that a program eventually
terminates.
In fact, termination is a ``canonical'' liveness property:
for a very general class of liveness properties, through monitor constructions, verifying liveness properties reduces to verifying termination \cite{AptOlderog,Vardi91}.

Unfortunately, under the usual notion of non-deterministic schedulers, some programs may fail
to terminate for uninteresting reasons.
Consider the following program:

\begin{lstlisting}[style=tinyc,name=cfl]
global bit := 1;
main() { spawn foo; spawn bar; } 
foo() { if bit = 1 then spawn foo; } 
bar() { bit := 0; }
\end{lstlisting}

A main thread spawns two additional threads $\mathtt{foo}$ and $\mathtt{bar}$.
The thread $\mathtt{foo}$ checks if a global bit is set and, if so, re-spawns itself.
The thread $\mathtt{bar}$ resets the global bit.
There is a non-terminating execution of this program in which $\mathtt{bar}$ is never scheduled.
However, a scheduler that never schedules a thread that is ready to run would be considered unfair.
Instead, one formulates the problem of \emph{fair termination}: termination under a \emph{fair} non-deterministic scheduler.
We abstract away from the exact mechanism of the scheduler, and only require that
every spawned thread that is infinitely often ready to run is eventually scheduled.
Then, every fair run of the above program is terminating: 
eventually $\mathtt{bar}$ is scheduled, after which $\mathtt{foo}$ does not spawn a new thread.

Fair termination of concurrent programs is highly undecidable.
A celebrated result by \citet{Harel1986} shows that fair termination is $\Pi_1^1$-complete;
in fact, the problem is already $\Pi_1^1$-complete when the global state is finite and 
there are a finite number of recursive threads.\footnote{
	Recall that the class $\Pi_1^1$ in the analytic hierarchy is the class of all relations on $\nats$ 
	that can be defined by a universal \emph{second-order} number-theoretic formula.
	}
In contrast, safety verification, modeled as state reachability, is merely $\Sigma_1^0$-complete.

Since the high undecidability relies on an unbounded exchange of information among threads,
a recent and apposite approach to verifying concurrent recursive programs is to explore only a representative
subset of program behaviors by limiting the number of inter-thread interactions \cite{QR05,MusuvathiQadeer}.
This approach, called {context bounding} by \citet{QR05}, considers the verification problem 
as a family of problems, one for each $K$.
The $K$-context bounded instance, for any fixed $K\geq 0$, considers only those executions where each thread is context switched
at most $K$ times by the scheduler.
In the limit as $K\rightarrow \infty$, the $K$-bounded approach explores all behaviors where each
thread runs a finite number of times. 
In practice, bounded explorations with small values of $K$ have proved to be 
effective to uncover many safety and liveness bugs in real systems.

In this paper, we prove the following results.
We first show that $K$-context bounded termination for multithreaded recursive programs with spawns is decidable and $\TWOEXPSPACE$-complete
when $K\geq 1$.
Then, we show that $K$-context bounded fair termination is decidable but non-elementary.
Our result implies fair termination is $\Pi_1^0$-complete when each thread is context-switched a finite number of times.
(Note that this does not contradict the $\Pi_1^1$-completeness of the general problem, in which a thread can be context switched infinitely often.)
We also study a stronger notion of fairness called \emph{fair
non-starvation}, where threads are given unique identities in order
to distinguish threads with the same local configuration, and show
that fair non-starvation is also decidable.

Our results generalize the special case of $K=0$ studied by \citet{GantyM12} as \emph{asynchronous programs}.
When $K=0$, each thread executes to completion without being interrupted in the middle.
Ganty and Majumdar show the decidability of safety and liveness verification for this model.
In particular, they prove safety and termination are both $\EXPSPACE$-complete and fair termination 
and fair starvation are decidable but non-elementary.\footnote{
	Their result shows a polynomial-time equivalence between fair termination and reachability in vector addition
	systems with states (VASS, a.k.a.\ Petri nets).
	The complexity bounds follow from our current knowledge of the complexity of VASS reachability \cite{CzerwinskiLLLM19}.
	}
Their proof depends on the observation that, since threads are not interrupted, one can replace the pushdown
automata for each thread by finite automata that accept Parikh-equivalent languages.
Unfortunately, their technique does not generalize when context switches are allowed.

For $K\geq 1$, \citet{AtigBQ2009} showed that the safety verification problem is
decidable in $\TWOEXPSPACE$.
Ten years later, a matching lower bound was shown by \citet{BaumannMajumdarThinniyamZetzsche2020a}.
The key observation in the decision procedure is that safety is preserved under downward closures:
one can analyze a related program where some spawned threads are ``forgotten.''
Since the downward closure of a context free language is effectively regular, one can replace
the pushdown automaton for each thread by a finite automaton accepting the downward closure.
In fact, our proof of termination also follows easily from this observation, as termination is also preserved
by downward closures.

Unfortunately, fair termination and fair non-starvation are not preserved under downward closures.
Thus, we cannot apply the preceding techniques to replace pushdown automata by finite automata in our construction.
Thus, our proof is more intricate and requires several insights into the computational model.

The key difficulty in our decision procedure is to maintain a finite representation for \emph{unboundedly} many active threads, 
each with \emph{unboundedly} large local stacks and potentially spawning \emph{unboundedly} many new
threads, and to compose their context-switched executions into a global execution.
In order to maintain and compose such configurations, we introduce a new model, called \emph{VASS with balloons} ($\VASSB$),
that extends the usual model of a vector addition systems with states (VASS) with ``balloons'': 
a token in a $\VASSB$ can be a usual VASS token or a balloon token that is itself a vector.
Intuitively, balloon tokens represent the possible new threads a thread can spawn along one of its execution segments.

We show through a series of constructions that the fair termination problem reduces to the fair termination
problem for $\VASSB$, and thence to the configuration reachability problem for $\VASSB$.
Finally, we show that configuration reachability for VASSB is decidable by a reduction to the reachability problem for VASS.
This puts $\VASSB$ in the rare class of infinite-state systems which generalize VASS and yet maintain a decidable reachability
(not just coverability!) problem.

Finally, we show a reduction from the fair starvation problem to
fair termination.  The reduction relies on two
combinatorial insights. The first is that if a program has an infinite
fair run, then it has one in which there exists a bound on the number
of threads spawned by each thread. The second is a novel pumping argument
based on Ramsey's theorem; it implies that it suffices to track a
finite amount of data about each thread to determine whether some
thread can be starved.

In conclusion, we prove decidability of liveness verification for multithreaded shared memory programs with the ability
to dynamically spawn threads, 
an extremely expressive model of multithreaded programming.
This model sits at the boundary of decidability and subsumes many other models studied before.

\smallskip
\noindent\textbf{Related Work.}
Safety verification for concurrent recursive
programs is already undecidable with just two threads and finite global store \cite{Ramalingam}.
Many results on context-bounded safety verification consider a model with a \emph{fixed} number of threads, without spawns.
The complexity of safety verification for this model is well understood at this point.
The key idea underlying the best algorithms reduce the problem to analyzing a sequential pushdown system \cite{LalReps}
by guessing the bounded sequences of context switches for each thread and using the finite state to ensure the
sequential runs can be stitched together.
 
When the model allows \emph{spawning} of new threads, as ours does, existing decision procedures are significantly
more complex, both in their technicalities and in computational cost.
There are relatively few results on decidability of liveness properties of infinite-state systems.
\citet{AtigBEL12} show a sufficient condition for fair termination for context-bounded executions
of a fixed number of threads, where they look for ultimately periodic executions, in which each thread
is context switched at most $K$ times in the loop.
They show that the search for such ultimately periodic executions can be reduced to safety verification.
In our model, fair infinite runs may involve
unboundedly many threads with unbounded stacks and need not be periodic---for example, 
there can always be more and more newly spawned threads.

Multi-pushdown systems model multithreaded programs with a fixed number of threads. 
Many decision procedures are known when the executions of such systems are restricted through different
bounds such as context, scope, or phase \cite{AtigBKS12,TorreNP16,AtigBKS17}, 
and also through limitations on communication patterns \cite{LalTKR08}.
These problems are orthogonal to us, either in the modeling capabilities or in the properties verified.

Decidability of linear temporal logic is known for weaker models of multithreaded recursive programs,
such as symmetric parameterized programs \cite{Kahlon08}
or leader-follower programs with non-atomic reads and writes \cite{Durand-Gasselin17,FortinMW17,MuschollSW17}.
These programs cannot perform compare-and-swap operations, and therefore, their computational power is quite
limited (in fact, LTL model checking is PSPACE-complete).
A number of heuristic approaches to fair termination of multithreaded programs provide sound but incomplete
algorithms, but for a more general class of programs involving infinite-state data variables \cite{CookPR11,PadonHLPSS18,FarzanKP16,CookPR07,KraglEHMQ20}.
The goal there is to provide a sound proof rule for verification but not to prove a decidability result.

In terms of fair termination problems for VASS, the theme of
computational hardness continues.  For example, the classical notion
of \emph{fair runs}, in which an infinitely activated transition has
to be fired infinitely often, leads to
undecidability~\cite{carstensen1987decidability} and even
$\Sigma_1^1$-completeness~\cite{howell1991taxonomy}. However,
\emph{weakly fair} termination, where only those transitions that are
almost always activated have to be fired infinitely often, is
decidable~\cite{janvcar1990decidability}. A rich taxonomy of fairness
notions with corresponding decidability results can be found in
\cite{howell1991taxonomy}. However, all of these notions appear to be
incomparable with our fairness notion for VASSB.

Our model of $\VASSB$ treads the boundary of models that generalize
VASS for which reachability can be proved to be decidable.  We note
that there are several closely related models, VASS with a
stack~\cite{leroux2015coverability} and branching
VASS~\cite{VermaGoubaultLarrecq2005}, for which decidability of
reachability is a long-standing open problem, and others, nested Petri
nets \cite{DBLP:conf/ershov/LomazovaS99}, for which reachability is
undecidable.

\section{Dynamic Networks of Concurrent Pushdown Systems ($\DCPS$)} %
\label{sec:preliminaries}

\subsection{Preliminary Definitions}

\subsubsection*{Multisets}
A \emph{multiset} $\mmap\colon S\rightarrow\nats$ over a set $S$ maps each
element of $S$ to a natural number.
Let $\multiset{S}$ be the set of all multisets over $S$.
We treat sets as a special case of multisets 
where each element is mapped onto $0$ or $1$.
We sometimes write
$\mmap=\multi{a_1,a_1,a_3}$ for the multiset
$\mmap\in\multiset{S}$ such that $\mmap(a_1)=2$, $\mmap(a_3)=1$, and $\mmap(a) = 0$ for each $a \in S\backslash\set{a_1,a_3}$. 
The empty multiset is denoted \(\emptyset\).
The size of a multiset $\mmap$, denoted $\card{\mmap}$, is
given by $\sum_{a\in S}\mmap(a)$.
This definition applies to sets as well.
 
Given two multisets $\mmap,\mmap'\in\multiset{S}$ we define $\mmap +
\mmap'\in\multiset{S}$ to be a multiset such that for all $a\in S$,
we have $(\mmap + \mmap')(a)=\mmap(a)+\mmap'(a)$.
For $c\in\nats$, we define $c\mmap$ as the multiset that maps each $a\in S$ to $c\cdot \mmap(a)$.
We also define the natural order
$\preceq$ on $\multiset{S}$ as follows: $\mmap\preceq\mmap'$ if{}f there
exists $\mmap^{\Delta}\in\multiset{S}$ such that
$\mmap + \mmap^{\Delta}=\mmap'$. 
We also define $\mmap - \mmap'$ for $\mmap' \preceq \mmap$ analogously: for all $a\in S$,
we have $(\mmap - \mmap')(a)=\mmap(a)-\mmap'(a)$. 

\subsubsection*{Pushdown Automata}
A \emph{pushdown automaton ($\PDA$)}
$\cP_{(g,\gamma)} = (Q, \Sigma, \Gamma, E, q_0, \gamma_0, Q_F)$ consists of 
a finite set of \emph{states} $Q$, 
a finite input alphabet $\Sigma$, 
a finite alphabet of \emph{stack symbols} $\Gamma$, 
an \emph{initial state} $q_0 \in Q$, 
an \emph{initial stack symbol} $\gamma_0 \in \Gamma$, 
a set of \emph{final states} $Q_F \subseteq Q$,
and a transition relation $E \subseteq (Q \times \Gamma) \times \Sigma_\varepsilon \times (Q \times \Gamma^{\leq 2})$, where $\Sigma_\varepsilon = \Sigma \cup \{\varepsilon\}$ and $\Gamma^{\leq 2} = \{\varepsilon\} \cup \Gamma \cup \Gamma^2$. For $((q,\gamma),a,(q',w)) \in E$ we also write $q \xrightarrow{a|\gamma/w} q'$.

The set of \emph{configurations} of $\mathcal{P}$ is $Q \times \Gamma^*$. The \emph{initial configuration} is $(q_0, \gamma_0)$. The set of \emph{final configurations} is $Q_F \times \Gamma^*$. For each $a \in \Sigma \cup \{\varepsilon\}$, the relation $\xRightarrow{a}$ on configurations of $\mathcal{P}$ is defined as follows: $(q,\gamma w) \xRightarrow{a} (q',w'w)$ for all $w \in \Gamma^*$ iff (1) there is a transition $q \xrightarrow{a|\gamma/w'} q' \in E$, or (2) there is a transition $q \xrightarrow{a|\varepsilon} q' \in E$ and $\gamma = w' = \varepsilon$.

For two configurations $c, c'$ of $\mathcal{P}$, we write $c \Rightarrow c'$ if $c \xRightarrow{a} c'$ for some $a$. Furthermore, we write $c \xRightarrow{\smash{u}}^* c'$ for some $u \in \Sigma^*$ if there is a sequence of configurations $c_0$ to $c_n$ with 
\[
  c = c_0 \xRightarrow{a_1} c_1 \xRightarrow{a_1} c_2 \cdots c_{n-1} \xRightarrow{a_n} c_n = c', 
\]
such that $a_1 \ldots a_n = u$. We then call this sequence a \emph{run} of $\mathcal{P}$ over $u$. We also write $c \Rightarrow^* c'$ if the word $u$ does not matter. A run of $\mathcal{P}$ is \emph{accepting} if $c$ is initial and $c'$ is final.
The \emph{language} accepted by $\mathcal{P}$, denoted $L(\mathcal{P})$
is the set of words $\in \Sigma^*$, over which there is an accepting
run of $\mathcal{P}$.

Given two configurations $c, c'$ of $\mathcal{P}$ with $c \Rightarrow^* c'$, we say that $c'$ is \emph{reachable} from $c$ and 
that $c$ is \emph{backwards-reachable} from $c'$. If $c$ is the initial configuration, we simply say that $c'$ is reachable.

\subsubsection*{Parikh Images and Semi-linear Sets}
The Parikh image of a word $u \in \Sigma^*$ is a function $\parikh(u): \Sigma \rightarrow \nats$
such that, for every $a\in \Sigma$, we have $\parikh(u)(a) = |u|_a$, where $|u|_a$ denotes the number of occurrences
of $a$ in $u$. 
We extend the definition to the Parikh image of a language $L \subseteq \Sigma^*$:
$\parikh(L) = \set{\parikh(u) \mid u\in L}$.
We associate the natural isomorphism between $\nats^\Sigma$ and $\nats^{|\Sigma|}$ and
consider the functions as vectors of natural numbers.

A subset of $\multiset{S}$ is \emph{linear} if it is of the form
$\set{\mmap_0 + t_1 \mmap_1 + \ldots + t_n \mmap_n \mid t_1,\ldots,t_n \in \nats}$
for some multisets $\mmap_0, \mmap_1,\ldots, \mmap_n \in \multiset{S}$. 
We call $\mmap_0$ the \emph{base vector} and $\mmap_1,\ldots,\mmap_n$ the \emph{period vectors}.
A linear set has a finite representation based on its base and period vectors.
A \emph{semi-linear} set is a finite union of linear sets.

\begin{theorem}[\cite{Parikh66}]\label{th:parikh}
For any context-free language $L$, the set $\parikh(L)$ is semi-linear.
A representation of the semi-linear set $\parikh(L)$ can be effectively constructed from a $\PDA$ for $L$.
\end{theorem}

\subsection{Dynamic Networks of Concurrent Pushdown Systems} \label{sec:dcps}

A \emph{Dynamic Network of Concurrent Pushdown Systems ($\DCPS$)} 
$\cA = (G,\Gamma,\Delta,g_0,\gamma_0)$ 
consists of 
a finite set of \emph{(global) states} $G$, 
a finite alphabet of \emph{stack symbols} $\Gamma$, 
an \emph{initial state} $g_0 \in G$, 
an \emph{initial stack symbol} $\gamma_0 \in \Gamma$, and 
a finite set of \emph{transition rules} $\Delta$. 
The set of transition rules $\Delta$ is partitioned into four kinds of rules:
\emph{creation rules} $\Deltac$, \emph{interruption rules} $\Deltai$, \emph{resumption rules} $\Deltar$, and
\emph{termination rules} $\Deltat$.
Elements of $\Deltac$ have one of two forms:
(1) $g|\gamma \hookrightarrow g'|w'$, or 
(2) $g|\gamma \hookrightarrow g'|w' \triangleright \gamma'$, where $g,g' \in G$, $\gamma,\gamma' \in \Gamma$, $w' \in \Gamma^*$, and $|w'| \leq 2$. 
Rules of type~(1) allow the $\DCPS$ to take a single step in one of its threads.
Rules of type~(2) additionally spawn a new thread with top of stack $\gamma'$.
Elements of $\Deltai$ have the form $g|\gamma \mapsto g'|w'$, where $g, g'\in G$, $\gamma\in\Gamma$, and $w'\in\Gamma^*$ with $1\le |w'|\leq 2$.
Elements of $\Deltar$ have the form $g\mapsto g'\lhd \gamma$, where $g, g'\in G$ and $\gamma\in\Gamma$. 
Elements of $\Deltat$ have the form $g\mapsto g'$, where $g,g'\in G$.

The \emph{size} $|\cA|$ of $\cA$ is defined as $|G| + |\Gamma| + |\Delta|$:
the number of symbols needed to describe the global states, the stack alphabet, and the transition rules.

The set of configurations of $\cA$ is $G \times \big((\Gamma^*
\times \nats) \cup \set{\#}\big) \times \multiset{\Gamma^* \times \nats}$. 
Given a configuration $\langle g, (w,i), \mmap \rangle$, we call $g$ the 
\emph{(global) state}, $(w,i)$ the \emph{local configuration} of the 
\emph{active thread}, 
and $\mmap$ the multiset of the \emph{local
configurations} of the \emph{inactive threads}.
In a configuration $\langle g,\#,\mmap \rangle$, we call $\#$ a \emph{schedule point}.

The initial configuration of $\cA$ is
$\langle g_0, \#, \multi{(\gamma_0,0)}\rangle$. For a configuration
$c$ of $\cA$, we will sometimes write $c.g$ for the state of $c$ and
$c.\mmap$ for the multiset of threads of $c$ (both active and
inactive).  The \emph{size} of a configuration
$c = \langle g, (w,i), \mmap \rangle$ is defined as
$|c| = |w| + \sum_{(w',j) \in \mmap}|w'|$.

Intuitively, a $\DCPS$ represents a multi-threaded, shared memory program. 
The global states $G$ represent the shared memory. 
Each thread is potentially recursive.
It maintains its own stack $w$ over the 
stack alphabet $\Gamma$ and uses the transition rules in $\Deltac$ to 
manipulate the global state and its stack.
It can additionally spawn new threads using rules of type (2) in $\Deltac$.
In a local configuration, the natural number $i$ keeps track of how many times a thread 
has already been context switched by the underlying scheduler.
Any newly spawned thread has its context switch number set to $0$.

The steps of a single thread defines the following
\emph{thread step} relation $\rightarrow$ on
configurations of $\cA$: 
we have $\langle g, (\gamma w,i), \mmap \rangle \rightarrow \langle g', (w'w,i), \mmap' \rangle$ 
for all $w \in \Gamma^*$ iff (1) there is a 
rule $g|\gamma \hookrightarrow g'|w'$ in  $\Deltac$ and $\mmap' = \mmap$
or (2) there is a rule $g|\gamma \hookrightarrow g'|w' \triangleright 
\gamma'$ in $\Deltac$ and $\mmap' = \mmap + \multi{(\gamma',0)}$. 
We extend the \emph{thread step} relation $\rightarrow^+$ to be the irreflexive-transitive closure of 
$\rightarrow$; thus $c \rightarrow^+c'$ if there is a sequence $c
\rightarrow c_1 \rightarrow \ldots c_k \rightarrow c'$ for some $k\geq 0$.

A non-deterministic scheduler switches between concurrent threads.
The active thread is the one currently being executed and the
multiset $\mmap$ keeps all other partially executed threads in the system.
Any spawned thread is put in $\mmap$ for future execution with an initial context switch number $0$.
The scheduler may interrupt a thread based on the interruption rules and non-deterministically
resume a thread based on the resumption rules.

The actions of the scheduler define the
\emph{scheduler step} relation $\mapsto$ on configurations of $\cA$: 
\[
\small
\inferrule[Swap]{
g|\gamma\mapsto g'|w' \in \Deltai
}{
\tuple{g, (\gamma w, i), \mmap} \mapsto \tuple{g', \#, \mmap+ \multi{w'w,i+1}}
}
\;
\inferrule[Resume]{
g\mapsto g'\lhd \gamma\in \Deltar
}{
\tuple{g, \#, \mmap + \multi{\gamma w, i}} \mapsto \tuple{g', (\gamma w, i), \mmap}
}
\;
\inferrule[Term]{
  g\mapsto g'\in\Deltat
}{
  \tuple{g, (\varepsilon, i), \mmap} \mapsto \tuple{g', \#, \mmap}
}
\]
If a thread can be interrupted, then {\sc Swap} swaps it out and increases the
context switch number of the thread.
The rule {\sc Resume} picks a thread that is ready to run based on the current
global state and its top of stack symbol and makes it active. 
The rule {\sc Term} removes a thread on termination (empty stack). 

A \emph{run} of a $\DCPS$ is a finite or infinite sequence of
alternating thread execution and scheduler step relations 
\[
  c_0 \rightarrow^+ c_0' \mapsto c_1 \rightarrow^+ c_1' \mapsto \ldots 
\]
such that $c_0$ is the initial configuration.
The run is \emph{$K$-context switch bounded} if, moreover, for each $j\geq 0$, the
configuration $c_j = (g, (w, i), \mmap)$
satisfies $i \leq K$.
In a $K$-context switch bounded run, each thread is context switched
at most $K$ times and the scheduler never schedules a
thread that has already been context switched $K+1$ times.
When the distinction between thread and scheduler steps is not
important, we write a run as a sequence $c_0 \Rightarrow c_1 \ldots$.

\subsection{Identifiers and the Run of a Thread}

Our definition of $\DCPS$ does not have thread identifiers associated with a thread.
However, it is convenient to be able to identify the run of a single thread along the execution.
This can be done by decorating local configurations with unique identifiers and modifying the
thread step for $g|\gamma \hookrightarrow g'|w\triangleright \gamma'$ to add a thread $(\ell,\gamma',0)$
to the multiset of inactive threads, where $\ell$ is a fresh identifier. 
By decorating any run with identifiers, we can freely talk about the run of a single thread,
the multiset of threads spawned by a thread, etc.

Let us focus on the run of a specific thread,
that starts executing from some global state $g$ with an initial stack
symbol $\gamma$.
In the course of its run, the thread updates its own local stack and spawns new threads,
but it also gets swapped out and swapped back in.

We show that the run of a thread corresponds to the run of an associated $\PDA$ that
can be extracted from $\cA$.
This $\PDA$ updates the global state and the stack based on the rules
in $\Deltac$, but additionally
(1) makes visible as the input alphabet the initial symbols (from
$\Gamma$) of the spawned threads, and
(2) non-deterministically guesses jumps between global states
corresponding to the effect of context switches.
There are two kinds of jumps.
A jump $(g_1, \gamma, g_2)$ in the $\PDA$ corresponds to the thread being
switched out leading to global state $g_1$ and later resuming at global state
$g_2$ with $\gamma$ on top of its stack (without being active in the interim).
A jump $(g,\bot)$ corresponds to the last time the $\PDA$ is swapped out
(leading to global state $g$).
We also make these guessed jumps visible as part of the input alphabet.
Thus, the input alphabet of the $\PDA$ is $\Gamma \cup G \times \Gamma \times G \cup G \times \set{\bot}$.

For any $g\in G$ and $\gamma \in \Gamma$, we define the
$\PDA$ $\cP_{(g,\gamma)} = (Q, \Sigma,
\Gamma_\bot, E, \mathsf{init}, \bot, \set{\mathsf{init},\mathsf{end}})$,
where
$Q = G \cup G \times \Gamma \cup \set{\mathsf{init}, \mathsf{end}}$,
$\Gamma_\bot$ = $\Gamma \cup \set{\bot}$,
$\Sigma =\Gamma \cup G\times \Gamma \times G \cup G \times \set{\bot}$,
$E$ is the smallest transition relation such that
\begin{enumerate}
\item There is a transition $\mathsf{init}\xrightarrow{\varepsilon|\bot/\gamma\bot}g$ in $E$,
\item For every $g_1|\gamma_1 \hookrightarrow g_2|w \in \Delta_{\mathsf{c}}$ there is a transition 
  $g_1\xrightarrow{\varepsilon|\gamma_1/w}g_2$ in $E$,
\item For every $g_1|\gamma_1 \hookrightarrow g_2|w\triangleright \gamma_2 \in \Delta_{\mathsf{c}}$ there is a transition 
  $g_1\xrightarrow{\gamma_2|\gamma_1/w}g_2$ in $E$,
\item For every $g_1|\gamma_1 \hookrightarrow g_2|w \in \Delta_{\mathsf{i}}$, $g_3 \in G$, and $\gamma_2 \in \Gamma$ there is a transition 
  $g_1 \xrightarrow{(g_2,\gamma_2,g_3)|\gamma_1/w} (g_3, \gamma_2)$, and a transition
  $(g_3, \gamma_2) \xrightarrow{\varepsilon|\gamma_2/\gamma_2} g_3$ in $E$,
\item For every $g_1|\gamma_1 \hookrightarrow g_2|w \in \Delta_{\mathsf{i}}$ and every $\gamma_2 \in \Gamma_\bot$ there is a transition
  $g_1 \xrightarrow{\varepsilon|\gamma_1/w} (g_2, \bot)$, and a transition
  $(g_2,\bot) \xrightarrow{(g_2,\bot)|\gamma_2/\gamma_2} \mathsf{end}$ in $E$, \label{tPDA:tomod}
\item For every $g_1 \hookrightarrow g_2 \in \Deltat$ there is a transition
  $g_1 \xrightarrow{(g_2,\bot)|\bot/\bot} \mathsf{end}$ in $E$.
\end{enumerate}
The set of behaviors of the $\PDA$ $\cP_{(g,\gamma)}$ which
correspond to a thread execution with precisely $i$ ($i \leq K+1$)
context switches is given by the following language:
\[
 L^{(i)}_{(g,\gamma)} = L(\cP_{(g, \gamma)}) \cap ((\Gamma^*\cdot G\times
 \Gamma \times G)^{i-1}(\Gamma^*\cdot G \times \set{\bot}))\
\]
The language $L^{(i)}_{(g,\gamma)}$ is a context-free language.
In the definition, we use the end of stack symbol $\bot$ to recognize when the stack is empty.

\subsection{Decision Problems and Main Results}

\subsubsection*{Previous Work: Safety}
The \emph{reachability problem} for $\DCPS$ asks, given a global state $g$ of $\cA$,
if there is a run $c_0 \Rightarrow c_1 \ldots \Rightarrow c_n$ such
that $c_n.g = g$.
It is well-known that reachability is undecidable (e.g., one can reduce the emptiness problem for intersection of context free languages).
Therefore, it is customary to consider \emph{context-switch bounded} decision questions.
Given $K \in \nats$, a state $g$ of $\cA$ is $K$-context switch bounded
reachable if there is a $K$-context switch bounded run $c_0
\Rightarrow \ldots \Rightarrow c_n$ with $c_n.g = g$.
For a fixed $K$, the \emph{$K$-bounded state reachability problem} ($\SRP[K]$) for a $\DCPS$ is defined as follows:
\begin{description}
  \item[Given] A $\DCPS$ $\cA$ and a global state $g$
  \item[Question] Is $g$ $K$-context switch bounded reachable in $\cA$?
\end{description}
This problem is known to be decidable; the $\TWOEXPSPACE$ upper bound for each $K$ was proved by \citet{AtigBQ2009} and a matching
lower bound for $K \geq 1$ by \citet{BaumannMajumdarThinniyamZetzsche2020a}.
In case $K=0$, the problem is known to be $\EXPSPACE$-complete \cite{GantyM12}.

\subsubsection*{This Paper: Liveness}
We now turn to context-bounded liveness specifications.
The simplest liveness specification is \emph{(non-)termination}: does a program
halt?
For a fixed $K\in \nats$, the \emph{$K$-bounded non-termination problem} $\NTERM[K]$ 
is defined as follows:
	\begin{description}
  \item[Given] A $\DCPS$ $\cA$. 
  \item[Question] Is there an infinite $K$-context switch bounded run?
\end{description}
When $K=0$, the non-termination problem is known to be $\EXPSPACE$-complete \cite{GantyM12}.
We show the following result.

\begin{theorem}[Termination]
\label{thm:term}
For each $K \geq 1$, the problem $\NTERM[K]$ is $\TWOEXPSPACE$-complete.
\end{theorem}

\subsubsection*{Fairness}

An infinite run is \emph{fair} if, intuitively, 
any thread that can be executed is eventually executed by the scheduler. 
Fairness is used as a way to rule out non-termination due to uninteresting scheduler choices.

We say a thread $t= (\gamma w,i)$ is \emph{ready} at a configuration $c=(g,\#,\mmap)$
if $t\in\mmap$ and there is some rule $g \mapsto g' \lhd \gamma$ in $\Deltar$.
A thread $t$ is \emph{scheduled} at $c$ if the scheduler
step makes $t$ the active thread. 
A run is unfair to thread $t$ if it is ready infinitely often but never scheduled.
A \emph{fair} run $\rho$ is one which is not unfair to any thread.
Restricting our attention to $K$-context switch bounded runs gives us the
corresponding notion of fair $K$-context switch bounded runs. 

For fixed $K\in\nats$, the \emph{$K$-context bounded fair non-termination problem} $\FAIRNTERM[K]$ asks:
\begin{description}
  \item[Given] A $\DCPS$ $\cA$. 
  \item[Question] Is there an infinite, fair $K$-context switch bounded run?
\end{description}

Note that since our model does not have individual thread identifiers, fairness is defined only 
over equivalence classes of threads that have the same stack $w$ and the same context switch number~$i$.
The reason for our taking into account stacks and context switch numbers is the following.
It is a simple observation that there exists an infinite fair run in our sense
if and only if there exists a run in the corresponding system \emph{with thread identifiers}--that is
fair to each individual thread.
This is because an angelic scheduler could always pick the earliest spawned thread among
those with the same stack and context switch number.
Therefore, our results allow us to reason about multi-threaded systems with identifiers.

This raises the question of whether there are runs that are fair in our
sense, but where a non-angelic scheduler would still yield unfairness for some thread identity. In other words, is it possible that
 a fair run \emph{starves} a specific thread. 
For example, consider a program in which the main thread spawns two copies of a thread $\mathtt{foo}$.
Each thread $\mathtt{foo}$, when scheduled, simply spawns another copy of $\mathtt{foo}$ and terminates.
Here is a fair run of the program (we omit the global state as it is not relevant), where we have decorated the threads with
identifiers:
\begin{align*}
(\#, \multi{(\mathtt{main},0)^0}) \xRightarrow{*} & 
((\mathtt{main}, 0)^0, \multi{}) \xRightarrow{*}
(\#, \multi{(\mathtt{foo},0)^1,(\mathtt{foo},0)^2}) \xRightarrow{*} 
((\mathtt{foo},0)^2, \multi{(\mathtt{foo},0)^1}) \xRightarrow{*} \\
& 
(\#, \multi{(\mathtt{foo},0)^1,(\mathtt{foo},0)^3}) \xRightarrow{*}
((\mathtt{foo},0)^3, \multi{(\mathtt{foo},0)^1}) \xRightarrow{*} \ldots
\end{align*}
The run is fair, but a specific thread marked with identifier $1$ is never picked.

Formally, a thread $t=(w,i)$ is \emph{starved} in an infinite fair run $\rho = c_0 \Rightarrow c_1 \Rightarrow \ldots$ iff
there is some $j$ such that $c_{i}.\mmap(t) \geq 1$ for all $i \geq j$ and whenever $t$ is resumed at $c_k$ for $k \geq j$,
we have $c.\mmap(t) \geq 2$. 

For fixed $K\in\nats$, the \emph{$K$-bounded fair starvation problem} $\STARV[K]$ is
defined as follows:
\begin{description}
  \item[Given] A $\DCPS$ $\cA$. 
  \item[Question] Is there an infinite, fair $K$-context switch bounded run that starves some thread?
\end{description}

We show the following results.

\begin{theorem}[Fair Non-Termination]
\label{thm:fair-nterm}
For each $K\in \nats$, the problem $\FAIRNTERM[K]$ is decidable.
\end{theorem}

\begin{theorem}[Fair Starvation]
\label{thm:fair-starv}
For each $K\in \nats$, the problem $\STARV[K]$ is decidable.
\end{theorem}

Previously, decidability results were only known when $K=0$ \cite{GantyM12}.
Recall that a decision problem is \emph{nonelementary} if it is not in $\bigcup_{k\geq 0} k\mhyphen\EXPTIME$.
Our algorithms are nonelementary: they involve an (elementary) reduction to the reachability problem 
for vector addition systems with states (VASS).
This is unavoidable: already for $K=0$, the fair non-termination and fair starvation problems 
are non-elementary, because there is a reduction from the reachability problem for VASS~\cite{GantyM12}, which is non-elementary~\cite{CzerwinskiLLLM19}.

In the rest of the paper, we prove Theorems~\ref{thm:term}, \ref{thm:fair-nterm}, and \ref{thm:fair-starv}.

%

%
%
%
%

%

\section{Warm-Up: Non-termination} %
\label{sec:non_termination_and_boundedness}

In this section, we prove Theorem~\ref{thm:term}.
The theorem follows easily from previous
results for safety verification \cite{AtigBQ2009j,BaumannMajumdarThinniyamZetzsche2020a}.
We recall the main ideas as a step toward the more complex proof for
\emph{fair} termination. 

\subsection*{Downward Closures: From $\DCPS$ to $\DCFS$}
A $\DCPS$
$\cA = (G,\Gamma,\Delta,g_0,\gamma_0)$
is called a dynamic network of concurrent \emph{finite} systems ($\DCFS$) if
in each transition rule in $\Deltac\cup\Deltai$, we have $|w'| \leq 1$.
Intuitively, a $\DCFS$ corresponds to the special case where each
thread is a finite-state process (and each stack is bounded by $1$).

We reduce the $K$-bounded non-termination problem for $\DCPS$
to the non-termination problem for $\DCFS$.
Fix $K\in\nats$ and a $\DCPS$ $\cA = (G, \Gamma, \Delta, g_0,
\gamma_0)$.
The crucial observation of \citet{AtigBQ2009j} is that answer to the $K$-bounded
reachability problem remains unchanged if we allow threads to
``drop'' some spawned threads.
That is, for every $g|\gamma\hookrightarrow g'|w'\triangleright
\gamma'$, we also add the rule $g|\gamma\hookrightarrow g'|w'$ to $\Delta_{\mathsf{c}}$.
Informally, the ``forgotten'' spawned thread $\gamma'$ is never
scheduled.
Clearly, a global state is reachable in the original $\DCPS$ iff it
is reachable in the new $\DCPS$.

We observe that this transformation also preserves non-termination:
if there is a ($K$-bounded) non-terminating run in the original
$\DCPS$, there is one in the new one.

The ability to forget spawned tasks allows us to transform the
language $L^{(i)}_{(g,\gamma)}$ of each thread into a \emph{regular} language by
taking \emph{downward closures}.

We need some definitions.
For any alphabet $\Sigma$, 
define the subword relation $\subword \subseteq
\Sigma^*\times\Sigma^*$ as follows:
for every $u, v \in\Sigma^*$, we have $u\subword v$ iff
$u$ can be obtained from $v$
by deleting some letters from $v$. For example,
$acbba \subword bacbacbac$ but $abba \not \subword baba$.
The \textit{downward closure} $\dclosure{w}$ with respect to the
subword order of a word $w\in\Sigma^*$ is defined as $\dclosure{w} :=
\{w'\in\Sigma^* \mid w' \subword w\}$. The downward closure $\dclosure{L}$ of a language
$L\subseteq\Sigma^*$ is given by $\dclosure{L}:=\{ w'\in\Sigma^* \mid \exists w \in L \colon w' \subword
w\}$.
Recall that the downward closure $\dclosure{L}$ of any language $L$ is a regular
language \cite{haines1969free}.
Moreover, a finite automaton accepting the downward closure of a
context-free language can be effectively constructed \cite{Courcelle}.
The size of the resulting automaton is at most exponential in the size of
the $\PDA$ for the context-free language \cite{BachmeierLS2015}.

Now consider the following language:
\[
  \hat{L}_{(g,\gamma)} = \bigcup_{i=1}^{K+1} \bigg(
  \dclosure{L^{(i)}_{(g,\gamma)}} \cap \Big(\big(\Gamma^*\cdot (G\times \Gamma \times G)\big)^{i-1}(\Gamma^*\cdot G \times \set{\bot})\Big)\bigg)
\]
This language is regular and can be effectively constructed from the
$\PDA$ $\cP_{(g,\gamma)}$.
It accepts all behaviors of a thread that is context switched at
most $K+1$ times such that, by adding additional spawned tasks, one
gets back a run of the original thread in $\cA$.

The $\DCFS$ simulates the downward closure of the $\DCPS$ by simulating the
composition of the automata for each downward closure. 
The construction is identical to \cite[Lemma 5.3]{AtigBQ2009j}.
Thus, we can conclude with the following lemma.

\begin{lemma}\label{lem:dcps-to-dcfs}
      The $K$-bounded non-termination problem for $\DCPS$ can be
      reduced in exponential time to the non-termination problem for
      $\DCFS$.
      The resulting $\DCFS$ is of size at most exponential in the size
      of the $\DCPS$.
\end{lemma}

\subsection*{From $\DCFS$ Non-Termination to $\VASS$ Non-Termination}
A \emph{vector addition system with states} ($\VASS$) is a tuple
$V = (Q, P, E)$ where
$Q$ is a finite set of \emph{states}, $P$ is a finite set of
\emph{places}, 
and
$E$ is a finite set of edges of the form $q \autstep[\delta] q'$ where
$\delta \in \integs^P$.
A \emph{configuration} of the $\VASS$ is a pair $(q,u) \in Q \times \multiset{P}$.
The edges in $E$ induce a transition relation on configurations: 
there is a transition $(q, u)\autstep[\delta](q', u')$ if
there is an edge $q\autstep[\delta] q'$ in $E$ such that
$u'(p) = u(p) +\delta(p)$ for all $p\in P$.
A \emph{run} of the $\VASS$ is a finite or infinite sequence of
configurations $c_0 \autstep[\delta_0] c_1 \autstep[\delta_1]\ldots$.
The \emph{non-termination} problem for $\VASS$ asks, given a $\VASS$ and an initial configuration $c_0$, 
is there an infinite run starting from $c_0$. 

\begin{lemma}\label{lem:dcfs-to-vass}
The $K$-context bounded non-termination problem for $\DCFS$ can be reduced in polynomial
time to non-termination problem for $\VASS$.
\end{lemma}

\begin{proof}
  Let $\cA = (G, \Gamma, \delta, g_0,\gamma_0)$ be a $\DCFS$.
  We define a $\VASS$ $V(\cA) = (G\times (\Gamma\times \set{0,\ldots, K}\cup \set{\#}), (\Gamma \cup \set{\varepsilon})\times\set{0,\ldots, K+1}, E)$.
  Intuitively, a configuration $((g,\gamma, i), u)$ of the $\VASS$ represents a configuration of
  the $\DCFS$ where the global state is $g$, the active thread has stack $\gamma$ and has been previously context switched $i$ times,
  and for each $\gamma'\in\Gamma$ and $i\in\set{0,\ldots, K+1}$, the value $u(\gamma',i)$ represents the number of
  pending threads with stack $\gamma'$ which have each been context switched $i$ times.
  A global state $(g,\#)$ indicates a state where the scheduler picks a new thread.
  The edges in $E$ update the configurations to simulate the steps of the $\DCFS$. 
  
  For each transition $g|\gamma\hookrightarrow g'|\gamma' \in \Deltac$ and for each $i\in\set{0,\ldots,K}$, the $\VASS$ has
  a transition that changes $(g,\gamma,i)$ to $(g',\gamma',i)$.
  For each transition  $g|\gamma\hookrightarrow g'|\gamma'\triangleright\gamma'' \in \Deltac$ and for each $i\in\set{0,\ldots,K}$, the $\VASS$ has
  a transition that changes $(g,\gamma,i)$ to $(g',\gamma',i)$ and puts a token in $(\gamma'',0)$.
  For each transition $g|\gamma \mapsto g'|\gamma' \in \Deltai$ and for each $i\in\set{0,\ldots, K}$, the $\VASS$ has a transition that changes $g$ to $(g', \#)$
  while putting a token into $(\gamma',i+1)$. 
  For each $g\mapsto g'\lhd \gamma \in \Deltar$, there is a transition $(g,\#)$ to $(g', \gamma, i)$ that takes a token from $(\gamma,i)$.
  For each $g\mapsto g' \in \Deltat$, there is a transition $(g,\varepsilon,i)$ to $(g', \#)$.

Clearly, there is a bijection between the runs of
$\cA$ and the runs of the $\VASS$ from $((g_0,\#),\multi{\gamma_0,0})$.
Thus, there is an infinite run in $\cA$ iff there is an
infinite run in $V(\cA)$ from $((g_0,\#),\multi{\gamma_0,0})$.
\end{proof}

\subsection*{Proof of Theorem~\ref{thm:term}}
The $\TWOEXPSPACE$ upper bound follows by combining
Lemmas~\ref{lem:dcps-to-dcfs}, \ref{lem:dcfs-to-vass}, and the
$\EXPSPACE$ upper bound for the non-termination problem for $\VASS$ \cite{Rackoff78}. 

The $\TWOEXPSPACE$ lower bound follows from the observation made 
already in \cite{BaumannMajumdarThinniyamZetzsche2020a} that the $\TWOEXPSPACE$-hardness of 
$K$-bounded reachability already holds for \emph{terminating}
$\DCPS$, in which every run is terminating.
It is now a simple reduction to take an instance of the $K$-bounded
state reachability problem for a terminating $\DCPS$ and add a
``gadget'' that produces an infinite run whenever the target global
state is reached.

\section{Fair Non-Termination}

We now turn to proving Theorem~\ref{thm:fair-nterm}.
Unfortunately, fair termination is not preserved under downward closure.
The example in Section~\ref{sec:intro} has no fair infinite run, since eventually 
(under fairness), \texttt{bit} is set to $1$ by the instance of \texttt{bar} and the program terminates.
However, the downward closure that omits \texttt{bar} has a fair infinite run.
Thus, we cannot replace the $\PDA$s for each thread with finite-state automata and there is no
obvious reduction to $\VASS$.

Our proof is more complicated.
First, we introduce an extension, VASS with \emph{balloons} ($\VASSB$), of $\VASS$ (Section~\ref{sec:vassb}).
A $\VASSB$ extends a $\VASS$ with balloon states and balloon places, and allows keeping
multisets of state-vector pairs over balloons.
We can use this additional power to store spawned threads.
As we shall see (Section~\ref{sec:dcps-to-vassb}), we can reduce $\DCPS$ to $\VASSB$.
Later, we shall show decidability of fair infinite behaviors for $\VASSB$, completing the proof. 

\subsection{VASS with Balloons}
\label{sec:vassb}

\def \Op{\mathsf{OP}}

A \emph{VASS with balloons} ($\VASSB$) is a tuple $\mathcal{V} =(Q,P,\bbq,\bbp,E)$,
where
$Q$ is a finite set of \emph{states},
$P$ is a finite set of \emph{places},
$\bbq$ is a finite set of \emph{balloon states},
$\bbp$ is a finite set of \emph{balloon places},
and
$E$ is a finite set of edges of the form $q\autstep[\op]q'$,
  where $\op$ is one of a finite set $\Op$ of operations of the
  following form:
  \begin{enumerate}
  \item $\op=\delta$ with $\delta\in\integs^P$,
  \item $\op=\newb(\sigma,S)$, where $\sigma\in\bbq$ and
    $S\subseteq\nats^{\bbp}$ is a semi-linear subset of
    $\nats^{\bbp}$.
  \item $\op=\deflateb(\sigma,\sigma', \pi,p)$, where
    $\sigma,\sigma'\in\bbq$, $\pi\in\bbp$, $p\in P$.
  \item $\op=\burstb(\sigma)$, where $\sigma\in\bbq$.
  \end{enumerate}

A \emph{configuration} of a $\VASSB$ is an element of
$Q\times \multiset{P}\times\multiset{\bbq\times \multiset{\bbp}}$.
That is, a configuration $c = (q, \mmap, \nmap)$ consists of a state $q\in Q$, a multiset $\mmap\in \multiset{P}$, 
and a multiset $\nmap \in \multiset{\bbq \times \multiset{\bbp}}$ of \emph{balloons}.
We assume $\nmap$ has finite support.
A \emph{semiconfiguration} is a configuration $(q,\mmap, \emptyset)$.
For semiconfigurations, we simply write $(q,\mmap)\in Q\times\multiset{P}$. 
For a configuration $c$, 
we write $c.q$, $c.\mmap$, and $c.\nmap$ to denote the components of $c$. 
For a balloon $b\in \bbq\times \multiset{\bbp}$,
we write $b.\barq$ and $b.\kmap$ to indicate its balloon state and
contents respectively and 
write $c.\nmap(b)$ for the number of balloons $b$ in $c$.

The edges in $E$ define a transition relation on configurations.
For an edge $q\autstep[\op]q'$, and configurations 
$c=(q,\mmap,\nmap)$ and $c'=(q',\mmap',\nmap')$, we define
$c \autstep[\op] c'$ iff one of the following is true:
\begin{enumerate}
\item If $\op=\delta\in\integs^P$ and 
$\mmap'=\mmap+\delta$ and $\nmap'=\nmap$.
  
\item If $\op=\newb(\sigma,S)$ and
  $\mmap'=\mmap$ and $\nmap'=\nmap+\multi{(\sigma,\kmap)}$ for some $\kmap\in S$.  
  That is, we create a new balloon with state
  $\sigma$ and multiset $\kmap$ for some $\kmap\in S$.
  
\item If $\op=\deflateb(\sigma,\sigma',\pi,p)$ and 
  there is a balloon $b=(\sigma,\kmap)\in \bbq\times\multiset{\bbp}$ with $\nmap(b)\ge 1$ and
  $\mmap'=\mmap+\kmap(\pi)\cdot\multi{p}$ and
  $\nmap'=(\nmap-\multi{b})+\multi{(\sigma',\kmap')}$,
  where $\kmap'(\pi) = 0$ and $\kmap'(\pi') =\kmap(\pi')$ for all $\pi'\in\bbp\setminus\set{\pi}$.
  That is, we pick a balloon $(\sigma,\kmap)$ from 
  $\nmap$, transfer the contents in place $\pi$ from $\kmap$
  to place $p$ in $\mmap$, and update the balloon state $\sigma$ to $\sigma'$.
  Here we say the balloon $(\sigma,\kmap)$ was \emph{deflated}.
  
\item If $\op=\burstb(\sigma)$ and
  there is a balloon $b=(\sigma,\kmap)\in \bbq\times\multiset{\bbp}$ with $\nmap(b) \geq 1$ and
  $\mmap'=\mmap$ and $\nmap' = \nmap - \multi{b}$.
  This means we pick some balloon $b$ with state $\sigma$ from our multiset $\nmap$ of balloons and remove it, 
  making any tokens still contained in its balloon places disappear as well.
  Here we say the balloon $b$ is \emph{burst}.
\end{enumerate}
The edge set $E$ is the disjoint union of the sets $E_p,E_n,E_d,E_b$
which stand for the edges with operations from (1),(2),(3),(4)
respectively.
We write $c\autstep c'$ if $c\autstep[\op]c'$ for some edge $q\autstep[\op]q'$ in $E$.
A \emph{run} $\rho = c_0 \autstep[\op_0] c_1 \autstep[\op_1] c_2 \autstep[\op_2] \cdots$ is a
finite or infinite sequence of configurations.
The size of $\mathcal{V}=(Q,P,\bbq,\bbp,E)$ is given by
$|\mathcal{V}|=|Q|+|P|+|\bbq|+|\bbp|+|E|$. 

An infinite run $\rho$ is \emph{progressive} iff the following holds:
\begin{enumerate}
  \item For every configuration $c_i=(q_i,\mmap_i,\nmap_i)$ and every balloon
    $b =(\sigma,\kmap) \in \bbq\times\multiset{\bbp}$ with $\nmap_i(b) \geq 1$ there
    is a $c_j$, $j > i$, such that $\op_j$ either bursts or deflates
    $b$.
  \item Moreover, for every configuration $c_i=(q_i,\mmap_i,\nmap_i)$ and every place
    $p \in P$ with $\mmap_i(p) \geq 1$ there is a $c_j$, $j > i$, such that a token is removed from $p$; that is,
    $\op_j \equiv \delta$ with $\delta(p) < 0$.
\end{enumerate}
We define the balloon-norm of a configuration $c =(q,\mmap,\nmap)$ as $\norm{c} = \max\set{ \sum_{p\in\bbp} \kmap(p)   \mid \exists \barq\; \nmap(\barq,\kmap) > 0}$.
A progressive run is \emph{shallow} if there is a number $B\in\nats$
such that
$\norm{c_j} \leq B$ for all $j\geq 0$. In other words, shallowness of
a run means
that each balloon in every configuration on the run contains at most $B$ tokens in the balloon places.
Note that this does not mean the size of the configurations become bounded:
the number of balloons and the number of tokens in $P$ can still be
unbounded. 

The \emph{progressive run problem for $\VASSB$} is the following:
\begin{description}
\item[Given] A $\VASSB$ $\mathcal{V}=(Q,P,\bbq,\bbp,E)$ and an initial
semiconfiguration $c_0$.
\item[Question] Does $\mathcal{V}$ have an infinite progressive run starting from $c_0$?
\end{description}

In \cref{sec:configReach}, we shall prove the following theorems.

\begin{theorem}
\label{th:vassb-fair-term}
The progressive run problem for $\VASSB$ is decidable.
\end{theorem}

The following is a by-product of the proof of
\cref{th:vassb-fair-term}, which will be used in
\cref{starvation}.
\begin{theorem}
\label{th:vassb-bounded-balloons}
A $\VASSB$ $\mathcal{V}$ has a progressive run iff it has a shallow progressive
run.
\end{theorem}

\def \cV{\mathcal{V}}
\def \cT{\mathcal{T}}
\def \SLA{\mathsf{SL}(\cA)}
\def \Ginf{G_{\mathsf{inf}}}

\subsection{From $\DCPS$ to $\VASSB$}\label{sec:dcps-to-vassb}

Instead of reducing fairness for $\DCPS$ to $\VASSB$, we would like to use a stronger notion, which simplifies many of our proofs. To this end, we introduce the notion of progressiveness that we already defined for $\VASSB$ now for $\DCPS$ as well: given a bound $K \in \nats$, an infinite run $\rho$ of a $\DCPS$ is called \emph{progressive} if the rule {\sc Term} is only ever applied when the active thread is at $K$ context switches, and for every local configuration $(w,i)$ of an inactive thread in a configuration of $\rho$, there is a future point in $\rho$ where the rule {\sc Resume} is applied to $(w,i)$, making it the local configuration of the active thread.

Intuitively, no type of thread can stay around infinitely long in a progressive run without being resumed, and every thread that terminates does so after exactly $K$ context switches. Note that progressiveness is a stronger condition than fairness, because it does not allow threads to ``get stuck'' or go above the context switch bound $K$. However, we can always transform a $\DCPS$ where we want to consider fair runs into one where we can consider progressive runs instead. 
The transformation is formalized by the following lemma.

\begin{lemma} \label{lem:progDCPS}
  Given a bound $K \in \nats$ and a $\DCPS$ $\cA$, we can construct a $\DCPS$ $\widetilde{\cA}$ such that:
  \begin{itemize}
    \item $\cA$ has an infinite fair $K$-context switch bounded run iff $\widetilde{\cA}$ has an infinite progressive $K$-context switch bounded run.
    \item $\cA$ has an infinite fair $K$-context switch bounded run that starves a thread iff $\widetilde{\cA}$ has an infinite progressive $K$-context switch bounded run that starves a thread.
  \end{itemize}
\end{lemma}

\subsubsection*{Idea}
To prove \cref{lem:progDCPS} we modify the $\DCPS$ $\cA$ by giving every thread a bottom of stack symbol $\bot$ and saving its context switch number 
in its top of stack symbol. 
We also save this number in the global state whenever a thread is active. 
This way we can still swap a thread out and back in again once it has emptied its stack, 
and we also can keep track of how often we need to repeat that, before we reach $K$ context switches and allow it to terminate.

Furthermore, we also keep a subset $G'$ of the global states of $\cA$ in our new global states, which 
restricts the states that can appear when no thread is active. 
This way we can guess that a thread will be ``stuck'' in the future, upon which we terminate it instead 
(going up to $K$ context switches first) and also spawn a new thread keeping track of its top of stack symbol in the bag. 
Then later we restrict the subset $G'$ to only those global states that do not have {\sc Resume} rules for 
the top of stack symbols we saved in the bag. This then verifies our guess of ``being stuck''.
The second part of the lemma is used in \cref{starvation}, where we reason about starvation. 

We now state the main reduction to $\VASSB$.

\begin{theorem} \label{th:dcps-to-vassb}
  Given a bound $K \in \nats$ and a $\DCPS$ $\cA$ we can construct a $\VASSB$ $\cV$ with a state $q_0$ such that $\cA$ has an infinite progressive $K$-context switch bounded run iff $\cV$ has an infinite progressive run from $(q_0,\emptyset,\emptyset)$.
\end{theorem}

\subsubsection*{Idea}
One of the main insights regarding the behavior of $\DCPS$ is that the order of the spawns of a thread during one round of being active does not matter. None of the spawned threads during one such segment can influence the run until the active thread changes. Thus we only need to look at the semi-linear Parikh image of the language of spawns for each segment. One can then identify a thread with the state changes and spawns it makes during segments $0$ through $K$.
The state changes can be stored in a balloon state and the spawns for each segment in balloon places that correspond to $K+1$ copies of the stack alphabet. 
The inflate operation then basically guesses the exact multiset of spawns of the corresponding thread.

Representing threads by balloons in this way does not keep track of stack contents, which was important for ensuring the progressiveness of a $\DCPS$ run. 
However, starting from a progressive $\DCPS$ run we can always construct a progressive run for the $\VASSB$ by always continuing with the oldest thread 
in a configuration if given multiple choices, and then building the balloons accordingly. 
Always picking the oldest balloon also works for the reverse direction. 

Now let us argue about the construction in more detail.
Given a $\DCPS$ with stack alphabet $\Gamma$ and a context switch bound $K$, construct a $\VASSB$ whose configurations mirror the ones of the $\DCPS$ in the following way. 
The set of places is $\Gamma$ and it is used to capture threads that have not been scheduled yet and therefore only carry a single stack symbol. Formally each thread with context switch number $0$ and stack content $\gamma \in \Gamma$ is represented by a token on place $\gamma$. The set of balloon places is $\Gamma \times \{0,\ldots,K\}$ and they are supposed to carry the future spawns of any given thread during segments $0$ to $K$. Every thread $t$ with context switch number $\geq 1$ is then represented by a balloon where the number of tokens on balloon place $(\gamma, i)$ is equal to the number of threads with stack content $\gamma$ that $t$ will spawn during its $i$th segment. The spawns for segment $0$ are transferred to the place set $\Gamma$ immediately after such a balloon is created, since the represented thread is now supposed to have made its first context switch. Furthermore, each balloon state consists of the context switch sequence and context switch number of its corresponding thread $t$.
The set of states of the $\VASSB$ mirrors the global states of the $\DCPS$.

For this idea to work, we need to compute the semi-linear set of spawns that each type of thread can make, so that we can correctly inflate the corresponding balloon using this set. Here, the \emph{type} of a terminating thread consists of the stack symbol $\gamma$ it spawns with, the global state $g$ in which it first becomes active, the sequence of context switches it makes, and the state in which it terminates. Given a $\DCPS$ $\cA = (G,\Gamma,\Delta,g_0,\gamma_0)$ and a context switch bound $K$, the formal definition of the set of thread types is
\[
  \cT(\cA,K) \coloneq G \times \Gamma \times (G \times \Gamma \times G)^K \times G.
\]
Since we want to decide existence of an infinite progressive run of $\cA$, we can restrict ourselves to threads that make exactly $K$ context switches. 
Now let $t = (g_0', \gamma_0', (g_1, \gamma_1, g_1') \ldots (g_K, \gamma_K, g_K'), g_{K+1}) \in \cT(\cA, K)$ 
be a thread type. 
We want to use $\cP_{(g_0',\gamma_0')}$, the PDA of a thread of this type, to accept the language of spawns such a thread can make. However, we have two requirements on this language, that the PDA does not yet fulfill. 
Firstly, in the spirit of progressiveness, we only want to consider threads that reach the empty stack and terminate. Secondly, we want the spawns during each segment of the thread execution to be viewed separately from one another.

For the first requirement, we modify the transition relation of $\cP_{(g_0',\gamma_0')}$, such that transitions of the form $(g_2,\bot) \xrightarrow{(g_2,\bot)|\gamma_2/\gamma_2} \mathsf{end}$ defined in (\ref{tPDA:tomod}) are only kept in the relation for $\gamma_2 = \bot$. This ensures that the PDA no longer considers thread executions that do not reach the empty stack.

Regarding the second requirement, we can simply introduce $K + 1$ copies of $\Gamma$ to the input alphabet of $\cP_{(g_0',\gamma_0')}$. It is then redefined as $\Sigma =\Gamma \times \{0,\ldots,K\} \cup G\times \Gamma \times G \cup G \times \set{\bot}$, while the stack alphabet and states stay the same. Any transition previously defined on input $\gamma \in \Gamma$ is now copied for inputs $(\gamma,0)$ to $(\gamma,k)$.

Let $\widetilde{\cP}_{(g_0',\gamma_0')}$ be the PDA these changes result in. Then the context-free language that characterizes the possible spawns of a thread of type $t$ is given by the following:
\[
  L_t \coloneq L(\widetilde{\cP}_{(g_0',\gamma_0')}) \cap (\Gamma \times 0)^* \cdot (g_1,\gamma_1,g_1') \cdot (\Gamma \times 1)^* \cdots (g_K,\gamma_K,g_K') \cdot (\Gamma \times K)^* \cdot (g_{K+1},\bot)
\]
Here we intersect the language of the PDA with a regular language, which forces it to adhere to the type $t$ and groups the spawns correctly. 
If the language $L_t$ is nonempty, using Parikh's theorem (\cref{th:parikh}), we can compute the semi-linear set characterizing the Parikh-image of 
this language projected to $\Gamma \times \{0,\ldots,K\}$, which we denote $\mathsf{sl}(t)$. 
We also define the set of all semi-linear sets that arise in this way as $\mathsf{SL}(\cA,K) \coloneq \{\mathsf{sl}(t) \mid t \in \cT(\cA,K), L_t \neq \emptyset \}$.

Now we can construct a $\VASSB$ whose configurations correspond to the ones of the $\DCPS$ in the way we mentioned earlier. From $(q_0,\emptyset,\emptyset)$ we put a token on $\gamma_0$ to simulate spawning the initial thread and thus begin the simulation of the $\DCPS$. We can then construct a progressive run of the $\VASSB$ from a progressive run of the $\DCPS$ by constructing the individual balloons as if the scheduler always picked the oldest thread out of all choices with the same local configuration. The converse direction works for similar reasons by always picking the oldest balloon to continue with.

The proof also allows us to reason about \emph{shallow} progressive runs of $\DCPS$. 
Following the same notion for $\VASSB$, we call a run of a $\DCPS$ \emph{shallow}
if there is a bound $B \in \nats$ such that each thread on that run spawns at most $B$ threads.
Obverse that in the $\VASSB$ construction of this section the spawns of $\DCPS$ threads are mapped to the contents of balloons, which is how the two notions of shallowness correspond to each other. Thus we can go from progressive $\DCPS$-run to progressive $\VASSB$-run by \cref{th:dcps-to-vassb}, to shallow progressive $\VASSB$-run by \cref{th:vassb-bounded-balloons}, to shallow progressive $\DCPS$ run by \cref{th:dcps-to-vassb} combined with the observation on the two notions of shallowness. This is formalized in the following:

\begin{corollary} \label{dcps-spawn-bounded}
  A $\DCPS$ $\cA$ has a progressive run iff it has a shallow progressive run.
\end{corollary}

\newcommand{\pc}{\mathsf{p}}

\newcommand{\tildeP}{\tilde{P}}
\def \tildebq{\tilde{Q}}
\def \cI{\mathcal{I}}
\def \dbl{\mathsf{dbl}}
\def \tildebp{\tilde{P}}

\newcommand{\runs}{\mathsf{runs}}

\newcommand{\id}{\mathsf{id}}
\newcommand{\seq}{\mathsf{seq}}
\newcommand{\Seq}{\mathsf{SEQ}}
\newcommand{\last}{\mathsf{last}}
\newcommand{\emp}{\mathsf{emp}}
\newcommand{\first}{\mathsf{first}}
\newcommand{\second}{\mathsf{second}}
\newcommand{\finc}{\mathsf{finc}}
\newcommand{\sinc}{\mathsf{sinc}}
\newcommand{\stage}{\mathsf{stage}}

\newcommand{\fdec}{\mathsf{fdec}}
\newcommand{\sdec}{\mathsf{sdec}}
\newcommand{\fshift}{\mathsf{fshift}}
\newcommand{\sshift}{\mathsf{sshift}}
\newcommand{\checked}{\mathsf{checked}}
\newcommand{\unchecked}{\mathsf{unchecked}}
\newcommand{\verify}{\mathsf{verify}}

\section{From Progressive Runs for $\VASSB$ to Reachability}
\label{sec:configReach}

In this section, we prove Theorems~\ref{th:vassb-fair-term} and \ref{th:vassb-bounded-balloons}.
We outline the main ideas and technical lemmas
used to obtain the proofs. The formal proofs can be found in the full
version of the paper.

We first establish that finite witnesses exist for infinite progressive runs. As a byproduct, this yields \cref{th:vassb-bounded-balloons}.
Then, we show that finding finite witnesses for progressive runs reduces to reachability in $\VASSB$.
Finally, we prove that reachability is decidable
for $\VASSB$.

\subsection{From Progressive Runs to Shallow Progressive Runs}

Fix a $\VASSB$ $\cV= (Q,P,\bbq,\bbp,E)$.
A \emph{pseudoconfiguration} 
$\pc(c)=(q,\mmap,\partial\nmap) \in Q \times \multiset{P} \times \multiset{\bbq}$ 
of a configuration $c =(q,\mmap,\nmap)$ is given by
$\partial\nmap(\barq) = \sum_{\kmap\in\multiset{\bbp}} \nmap(\barq, \kmap)$.
That is, a pseudoconfiguration is obtained by counting the number of balloons in
a given state $\barq\in\bbq$ but ignoring the contents. The
\emph{support} $\sup(\mmap)$ of a
multiset $\mmap$ is the set of places $\set{p \mid \mmap(p) >0}$
where $\mmap$ takes non-zero values.

For configurations $c=(q,\mmap,\nmap)$ and $c'=(q',\mmap',\nmap')$,
we write $c\leq c'$ if $q = q'$, $\mmap \preceq \mmap'$, and $\nmap \preceq \nmap'$. Moreover, we write
$\pc(c) \leq_{\pc} \pc(c')$ if $q=q'$, $\mmap\preceq \mmap'$, and
$\partial\nmap\preceq \partial\nmap'$.  
Both $\leq$ and $\leq_{\pc}$ are \emph{well-quasi orders} (wqo), that is, any infinite
sequence of configurations (resp.\ pseudoconfigurations) has an infinite increasing subsequence
with respect to $\leq$ (resp.\ $\leq_{\pc}$) (see, e.g., \cite{ACJT,FinkelSchnoebelen} for details).
This follows
because wqos are closed under multiset embeddings.
Moreover, $\cV$ is \emph{monotonic} w.r.t.\ $\leq$: if $c_1 \rightarrow c_1'$ and $c_1\leq c_2$,
then there is some $c_2'$ such that $c_1' \leq c_2'$ and $c_2
\rightarrow c_2'$. 

In analogy with algorithms for finding infinite runs in VASS (in
particular the procedure for checking fair termination in the case
$K=0$ in \cite{GantyM12}), one might try to find a self-covering run
w.r.t.\ the ordering $\leq$. However, checking for such a run would
require comparing an unbounded collection of pairs of
balloons. In order to overcome this issue, we use a
\emph{token-shifting
surgery} which moves
tokens from one balloon to another. The surgery is performed on the
given progressive run
$\rho$, converting it into a progressive run $\rho'$ with a special
property:
there exist
infinitely many configurations in $\rho'$ which contain only empty
balloons.
By restricting ourselves to such configurations, we are able to show
the existence of a special kind of self-covering run where the
cover and the original configuration only contain empty balloons and
thus it suffices to compare them using the ordering $\leq_{\pc}$. 
First, we need the following notion of a witness for progressiveness.

For $A \subseteq P, B \subseteq \bbq$, a run $\rho_{A,B}=c_0 \autstep
[*] c
\autstep
[*] c'$ is
called an \emph{$A,B$-witness} for progressiveness if it satisfies the
following
properties:
\begin{enumerate}
	\item For any $c''$ occurring between $c$ and $c'$, we have $\sup
	(c''.\mmap) \subseteq
	A$,
	\item for each $p \in A$, there exists $\op$ between $c$ and $c'$ in
	$\rho_{A,B}$
	such that $\op=\delta$ where $\delta(p)<0$,
	\item for any balloon $b$, we have $c.\nmap(b) \ge 1$ iff $c'.\nmap
	(b)
	\ge 1$ iff
	($b.\kmap=\emptyset $ and $ b.\barq \in B$),
	\item for any $\barq \in B$, there exists $\op$ occurring between $c$
	and $c'$ 
	such that $\op=\deflateb(\barq,\cdot,\cdot,\cdot)$ or $\op=\burstb
	(\barq)$ is applied to an empty balloon with state $\barq$, and
		\item $\pc(c) \leq_{\pc} \pc(c')$ and $\sup(c.\mmap)=\sup(c'.\mmap)$.
\end{enumerate}

In order to formalize the idea of a token-shifting surgery, we
associate a unique identity with each balloon. In particular, we may
associate the
unique number $i \in \nats$ with the balloon $b$ which is inflated by
the $i^{th}$ operation $\op_i$ in a run $\rho=c_0 \autstep[\op_1] c_1
\autstep[\op_2] c_2 \cdots$,
giving us a
\emph{balloon-with-id} $(b,i)$. The id of a balloon is preserved on
application of deflate and burst operations and thus we can speak of
\emph{the balloon $i$} and the sequence of operations $\seq_i$
that it undergoes. Given a run $\rho$, we produce a corresponding
\emph{canonical run-with-id} $\tau$ inductively as follows: the
balloon identities are assigned
as above and the balloon with least id is chosen for execution every
time. 
Extending the notion of a
balloon-with-id to a configuration-with-id $d$ (where every balloon
has an
id)
and a run-with-id $\tau$ (which consist of sequences of
configurations-with-id), we define the notion of a \emph{progressive}
run-with-id
$\tau$,
which is
one such that the sequence $\seq_i$ associated with any id $i$ in
$\tau$ is
either infinite, or $\seq_i$  is finite with the last operation being
a
burst. If every id $i$ undergoes a burst operation in a run-with-id
$\tau$, we say ``$\tau$ bursts every balloon.'' We abuse terminology by
saying ``$\rho$ bursts every balloon'' for a run $\rho$ to mean that
there is a \emph{corresponding} run-with-id $\tau$ which bursts every
balloon. It is easy to see that the canonical run-with-id $\tau$
associated
with a progressive run $\rho$ is ``almost''
progressive. By always picking the id which has been idle for the
longest time, we can convert
$\tau$ into a progressive run-with-id. Thus there exists a progressive
run if and only if there exists a progressive
run-with-id, but the latter retains more information,
allowing us to argue formally in proofs. 
%
With these notions in
hand, we prove the following lemma:

\begin{lemma}
\label{new-witness-prog}
Given a $\VASSB$ $\cV$ and a semiconfiguration $s$ of $\cV$, one can
construct a $\VASSB$ $\cV'=(Q',P',\bbq',\bbp',E')$ and a
semiconfiguration $s'$ of $\cV'$ such that
the following are equivalent:
\begin{enumerate}
	\item $\cV$ has a progressive run from $s$,
\item $\cV'$ has a progressive run from 
$s'$ that bursts every balloon, and
\item  $\cV'$ has an $A,B$-witness for some $A
	\subseteq P'$ and $B \subseteq \bbq'$.
\end{enumerate}
\end{lemma}
The remainder of this subsection is devoted to the proof of \cref{new-witness-prog}.
The lemma is proved in two steps: first we show (1)~$\iff$~(2), then
we show (2)~$\iff$~(3). We do some preprocessing before (1)~$\iff$~(2),
 by showing that one can convert $\cV$ into a
$\VASSB$ $\cV'$ with two special properties:
 (i) the 
\emph{zero-base}
property, by
which
every linear set of $\cV'$ has base vector equal to $\bfz$, and 
(ii) the
property of being \emph{typed}, which means that we guess and verify
the sequence of deflates performed by a balloon $b$ that could
potentially transfer a non-zero
number of tokens by including this information in its balloon state
$b.\barq$. A deflate operation which transfers a non-zero number of
tokens is called a \emph{non-trivial} deflate. The \emph{type} $t=
(L,S)$ of a
balloon $i$
consists of the linear set $L$ used during its inflation, along with
the sequence $S=(\barp_1,p_1),(\barp_2,p_2),\cdots,(\barp_n,p_n)$ of
deflate operations on $i$, such that for each $j \in 
\set{1,\cdots,n}$, the first deflate operation acting on the balloon
place $\barp_j$ sends tokens to the place $p_j$.
\begin{lemma}
\label{lem:zerotype}
Given a $\VASSB$ $\cV$ along with its semiconfigurations $s_0,s_1$, we
can construct a zero-base, typed $\VASSB$ $\cV'$ and its
semiconfigurations 
$s'_0,s'_1$ such that:
\begin{enumerate}
	\item There is a progressive run of $\cV$ from $s_0$ iff there is a
	progressive run of $\cV'$ from $s'_0$.
	\item There is a run $s_0 \xrightarrow{*} s_1$ in $\cV$ iff there is
	a run $s'_0 \xrightarrow{*} s'_1$ in $\cV'$.
\end{enumerate}
\end{lemma}
The zero-base property is easily obtained by making sure that the
portion of tokens transferred which correspond to the base vector are
separately transferred using $E_p$-edges. The addition of the types
into the global state can be done by expanding the set of balloon
states exponentially. A proof of the lemma is given in the full
version. 

We return to the proof of \cref{new-witness-prog}.
The direction (1)~$\Rightarrow$~(2) requires us to show that $\cV'$
can be
assumed to ``burst every balloon'' in a progressive run-with-id. Consider a balloon $i$
occurring in an arbitrary progressive
run-with-id $\tau''$. In order to convert the given $\tau''$ into a
progressive $\tau'$ in which every balloon is burst, we need to burst
those id's $i$ such that $\seq_i$ is infinite in $\tau''$. Since only
a
finite number of non-trivial deflates can be performed by a given
balloon $i$, this implies that $\seq_i$ consists of a finite prefix in
which non-trivial deflates are performed, followed by an infinite
suffix of trivial deflates. Every such balloon $i$ can then be
burst and replaced by a special $\VASS$	token.
The
infinite trivial suffix is then simulated by using addition and
subtraction operations of
these special tokens using additional places, since any such
balloon $i$ will not transfer any more tokens. The converse
direction (1)~$\Leftarrow$~(2) is a reversal of the construction where
we replace
the
special $\VASS$ tokens with infinite sequences of trivial deflate
operations.

\definecolor{c00}{RGB}{255, 255, 255} %
\definecolor{c01}{RGB}{204,204,0} %
\definecolor{c02}{RGB}{0,153,153} %
\definecolor{c03}{RGB}{225,128,0} %
\definecolor{c04}{RGB}{204,0,102} %
\definecolor{c05}{RGB}{0,102,0} %

\def \n {0.8} %
\def \yn {4} %

\def \l {0.04} %
\def \tk {0.1} %

\def \rh {1} %
\def \rw {0.3} %

\def \eh {0.2} %
\def \ew {0.1} %

\newcommand{\num}[2]{
\coordinate (n) at (#1*\n, 0) {};
\draw ($(n)+(0,\tk cm)$) -- ($(n)-(0,\tk cm)$);
\node[label=below:{\tiny #2}] at (n) {};
}

\newcommand{\narow}[2]{
	\coordinate (n) at (#1*\n,0) {};
	\coordinate (m) at (#2*\n,-\yn cm){};
	\draw[dashed, line width=\l cm, black!40, arrows=->] (n) -- (m);
}

\newcommand{\nvar}[2]{
\coordinate (n) at (#1*\n, 0) {};
\draw ($(n)+(0,\tk cm)$) -- ($(n)-(0,\tk cm)$);
\node[label=above:{\tiny #2}] at (n) {};
}

\newcommand{\carowa}[3]{
	\coordinate (n) at (#1*\n,0) {};
	\coordinate (m) at (#2*\n,0) {};
	\path (n)  edge[bend right] node[below] {\scriptsize#3} (m);
}

\newcommand{\carowb}[3]{
	\coordinate (n) at (#1*\n,0) {};
	\coordinate (m) at (#2*\n,0) {};
	\path (n)  edge[bend right=60] node[below] {\scriptsize#3} (m);
}

\newcommand{\numline}[1]{

	\coordinate (n0) at (0*\n, 0) {};
	\coordinate (nlast) at (20*\n, 0) {};
	\node[label=right:{\tiny #1}] at (nlast) {};

	\draw[line width=\l cm, black] (n0) -- (nlast);

	\num{0}{};
	\num{1}{};
	\num{2}{};
	\num{3}{};
	\num{4}{};
	\num{5}{};
	\num{6}{};
	\num{7}{};
	\num{8}{};
	\num{9}{};
	\num{10}{};
	\num{11}{};
	\num{12}{};
	\num{13}{};
	\num{14}{};
	\num{15}{};
	\num{16}{};
	\num{17}{};
	\num{18}{};
	\num{19}{};
	\num{20}{};

}

\newcommand{\dballoonInitl}[3]{  
\begin{scope}
    \coordinate (v1) at (#1*\n,1.2) {};
    \draw[#2,thick] (v1) ellipse (0.12 and 0.22);
    \tikzset{decoration={snake,amplitude=0.2mm,segment length=1mm, post length=0mm, pre length=0mm}}
\draw[line width=0.1mm, decorate] ($(v1)+(0,-0.22)$) -- ($(v1)+
(0,-0.42)$);

    \clip (v1) ellipse (0.1 and 0.2);
    \coordinate (v2) at ($(v1)+(0,#3)$) {};

    \fill[black!30] ($(v2)+(-0.2,0)$)--($(v2)+(0.2,0)$)--($(v2)+
    (0.2,-0.4)$)--
    ($(v2)+(-0.2,-0.4)$)--cycle;
 \end{scope}
}

\begin{figure}[t!]
\label{token-shifting-tikz}

\begin{tikzpicture}[scale=0.8]

\numline{$\tau$}
\nvar{2}{$i_0$}
\dballoonInitl{2}{red}{0.2}

\nvar{4}{$i_1$}
\dballoonInitl{4}{red}{0.1}
\carowa{2}{4}{$(\pi_1,p_1)$}

\nvar{7}{$i_2$}
\dballoonInitl{7}{red}{0}
\carowa{4}{7}{$(\pi_2,p_2)$}

\nvar{10}{$i_3$}
\dballoonInitl{10}{red}{-0.2}
\carowa{7}{10}{$(\pi_3,p_3)$}

\nvar{1}{$j_0$}
\dballoonInitl{1}{cyan}{0.2}

\nvar{6}{$j_1$}
\dballoonInitl{6}{cyan}{0}
\carowb{1}{6}{$(\pi_1,p_1)$}

\nvar{11}{$j_2$}
\dballoonInitl{11}{cyan}{-0.1}
\carowb{6}{11}{$(\pi_2,p_2)$}

\nvar{14}{$j'_2$}
\dballoonInitl{14}{cyan}{-0.1}
\carowb{11}{14}{$(\pi_2,p'_2)$}

\nvar{16}{$j_3$}
\dballoonInitl{16}{cyan}{-0.2}
\carowb{14}{16}{$(\pi_3,p_3)$}

\begin{scope}[yshift=-\yn cm]
\numline{$\tau'$}
\nvar{2}{$i_0$}
\dballoonInitl{2}{red}{0.2}

\nvar{4}{$i_1$}
\dballoonInitl{4}{red}{0.1}
\carowa{2}{4}{$(\pi_1,p_1)$}

\nvar{7}{$i_2$}
\dballoonInitl{7}{red}{0}
\carowa{4}{7}{$(\pi_2,p_2)$}

\nvar{10}{$i_3$}
\dballoonInitl{10}{red}{-0.2}
\carowa{7}{10}{$(\pi_3,p_3)$}

\nvar{1}{$j_0$}
\dballoonInitl{1}{cyan}{-0.2}

\nvar{6}{$j_1$}
\dballoonInitl{6}{cyan}{-0.2}
\carowb{1}{6}{$(\pi_1,p_1)$}

\nvar{11}{$j_2$}
\dballoonInitl{11}{cyan}{-0.2}
\carowb{6}{11}{$(\pi_2,p_2)$}

\nvar{14}{$j'_2$}
\dballoonInitl{14}{cyan}{-0.2}
\carowb{11}{14}{$(\pi_2,p'_2)$}

\nvar{16}{$j_3$}
\dballoonInitl{16}{cyan}{-0.2}
\carowb{14}{16}{$(\pi_3,p_3)$}

\end{scope}

\end{tikzpicture}
\caption{Top: Initial run $\tau$ with two non-empty balloons which
perform the same sequence of three non-trivial deflates $(\pi_1,p_1),
(\pi_2,p_2),(\pi_3,p_3)$. Bottom: Modified run $\tau'$ after shifting
tokens
from the cyan balloon inflated at $j_0$ to the red balloon inflated at
$i_0$.}
\label{token-shift-fig}
\end{figure}
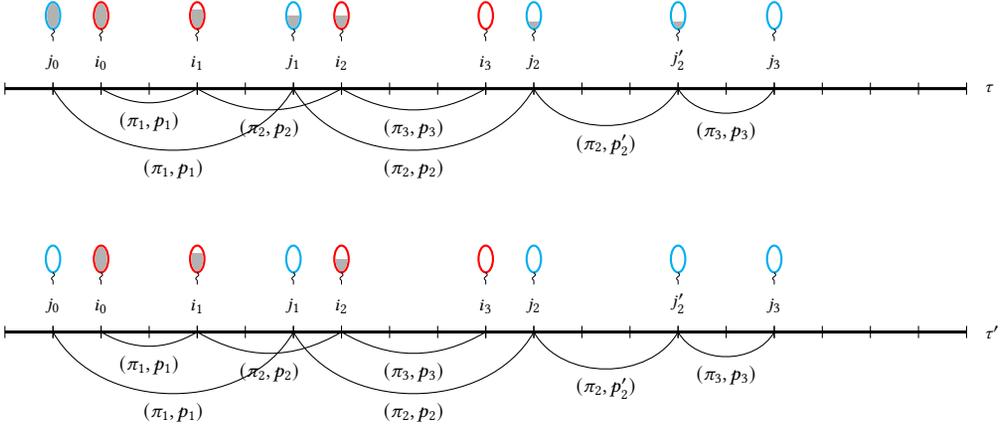

\paragraph{Token-Shifting}
We move on to show (2)~$\iff$~(3).
The key idea is a token-shifting surgery.
A \emph{token-shifting surgery} creates a run $\tau'$ from a run
$\tau$ as depicted in Figure \ref{token-shift-fig}. 
We start with the run-with-id $\tau=d_0 \autstep[\op_1] d_1
\autstep[\op_2] d_2 \cdots$ in which two balloons have the
same type $t=(L,S)$ with $S=(\pi_1,p_1), (\pi_2,p_2),(\pi_3,p_3)$
being the sequence of (potentially) non-trivial
deflates. Recall that an index $i$ relates to the
operation $\op_i$ in $\tau$. The region of a balloon which is shaded
grey
visualizes
the total number of tokens contained in the balloon. This number is seen to decrease
after each non-trivial deflate operation. An empty balloon is white
in color. Balloons having the
same identity have the same outline color: red for the balloon $b_1$
inflated
at $i_0$ and cyan for the balloon $b_2$ inflated at $j_0$. At $i_3$ (resp.
$j_3$) all tokens have been transferred and the red balloon
(resp.
cyan balloon) is empty. The
crucial property
satisfied is that $i_k < j_k$ for $k \in \set{1,2,3}$, i.e., every
deflate from $S$ of the red balloon occurs before the corresponding
deflate
of the cyan balloon. In the modified run $\tau'$,
we have the inflation of an
empty cyan balloon
and the inflation of the red balloon with the sum of the
tokens of both red and cyan balloons
in $\tau$. We require the zero-base property in order to be
able
to shift the tokens in this manner: observe that a linear set with
zero base vector is closed under addition. Note that the cyan balloon
undergoes a trivial deflate
at $j'_2$ where no tokens are transferred: this deflate is not
part of $S$ and is not relevant for the token-shifting. Thus the
run-with-id $\tau$
may be modified to give a \emph{valid}
run-with-id
$\tau'$. Note that the configurations-with-id $d'_j$ of $\tau'$
satisfy $d'_j \geq d_j$ for each $j$ and so by monotonicity every
operation that is applied at $d_j$ can be applied at $d'_j$.

We now show how token-shifting is applied to prove (2)~$\iff$~(3) in
\cref{new-witness-prog}.
First, from a progressive run $\tau$ of $\cV'$ such that every balloon
is
burst, we produce a run $\tau'$ of $\cV'$ from which it will be easy
to extract an $A,B$-witness. Let $T_
{\infty}$
be the set of types of balloons which occur infinitely often in
$\tau$. Since every balloon is eventually burst in $\tau$, there has to be a configuration $d_0$ such that
after $d_0$, every occurring balloon has a type in $T_\infty$.
We now inductively pick a sequence of
configurations
$d_1,d_2,\cdots$ and a sequence of sets of balloons
$I_1,I_2,\cdots$ with the following properties, for each $k \ge 1$:
\begin{enumerate}
 	\item $I_k$ contains exactly one balloon inflated after $d_{k-1}$
 	for each
 	type in $T_{\infty}$ and
\item every balloon in $I_k$ is burst before $d_k$.
 \end{enumerate}  
 For every
balloon $i$ not in any of the sets $I_1,I_2,\cdots$, which is inflated
between $d_k$ and $d_{k+1}$ for $k \ge 1$, we shift its tokens to
the corresponding balloon $j$ in $I_k$ of the same type as $i$.
Clearly this is
allowed since all of the deflate operations of $j$ occur before $i$ is
inflated. Thus we obtain the run $\tau'=d''_0 \xrightarrow{*} d'_0 
\xrightarrow{*} d'_1 \xrightarrow{*} d'_2 \ldots$ from $\tau$. The
prefix $d''_0 \xrightarrow{*} d'_0$ of
the modified run $\tau'$ may contain balloons of arbitrary
type. Between $d'_0$ and $d'_1$, there are only balloons in $T_\infty$. 
After $d'_2$, we have an infinite suffix where all balloons
are of a
type from $T_{\infty}$ and the only non-empty balloons
are those belonging to $I_k$ for some $k \ge 2$. This
means that the configurations-with-id $d'_k$ for each
$k \ge 2$ only contain empty balloons. Since $\le_{\pc}$
is a well-quasi-ordering, the sequence $d'_2,d'_3,\ldots$
must contain configurations $d'_l$ and $d'_m$ with $d'_l\leq_{\pc}
d'_m$. We obtain an $A,B$-witness as follows. We choose the set $P_
{\infty}$ which is the set of places
which are non-empty infinitely often along $\tau'$ for the set $A$.
Since the set
of possible balloon states and places in a given configuration is
finite, by Pigeonhole Principle, we may assume that $d'_l$ and $d'_m$
have the same set of non-empty places $A \subseteq P'$ and balloon
states $B \subseteq \bbq$. We may also assume that progressiveness
checks (2) and (4) corresponding to $A$ and $B$ occur between $d'_l$
and $d'_m$.

 Conversely, an $(A, B)$-witness $\rho_{A,B}$ can be
  ``unrolled'' to give a
 progressive run $\tau''$ of $\cV'$. Furthermore, since the unrolling
 $\tau''$ only contains
 balloons with contents present in the finite run $\rho_{A,B}$,
 giving us a shallow progressive run as stated in
Theorem~\ref{th:vassb-bounded-balloons}.

\subsection{Reduction to Reachability}

The \emph{reachability problem} $\REACH$ for $\VASSB$ asks: 
\begin{description}
\item[Given] A $\VASSB$ $\cV=(Q,P,\balloonQ,\balloonP,E)$ and
two semiconfigurations $c_0$ and $c$.
\item[Question] Is there a run $c_0 \xrightarrow{*} c$ ?
\end{description}
The more general version of the problem, where $c_0$ and $c$
can be arbitrary configurations (i.e., with balloon contents), easily
reduces to this problem. 
However, the exposition is simpler if we restrict to semiconfigurations here.
In this subsection, we shall reduce the progressive run problem for
$\VASSB$ to the \emph{reachability problem} for $\VASSB$.
In the next subsection, we shall reduce the reachability problem to the reachability problem for $\VASS$,
which is known to be decidable \cite{Kosaraju,Mayr81}. 

\begin{lemma}
\label{lem:fairnterm-to-reach}
	The progressive run problem for $\VASSB$ reduces to the problem $\REACH$
	for $\VASSB$.
\end{lemma}
Fix a $\VASSB$ $\cV$.
Using Lemma~\ref{new-witness-prog}, we look for progressive witnesses.
Let $A \subseteq P$ and $B \subseteq \bbq$.
We shall iterate over the finitely many choices for $T=(A,B)$ and
check that $\cV$ has an
infinite progressive run with a $A,B$-witness 
by reducing to the configuration reachability problem for an associated $\VASSB$ $\cV(T)$.

The $\VASSB$ $\cV(T)$ simulates $\cV$ and guesses the two configurations $c_1$ and $c_2$ such that
$c_0 \xrightarrow{*} c_1 \xrightarrow{*} c_2$ satisfies the conditions for a progressive witness. 
It operates in five total stages, with three main stages and two
auxiliary ones sandwiched between the main stages. 
In the first main stage, it simulates two identical copies of the run
of $\cV$ starting from $c_0$. The global state is shared by the two
copies while we have separate sets of places. We cannot maintain
separate sets of balloons for each copy since the inflate operation is
inherently non-deterministic and hence the balloon contents may be
different in the two balloons produced. The trick to maintaining two
copies of
the same balloon is to in fact only inflate a single instance of a ``doubled'' balloon
which uses ``doubled'' vectors and two copies of balloon places. Deflate
operations are then performed twice on each doubled balloon, moving
tokens to the corresponding copies of places.

The $\VASSB$  $\cV(T)$ also tracks the number of balloons in each balloon
state (independent of their contents), for each copy of the run.
At some point, it guesses that the current configuration is $c_1$ (in both copies) and moves
to the first auxiliary stage. The first auxiliary stage checks
whether all the balloons in $c_1$ are empty. Control is then passed to
the
second main stage.

In the second main stage, the first copy of the run is frozen to
preserve
$\pc(c_1)$ and $\cV(T)$ continues
to simulate $\cV$ on the second copy. This is implemented by producing
only ``single'' balloons during this stage and only deflating the second
copy of the places in the ``double'' balloons which were produced in the
first main stage.
While simulating $\cV$ on the second copy, $\cV(T)$ additionally
checks the progressiveness constraints (2) and (4) corresponding to an
$A,B$-witness in its global state.
The second main stage non-deterministically guesses when the second
copy
reaches $c_2$ (and ensures all progress
constraints have been met) and moves on to the second auxiliary stage.
Here the fact that all the balloons in $c_2$ are empty is checked and
then control passes
to the third main stage.
In the third main stage, $\cV(T)$ verifies that the two configurations
$c_1$ and $c_2$ also satisfy conditions (1), (3), and (5) for an
$A,B$-witness.
A successful verification puts $\cV(T)$ in a specific final
semi-configuration. 

\subsection{From Reachability in $\VASSB$ to Reachability in $\VASS$}
\label{sec:reach2reach}

In this subsection, we show that reachability for $\VASSB$
reduces to reachability in ordinary $\VASS$. 
We write $\exp_k(x)$ for the $k$-fold
exponential function i.e. $\exp_1(x)$ is $2^x$, $\exp_2(x)$ is
$2^{2^x}$ etc. 
%
%

A run $\rho=s_1 \autstep[*] s_2$ of a
$\VASSB$ $\cV$ between two semiconfigurations is said to be 
\emph{$N$-balloon-bounded} for some $N \in \nats$
if there exist at most $N$ non-empty
balloons which are inflated in $\rho$.
The following lemma is the crucial observation
for our reduction.
\begin{lemma}
\label{lem:N-bal-bdd}
  Given any $\VASSB$ $\cV=(Q,P,\balloonQ,\balloonP,E)$, there
  exists $N \in \nats$ with $N \leq O(\exp_4(|\cV|))$ such
  that for
  any two semiconfigurations $s_1,s_2$, if $(\cV,s_1,s_2) \in \REACH$,
  then there exists a run $\rho=s_1 \autstep[*] s_2$ of $\cV$
  that is $N$-balloon-bounded.
\end{lemma}
Before we prove Lemma~\ref{lem:N-bal-bdd}, let us see how it allows us to
reduce reachability in $\VASSB$ to reachability in $\VASS$.
\begin{lemma}
\label{lem:reach2reach}
  The problem $\REACH$ for $\VASSB$ reduces to the problem $\REACH$ for
  $\VASS$.
\end{lemma}
From a given $\VASSB$ $\cV$, we construct a $\VASS$ $\cV'$ which has
extra places $\bbq \times \set{1,\cdots,N}\times \bbp$ for storing the
contents of all the non-empty
balloons, as well as extra places $\bbq$ that store the number of 
balloons which were created empty for each balloon state $\barq$ of
$\cV$. The global state of $\cV'$ is used to keep track of the total
number of
non-empty balloons created as well as their state changes. Deflate and
burst operations are replaced by appropriate token transfers such that
there is only one opportunity for $\cV'$ to transfer tokens of any
non-empty balloon by using the global state. This results in a
 faithful simulation in the forward direction, as well as the
 easy extraction of a run of $\cV$ from a run of
 $\cV'$ in the
 converse direction.

\paragraph{Proof of \cref{lem:N-bal-bdd}}
We now prove \cref{lem:N-bal-bdd},
which will complete the proof of \cref{th:vassb-fair-term}.
 
We observe that if, for every balloon state $\barq$, the number of
balloons that are inflated in $\rho$ with state $\barq$ is bounded by
$N$, this
implies a bound of $|\bbq|N$ on the total number of balloons inflated
in $\rho$. Hence, we can equivalently show the former bound assuming
a particular balloon state.
We assume that $\cV$ is both
zero-base and typed while preserving reachability by \cref{lem:zerotype}. This
implies that
the type information is contained in the state of a balloon.
The lemma is then proved by showing that if more
than $N$ non-empty balloons of a particular state $\barq$ are
inflated
in a run
$\rho=s_1
\autstep[*] s_2$, then
it is possible to perform an \emph{id-switching} surgery, resulting in a run $\rho'=s_1
\autstep[*] s_2$ which creates one less non-empty balloon with state
$\barq$. 

\newcommand{\numlineIds}[1]{

	\coordinate (n0) at (0*\n, 0) {};
	\coordinate (nA) at (2.5*\n, 0) {};
	\coordinate (nB) at (3.5*\n, 0) {};
	\coordinate (nC) at (17.5*\n, 0) {};
	\coordinate (nD) at (18.5*\n, 0) {};
	\coordinate (nlast) at (20*\n, 0) {};
	\node[label=right:{\tiny #1}] at (nlast) {};

	\draw[line width=\l cm, black] (n0) -- (nA);
	\draw[dotted, line width=\l cm, black] (nA) -- (nB);
	\draw[line width=\l cm, black] (nB) -- (nC);
	\draw[dotted, line width=\l cm, black,] (nC) -- (nD);
		\draw[line width=\l cm, black] (nD) -- (nlast);

	\num{0}{};
	\num{1}{};
	\num{2}{};
	\num{4}{};
	\num{5}{};
	\num{6}{};
	\num{7}{};
	\num{8}{};
	\num{9}{};
	\num{10}{};
	\num{11}{};
	\num{12}{};
	\num{13}{};
	\num{14}{};
	\num{15}{};
	\num{16}{};
	\num{17}{};
	\num{19}{};
	\num{20}{};

}

\begin{figure}[t!]

\begin{tikzpicture}[scale=0.8]

\numlineIds{$\tau$}

\nvar{1}{$i_0$}
\dballoonInitl{1}{cyan}{0.2}

\nvar{2}{$j_0$}
\dballoonInitl{2}{orange}{0.2}

\nvar{4}{$i_1$}
\dballoonInitl{4}{cyan}{0.1}

\nvar{5}{$j_1$}
\dballoonInitl{5}{orange}{0.1}

\nvar{7}{$j_2$}
\dballoonInitl{7}{orange}{0}
\node at (7*\n,-1) {\scriptsize$i_0.\barq.=j_0.\barq$};
\draw[->] (7*\n,-0.9) -- (7*\n,-0.2);

\nvar{8}{$i_2$}
\dballoonInitl{8}{cyan}{0}

\nvar{10}{$j_3$}
\dballoonInitl{10}{orange}{-0.07}

\nvar{11}{$i_3$}
\dballoonInitl{11}{cyan}{-0.07}
\node at (11*\n,-1) {\scriptsize$i_0.\barq=j_0.\barq$};
\draw[->] (11*\n,-0.9) -- (11*\n,-0.2);

\nvar{13}{$i_4$}
\dballoonInitl{13}{cyan}{-0.12}

\nvar{14}{$j_4$}
\dballoonInitl{14}{orange}{-0.12}

\begin{scope}[yshift=-0.7*\yn cm]
\numlineIds{}

\nvar{1}{$i_0$}
\dballoonInitl{1}{cyan}{0.2}

\nvar{2}{$j_0$}
\dballoonInitl{2}{orange}{0.2}

\nvar{4}{$i_1$}
\dballoonInitl{4}{cyan}{0.1}

\nvar{5}{$j_1$}
\dballoonInitl{5}{orange}{0.1}

\nvar{7}{$i_2$}
\dballoonInitl{7}{cyan}{0}

\nvar{8}{$j_2$}
\dballoonInitl{8}{orange}{0}

\nvar{10}{$i_3$}
\dballoonInitl{10}{cyan}{-0.07}

\nvar{11}{$j_3$}
\dballoonInitl{11}{orange}{-0.07}

\nvar{13}{$i_4$}
\dballoonInitl{13}{cyan}{-0.12}

\nvar{14}{$j_4$}
\dballoonInitl{14}{orange}{-0.12}

\end{scope}

\begin{scope}[yshift=-1.3*\yn cm]
\numlineIds{$\tau'$}

\nvar{1}{$i_0$}
\dballoonInitl{1}{cyan}{0.2}

\nvar{2}{$j_0$}
\dballoonInitl{2}{orange}{-0.2}

\nvar{4}{$i_1$}
\dballoonInitl{4}{cyan}{0.1}

\nvar{5}{$j_1$}
\dballoonInitl{5}{orange}{-0.2}

\nvar{7}{$i_2$}
\dballoonInitl{7}{cyan}{0}

\nvar{8}{$j_2$}
\dballoonInitl{8}{orange}{-0.2}

\nvar{10}{$i_3$}
\dballoonInitl{10}{cyan}{-0.07}

\nvar{11}{$j_3$}
\dballoonInitl{11}{orange}{-0.2}

\nvar{13}{$i_4$}
\dballoonInitl{13}{cyan}{-0.12}

\nvar{14}{$j_4$}
\dballoonInitl{14}{orange}{-0.2}

\end{scope}

\end{tikzpicture}
\caption{Top: Initial run $\tau$ with two non-empty balloons of the
same type: cyan balloon inflated at $i_0$ and orange balloon inflated
at $j_0$. Middle: Switching cyan and orange balloons in
the part of $\tau$ between $j_2$ and $i_3$. Bottom: Modified run
$\tau'$ obtained by shifting token from orange balloon to cyan
balloon.}
\label{id-switch-fig}
\end{figure}
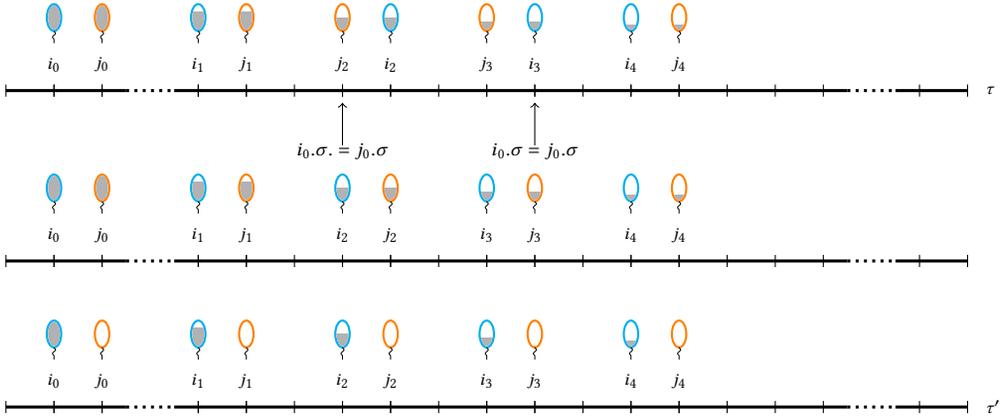

The id-switching surgery is depicted in Figure \ref {id-switch-fig} . The
formal proof of correctness uses runs-with-id. Fix a run-with-id $\tau=d_0 \autstep[\op_1] d_1
\autstep[\op_2] d_2 \cdots$. 
Suppose the cyan and orange balloons are inflated at $i_0$ and $j_0$
respectively
with the same balloon state. Since the type information is included in the
balloon state, this implies that they are also of the same type $t=
(L_t,S_t)$.
The points marked with indices $i_1,i_2,i_3,i_4$ (resp.
$j_1,j_2,j_3,j_4$) are
those at which the cyan (resp. orange) balloon undergoes a deflate
from $S_t$. While
$i_1 < j_1$ and $i_4 < j_4$, we have $j_2 < i_2$ and $j_3 < i_3$;
therefore token-shifting is not possible.
However, let  us assume that the state of the cyan and orange
balloons
is
the same
at $d_{j_2-1}$ and $d_{i_3}$ of $\tau$. This implies that the
operations
performed on the two balloons in $\tau[j_2,i_3]$ can be switched as
shown in the middle of Figure \ref{id-switch-fig}. Note that this
need not be a valid
run
as the number of tokens transferred by the orange at $j_2$ may exceed
that
transferred by the cyan balloon at $i_2$ and the extra tokens may be
required for
the run $\tau[j_2,i_2]$ to be valid. However, the id-switching now
enables a token-shifting operation since $j_k <i_k$ for each $k \in
\set{0,1,\ldots,4}$. Thus, combining
the switch with a token-shifting operation which moves
all tokens from orange to cyan results in the valid run $\tau'$
shown at the bottom of Figure \ref{id-switch-fig}, which contains one less non-empty
balloon of
type $t$ than $\tau$. It remains to show that such an
id-switching
surgery is always possible in a run $\tau$ when the number of
non-empty balloons of a type $t$ exceeds the bound $N$ given in the
lemma.

\paragraph{Ramsey's Theorem} \label{ramsey} To this end, we employ
  the well-known (finite) Ramsey's theorem~\cite[Theorem B]{Ramsey1930}, which
  we recall first.  For a set $S$ and $k\in\nats$, we denote by
  $\powerset[k]{S}$ the set of all $k$-element subsets of $S$. An
  \emph{$r$-colored (complete) graph} is a tuple $(V,E_1,\ldots,E_r)$,
  where $V$ is a finite set of \emph{vertices} and the sets
  $E_1,\ldots,E_r$ form a partition of all possible edges
  (i.e.~two-element subsets), i.e.
  $\powerset[2]{V}=E_1\mathop{\dot{\cup}}\cdots\mathop{\dot{\cup}}
  E_r$. A subset $U\subseteq V$ of vertices is
  \emph{monochromatic} if all edges between members of $U$ have the
  same color, in other words, if $\powerset[2]{U}\subseteq E_j$
  for some $j\in[1,r]$.  Ramsey's theorem says that for each
  $r,n\in\nats$, there is a number $R(r;n)$ such that any $r$-colored
  graph with at least $R(r;n)$ vertices contains a monochromatic subset of
  size $n$. It is a classical result by Erd\H{o}s and
  Rado~\cite[Theorem~1]{ErdosRado1952} that $R(r;n)\le r^{r(n-2)+1}$.

\def \polyside {2} %

\newcommand{\numlineR}[1]{

	\coordinate (n0) at (0*\n, 0) {};
	\coordinate (nlast) at (25*\n, 0) {};
	\node[label=right:{\small #1}] at (nlast) {};

	\draw[line width=\l cm, black] (n0) -- (nlast);

	\num{0}{};
	\num{1}{};
	\num{2}{};
	\num{3}{};
	\num{4}{};
	\num{5}{};
	\num{6}{};
	\num{7}{};
	\num{8}{};
	\num{9}{};
	\num{10}{};
	\num{11}{};
	\num{12}{};
	\num{13}{};
	\num{14}{};
	\num{15}{};
	\num{16}{};
	\num{17}{};
	\num{18}{};
	\num{19}{};
	\num{20}{};
	\num{21}{};
	\num{22}{};
	\num{23}{};
	\num{24}{};
	\num{25}{};
}

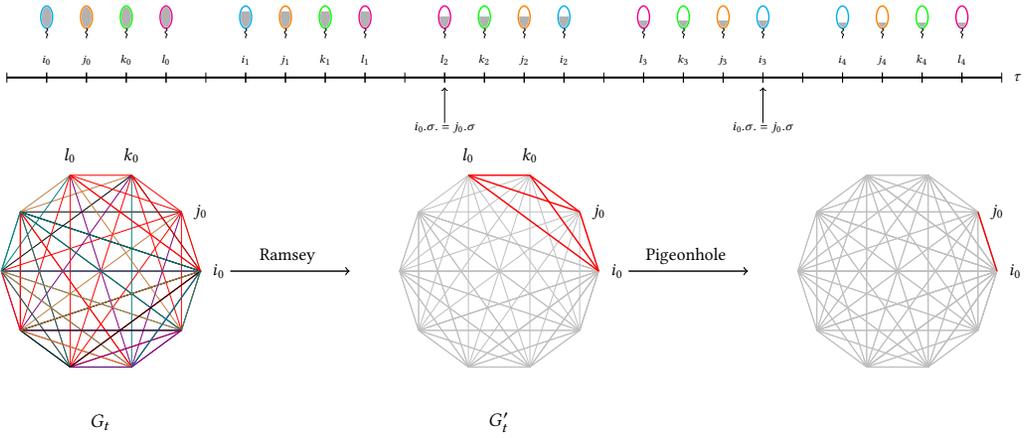
\begin{figure}
\resizebox{.99\linewidth}{!}{
	\begin{tikzpicture}
		\numlineR{$\tau$}
\nvar{1}{$i_0$}
\dballoonInitl{1}{cyan}{0.2}

\nvar{2}{$j_0$}
\dballoonInitl{2}{orange}{0.2}

\nvar{3}{$k_0$}
\dballoonInitl{3}{green}{0.2}

\nvar{4}{$l_0$}
\dballoonInitl{4}{magenta}{0.2}

\nvar{6}{$i_1$}
\dballoonInitl{6}{cyan}{0.1}

\nvar{7}{$j_1$}
\dballoonInitl{7}{orange}{0.1}

\nvar{8}{$k_1$}
\dballoonInitl{8}{green}{0.1}

\nvar{9}{$l_1$}
\dballoonInitl{9}{magenta}{0.1}

\nvar{11}{$l_2$}
\dballoonInitl{11}{magenta}{0}
\node at (11*\n,-1) {\scriptsize$i_0.\sigma.=j_0.\sigma$};
\draw[->] (11*\n,-0.9) -- (11*\n,-0.2);

\nvar{12}{$k_2$}
\dballoonInitl{12}{green}{0}

\nvar{13}{$j_2$}
\dballoonInitl{13}{orange}{0}

\nvar{14}{$i_2$}
\dballoonInitl{14}{cyan}{0}

\nvar{16}{$l_3$}
\dballoonInitl{16}{magenta}{-0.07}

\nvar{17}{$k_3$}
\dballoonInitl{17}{green}{-0.07}

\nvar{18}{$j_3$}
\dballoonInitl{18}{orange}{-0.07}

\nvar{19}{$i_3$}
\dballoonInitl{19}{cyan}{-0.07}
\node at (19*\n,-1) {\scriptsize$i_0.\sigma.=j_0.\sigma$};
\draw[->] (19*\n,-0.9) -- (19*\n,-0.2);

\nvar{21}{$i_4$}
\dballoonInitl{21}{cyan}{-0.12}

\nvar{22}{$j_4$}
\dballoonInitl{22}{orange}{-0.12}

\nvar{23}{$k_4$}
\dballoonInitl{23}{green}{-0.12}

\nvar{24}{$l_4$}
\dballoonInitl{24}{magenta}{-0.12}

	\end{tikzpicture}
  }
  \resizebox{.99\linewidth}{!}{
	\begin{tikzpicture}
\pgfmathdeclarerandomlist{MyRandomColors}{%
    {red}%
    {teal}%
    {black}
    {brown}%
    {violet}%
}

    \foreach \i in {0, ..., 10} {
     \coordinate (\i) at (\i*36:\polyside) {};
    }
  	\foreach \x in {0, ..., 10}
        \foreach \y in {4,...,10} {
        \pgfmathrandomitem{\RandomColor}{MyRandomColors} 
            \draw[\RandomColor] (\x)--(\y);
        }

   \draw[red] (0)--(1);
    \draw[red] (0)--(2);
    \draw[red] (0)--(3);
    \draw[red] (1)--(2);
    \draw[red] (1)--(3);
    \draw[red] (2)--(3);
    \node[label=right:$i_0$] at (0) {};
    \node[label=right:$j_0$] at (1) {};
    \node[label=above:$k_0$] at (2) {};
    \node[label=above:$l_0$] at (3) {};
    \node at (0,-3) {\large $G_t$};

   \begin{scope}[xshift=8cm]
   	\foreach \i in {0, ..., 10} {
     \coordinate (n\i) at (\i*36:\polyside) {};
    }

  	\foreach \x in {0,1,...,10}
        \foreach \y in {4,...,10} {
            \draw[black!25] (n\x)--(n\y);
        }
    \draw[red,thick] (n0)--(n1);
    \draw[red,thick] (n0)--(n2);
    \draw[red,thick] (n0)--(n3);
    \draw[red,thick] (n1)--(n2);
    \draw[red,thick] (n1)--(n3);
    \draw[red,thick] (n2)--(n3);
    \node[label=right:$i_0$] at (n0) {};
    \node[label=right:$j_0$] at (n1) {};
    \node[label=above:$k_0$] at (n2) {};
    \node[label=above:$l_0$] at (n3) {};
    \node at (0,-3) {\large $G'_t$};
   \end{scope}
   \draw[->,thick] (2.6,0) -- (5,0);
   \node at (3.75,0.3) {Ramsey};

   \begin{scope}[xshift=16cm]
   	\foreach \i in {0, ..., 10} {
     \coordinate (m\i) at (\i*36:\polyside) {};
    }
  	\foreach \x in {0, ..., 10}
        \foreach \y in {0,...,10} {
            \draw[black!25] (m\x)--(m\y);
        }
   \draw[red,thick] (m0)--(m1);
    \node[label=right:$i_0$] at (m0) {};
    \node[label=right:$j_0$] at (m1) {};
   \end{scope}
  \draw[->,thick] (10.6,0) -- (13,0);
   \node at (11.75,0.3) {Pigeonhole};

	\end{tikzpicture}
  }
	 \caption{Above:Four balloons inflated at $i_0<j_0<k_0<l_0$ of the
	 same type $t$ performing
four deflate
operations (subscripts denote deflate operations of the same
balloon). Note that the ordering relationship of deflate
operations between any pair of the four balloons is the same: between
balloons $i_0$ and $j_0$, their deflate sequences are related as
$i_1<j_1,j_2
<i_2,j_3<i_3,i_4 <j_4$, which is represented by the string $0110$. The
edge-color red is used to represent $0110$ in the figure.
The balloons $i_0$ and $j_0$ share the same states at
configurations $d_{l_2}-1$ and
$d_{i_3}$ of $\tau$. \\
Below: The same four balloons inflated at $i_0<j_0<k_0<l_0$ shown as
forming a monochromatic
subgraph $G'_t$ in the graph $G_t$. For large enough $|G'_t|$, by
Pigeonhole Principle we find $i_0,j_0$ which share the same states.}
\label{ramsey-fig}
\end{figure}

The application of Ramsey's Theorem is shown in Figure 
\ref{ramsey-fig}. The bottom half of the figure depicts the
construction of a graph $G_t$ whose vertices are balloons, on which
Ramsey's theorem is applied. This results in the identification of a
monochromatic clique with vertices $i_0,j_0,k_0,l_0$.
The
top half of the figure shows the deflate operations on the
balloons
inflated at $i_0,j_0,k_0,l_0$ in the run $\tau$.

Formally, we construct a
graph $G_t$ with vertex set 
$V_t$ of all id's in $\tau$ of a fixed type
$t$. By assumption, $|V_t| \ge N$. For id's $i,j \in V_t$ with $i<j$
and $S=(\barp_1,p_1),\cdots,
(\barp_n,p_n)$, define a sequence $s_{i,j} \in
\set{0,1}^{|S|}$ by $s_{i,j}(k)=0$ if and only if $i$ undergoes the
deflate transferring tokens from $\barp_k$ to $p_k$ before $j$ does.

In Figure 3, assuming that $|S|=4$, we have $s_{i_0,j_0}=0110$. 
Interpreting each word from $\set{0,1}^{|S|}$ as a color, we obtain a
finite coloring of the edges of $G_t$. Red colored edges in $G_t$ are
to be interpreted as the string $0110$, with other colors representing
other strings. For a large enough value of
$N$, Ramsey's Theorem gives us a monochromatic subgraph $G'_t$ of
$G_t$ induced by a set of vertices $V'_t$. As shown in the figure,
any pair of balloons chosen from the
cyan, orange, green and magenta balloons inflated at $i_0,j_0,k_0,l_0$
respectively, behave in the same way with respect to their order of
deflates and thus form a monochromatic subgraph $G'_t$ colored red.

Let a
maximal contiguous sequence of $1$'s in $s_{i,j}$ be called a
\emph{$1$-block}. Since the number of balloon
states is finite, this implies that for a large enough value of
$|V'_t|$, there will exist two id's $i_0,j_0 \in |V'_t|$ (represented
by the cyan and orange balloons respectively) which share
the
same states at configurations at the beginning and end of every
$1$-block by the Pigeonhole Principle. The id-switching surgery can be
performed on the cyan and orange balloons, which are the ones
considered in the
id-switching surgery of Figure \ref{id-switch-fig}. While Ramsey's
Theorem gives
a double-exponential
bound
on $N$ in order to obtain a large monochromatic subgraph, the
second condition requiring same states at the beginning and end of
$1$-blocks further increases our requirement to $\exp_4$ for
id-switching to
be enabled.

This concludes the proof of \cref{lem:N-bal-bdd} as well as \cref{th:vassb-fair-term}.

\newcommand{\wdcl}[2]{#1\mathord{\downarrow}_{#2}}
\renewcommand{\vec}[1]{\mathbf #1}
\newcommand{\femph}[1]{\textbf{#1}}

\newcommand{\frakp}{\mathfrak{p}}
\newcommand{\frakf}{\mathfrak{f}}
\newcommand{\frakt}{\mathfrak{t}}
\newcommand{\fraks}{\mathfrak{s}}
\newcommand{\frakm}{\mathfrak{m}}
\newcommand{\fraku}{\mathfrak{u}}
\newcommand{\frakP}{\mathfrak{P}}
\newcommand{\frakK}{\mathfrak{K}}
\newcommand{\frakT}{\mathfrak{T}}
\newcommand{\frakS}{\mathfrak{S}}
\newcommand{\frakV}{\mathfrak{V}}

\section{Starvation}\label{starvation}
We now prove \cref{thm:fair-starv}.
Let us first explain the additional difficulty
of the starvation problem. For deciding progressive termination, we
observed that each thread execution can be abstracted by its type
and the threads it spawns. (In other words, two executions
that agree in these data are interchangeable without affecting
progressiveness of a run.) However, for starvation of a thread, it is
also important whether each thread visits some stack content $w$
after $i$ context switches. Here, $w$ is not known in advance and has
to be agreed upon by an infinite sequence of threads.

Very roughly speaking, we reduce starvation to progressive termination
as follows.  For each thread, we track its spawned multiset up to some
bound $B$.  Using Ramsey's theorem~\cite[Theorem
B]{Ramsey1930}, we show that if we choose $B$ high enough, then this
abstraction already determines whether a sequence of thread executions
can be replaced with different executions that actually visit some
agreed upon stack content $w$ after $i$ context switches. The latter
condition permitting replacement of threads will be called
``consistency.''

A further subtlety is that consistency of the abstractions up to $B$
only guarantees consistency of the (unabstracted) executions if the
run is shallow. Here, \cref{dcps-spawn-bounded} will yield a
shallow run, so that we may conclude consistency of the
unabstracted executions.

\subsection*{Terminology}
In our terminology, a
thread is a pair $(w,i)$, where $w$ is a stack content and $i$ is a
context switch number.  To argue about starvation, it is 
convenient to talk about how a thread evolves over time. By a
\emph{(thread) execution} we refer to the sequence of (pushdown and
swap) instructions that belong to a single thread, from its creation
via spawn until its termination. A thread execution can spawn new
threads during each of its segments. We
say that a thread execution $e$ \emph{produces} the multiset
$\mmap\in\multiset{\Lambda}$, where $\Lambda=\Gamma\times\{0,\ldots,K\}$
if the following holds: For each $i\in\{0,\ldots,K\}$ and
$\gamma\in\Gamma$, the thread execution $e$ spawns $\mmap((\gamma,i))$
new threads with top of stack $\gamma$ in segment $i$. In this
case, we also call $\mmap$ the \emph{production} of $e$.

According to \cref{lem:progDCPS}, in order to decide $\STARV[K]$, it
suffices to decide whether in a given $\DCPS$ $\cA$, there exists a
\emph{progressive} run that starves some thread $(w,i)$. Therefore, we
say that a run $\rho$ is \emph{starving} if it is progressive and
starves some thread $(w,i)$.  Let us first formulate starvation in
terms of thread executions.  We observe that a progressive run $\rho$
starves a thread $(w,i)$ if and only if there are configurations
$c_1,c_2,\ldots$ and executions $e_1,e_2,\ldots$ in $\rho$ such that:
\begin{enumerate}
  \item For each $j=1,2,\ldots$, in configuration $c_j$, both $e_j$ and
    $e_{j+1}$ are in state $(w,i)$,
  \item $e_j$ is switched to in the step after $c_j$, and
  \item $e_{j+1}$ is not switched to until $c_{j+1}$.
\end{enumerate}
For the ``if'' direction, note that if a progressive run $\rho$ starves
$(w,i)$, then $(w,i)$ must be in the bag from some point on and
whenever $(w,i)$ becomes active, there are at least two instances of
$(w,i)$ in the bag.  We choose $c_1,c_2,\ldots$ as exactly those
configurations in $\rho$ after which $(w,i)$ becomes active.  Moreover,
$e_j$ is the thread execution that is switched to after
$c_j$. Furthermore, since in $c_j$, there must be another instance of
$(w,i)$ in the bag, there must be some execution $e'_j$ whose state
$(w,i)$ is in the bag at $c_j$. However, since $e_{j+1}$ will start
from $(w,i)$ in $c_{j+1}$ and $e'_j$ is in $(w,i)$ at $c_j$, we may
assume that $e_{j+1}=e'_j$. With this choice, we clearly satisfy
(1)--(3) above.

For the ``only if'' direction, note that conditions (1)--(3) allow
$(w,i)$ to become active in between $c_j$ and $c_{j+1}$. However, since
$e_{j+1}$ is not switched to between $c_j$ and $c_{j+1}$, we know that
any time $(w,i)$ becomes active, there must be another instance of $(w,i)$.

\subsection*{Consistency}
Our first step in deciding starvation is to find a reformulation
that does not explicitly mention the stack $w$. Instead, it states
the existence of $w$ as a consistency condition, which we will develop
now.

Of course, it suffices to check whether a $\DCPS$ can starve some thread
$(w,i)$ when $i\in[1,K]$ is fixed. Therefore, from now on, we choose
some $i\in[1,K]$ and want to decide whether there is a stack
$w\in\Gamma^*$ such that the thread $(w,i)$ can be starved by our
$\DCPS$.

First, a note on notation. In the following, we will abbreviate the
set $\cT(\cA,K)$ of thread types with $\cT$. We will work with
families $(X_t)_{t\in\cT}$ of subsets $X_t\subseteq X$ of some set $X$
indexed by types $t\in\cT$. We identify the set of such tuples indexed
by $\cT$ with $\powerset{X}^\cT$. Sometimes, it is more natural to
treat them as tuples $(X_1,\ldots,X_k)$ with $k=|\cT|$. For
simplicity, we will call both objects tuples.

For each type $t\in\cT$, we consider the following set
\begin{align*} S_t=\{(w,\vec m)\in\Gamma^*\times\multiset{\Lambda} \mid \text{there is an execution of type $t$ that produces $\mmap$} \\
  \text{and reaches stack $w$ after segment $i$}\}
\end{align*}
The set $S_t$ encodes the following information: Is there a thread
execution of type $t$ that produces $\mmap\in\multiset{\Lambda}$ and
at the same time arrives in $w$ after $i$ segments? The tuple
$\frakS_{\cA}=(S_t)_{t\in\cT}$ encodes this information for all types at
once.

We will analyze $\frakS_{\cA}$ to show that if our decision procedure claims
that there exists a starving run, then we can construct one. This
construction will involve replacing one execution with another that
(i)~has the same type, (ii)~arrives in $w$ after $i$ segments,
and (iii)~spawns more threads. Formally, the inserted execution must
be larger w.r.t.\ the following order: For
$\mmap,\mmap'\in\multiset{\Lambda}$, we have $\mmap\preceq_1\mmap'$ if and
only if $\mmap\preceq\mmap'$ and also $\sup(\mmap)=\sup(\mmap')$.
Recall that $\sup(\mmap)=\{x\in\Lambda\mid \mmap(x)>0\}$ is the support of $\mmap\in\multiset{\Lambda}$.
Here, the condition $\sup(\mmap)=\sup(\mmap')$ makes sure that the replacement
does not introduce thread spawns with new stack symbols, as this
might destroy progressiveness of the run.

Let $S\subseteq\Gamma^*\times\multiset{\Lambda}$ be a set.
For $w\in\Gamma^*$, we define
\[ \wdcl{S}{w} = \{\vec m\in\multiset{\Lambda} \mid \exists \vec m'\in\multiset{\Lambda}\colon \vec m\preceq_1 \vec m',~(w,\vec m')\in S\}.\]
Observe that $\mmap\in\wdcl{S_t}{w}$ expresses that there exists an
execution of type $t$ that visits $w$ after segment $i$ and
produces a vector $\mmap'\succeq_1 \mmap$.

Our definition of consistency involves the tuple
$\frakS_{\cA}$. However, since some technical proofs will be more natural in a
slightly more abstract setting, we define consistency for a general
tuple $\frakS=(S_1,\ldots,S_k)$ of subsets
$S_l\subseteq\Gamma^*\times\multiset{\Lambda}$. Hence, the following
definitions should be understood with the case $\frakS=\frakS_{\cA}$
in mind. Suppose we have a run with thread executions $e_1,e_2,\ldots$
and for each type $t$, let $V_t$ be the set of productions of all
executions in $\{e_1,e_2,\ldots\}$ that have type $t$. We want to
formulate a condition expressing the existence of a stack $w$ such
that for any $t\in\cT$ and any multiset $\mmap\in V_t$, there exists
an execution of type $t$ that visits $w$ (after $i$ context switches)
and produces a multiset $\mmap'\succeq_1 \mmap$. This would allow us
to replace each $e_j$ by an execution $e'_j$ that actually visits $w$:
Since $\mmap'\succeq_1\mmap$, we know that $e'_j$ produces more
threads of each stack symbol (and can thus still sustain the run), but
also does not introduce new kinds of threads (because $\mmap'$ and
$\mmap$ have the same support), so that progressiveness will not
be affected by the replacement.

Let us make this formal. We say that a tuple $\frakV=(V_1,\ldots,V_k)$ with
$V_l\subseteq\multiset{\Lambda}$ is \emph{$\frakS$-consistent} if
there exists a $w\in\Gamma^*$ with $V_l\subseteq \wdcl{S_l}{w}$ for
each $l\in[1,k]$. 
In this case, we call $w$ an
\emph{$\frakS$-consistency witness} for $\frakV$. For words
$w,w'\in\Gamma^*$, we write $w\le_\frakS w'$ if
$\wdcl{S_l}{w}\subseteq\wdcl{S_l}{w'}$ for every $l\in[1,k]$.

\subsection*{Starvation in Terms of Consistency}
This allows us to state the following reformulation of starvation,
where $w$ does not appear explicitly.
A progressive run $\rho$ is said to be \emph{consistent} if
there are configurations $c_1,c_2,\ldots$ and thread
executions $e_1,e_2,\ldots$ that produce $\vec m_1,\vec m_2,\ldots$
and such that:
\begin{enumerate}
\item For each $j=1,2,\ldots$, in configuration $c_j$, the executions
  $e_j$ and $e_{j+1}$ have completed $i$ segments,
\item $e_j$ is switched to in the step after $c_j$,
\item $e_{j+1}$ is not switched to until $c_{j+1}$, and:
\item Let $V_t=\{\vec m_j \mid \text{$j\in\nats$, execution
$e_j$ has
    type $t$}\}$. Then the tuple $\frakV=(V_t)_{t\in\cT}$ is
  $\frakS_{\cA}$-consistent.
\end{enumerate}
Note that the consistency condition in (4) expresses that there exists
a stack
content $w$ such that we could, instead of each $e_j$, perform a
thread execution that actually visits $w$. It is thus straightforward to show:
\begin{restatable}{lemma}{starvingConsistent}\label{starving-consistent}
  A $\DCPS$ has a starving run if and only if it has a consistent run.
\end{restatable}

\subsection*{Tracking Consistency}
Our next step is to find some finite data that we can track about each
of the produced vectors $\vec m_1,\vec m_2,\ldots$ such that this data
determines whether the tuple $(V_t)_{t\in\cT}$ is $\frakS_{\cA}$-consistent.
We do this by abstracting vectors ``up to
a bound.''  Let $B\in\nats$. We define the map
$\alpha_B\colon\multiset{\Lambda}\to\multiset{\Lambda}$ by
$\alpha_B(\mmap)=\mmap'$, where $\mmap'(x) =\min(\mmap(x), B)$
for $x\in\Lambda$. We
naturally extend $\alpha_B$ to subsets of $\multiset{\Lambda}$
(point-wise) and to tuples of subsets of $\multiset{\Lambda}$
(component-wise). Note that for a tuple $\frakV=(V_t)_{t\in\cT}$ with
$V_t\subseteq\multiset{\Lambda}$ for $t\in\cT$, the tuple
$\alpha_B(\frakV)$ belongs to the finite set
$\powerset{[0,B]^{\Lambda}}^{\cT}$.
The following \lcnamecref{abstraction-consistency} tells us that by
abstracting w.r.t.\ some suitable $B$, we do not lose information
about $\frakS_{\cA}$-consistency.

\begin{theorem}\label{abstraction-consistency}
  Given a $\DCPS$ $\cA$, there is an effectively computable bound
  $B\in\nats$ such that the following holds. If
  $\frakV=(V_t)_{t\in\cT}$ is a tuple of \femph{finite}
  subsets $V_t\subseteq\multiset{\Lambda}$, then $\frakV$ is
  $\frakS_{\cA}$-consistent if and only if $\alpha_B(\frakV)$ is
  $\frakS_{\cA}$-consistent.
\end{theorem}
Roughly speaking, \cref{abstraction-consistency} allows us to check
for the existence of a consistent run by checking whether there is one
with an $\frakS_{\cA}$-consistent tuple $\alpha_B(\frakV)$.  However,
we may only conclude consistency of $\frakV$ (and hence of the run)
from consistency of $\alpha_B(\frakV)$ if $\frakV$ is finite. To
remedy this, we shall employ the fact that a DCPS with a progressive
run also has a shallow progressive run (\cref{dcps-spawn-bounded}). We
will show that if our algorithm detects a run $\rho$ with consistent
$\alpha_B(\frakV)$, then there also exists a run $\rho'$ with finite
$\frakV$ such that $\alpha_B(\frakV)$ is consistent, meaning by
\cref{abstraction-consistency}, $\rho'$ has to be consistent.

Moreover, given a tuple of finite subsets, we can decide
$\frakS_{\cA}$-consistency:
\begin{theorem}\label{consistency-decidable}
  Given a tuple $\frakV=(V_t)_{t\in\cT}$ of finite
  subsets $V_t\subseteq\multiset{\Lambda}$,
  it is decidable whether $\frakV$ is $\frakS_{\cA}$-consistent.
\end{theorem}

\subsection*{Deciding Starvation}
We will prove \cref{abstraction-consistency,consistency-decidable} later in this section.
Before we do that, 
let us show how they are used to decide starvation. First, we use
\cref{abstraction-consistency} to compute $B\in\N$. Let us fix $B$ for the decision procedure. Let
$\fraku\in\powerset{[0,B]^{\Lambda}}^{\cT}$ with $\fraku=(U_t)_{t\in\cT}$. We use a lower-case letter for
this tuple to emphasize that it is of bounded size.
An infinite progressive run $\rho$
of $\cA$ is said to be \emph{$(i,\fraku)$-starving} if it contains
configurations $c_1,c_2,\ldots$ and executions $e_1,e_2,\ldots$ that
produce $\vec m_1,\vec m_2,\ldots$ such that:
\begin{enumerate}
\item For each $j=1,2,\ldots$, in configuration $c_j$, the executions
  $e_j$ and $e_{j+1}$ have completed $i$ segments,
\item $e_j$ is switched to in the step after $c_j$,
\item $e_{j+1}$ is not switched to until $c_{j+1}$, and:
\item Let $V_t=\{\vec m_j \mid \text{$j\in\nats$, execution $e_j$ has
    type $t$}\}$. Then $\alpha_B(V_t)\subseteq U_t$ for each $t\in\cT$.
\end{enumerate}

Now using the bound $B$ from \cref{abstraction-consistency}, we can show the
following.
\begin{restatable}{lemma}{spawnBoundedStarving}\label{spawn-bounded-starving}
  If $\cA$ has a starving run, then it has an $(i,\fraku)$-starving run
  for some $i\in[1,K]$ and some $\frakS_{\cA}$-consistent
  $\fraku\in\powerset{[0,B]^{\Lambda}}^{\cT}$.  Moreover, if $\cA$ has a
  \femph{shallow}
  $(i,\fraku)$-starving run for some $i\in[1,K]$
  and some $\frakS_{\cA}$-consistent $\fraku\in\powerset{[0,B]^{\Lambda}}^{\cT}$,
  then it has a starving run.
\end{restatable}
Here, we need to assume shallowness for the converse direction because
we need finiteness of $\frakV$ in the converse of
\cref{abstraction-consistency}.

Because of \cref{spawn-bounded-starving}, we can proceed as follows to
decide starvation.  We first guess a tuple
$\fraku\in\powerset{[0,B]^\Lambda}^{\cT}$ and check whether it is
$\frakS_{\cA}$-consistent using \cref{consistency-decidable}. Then, we
construct a $\DCPS$ $\cA_{(i,\fraku)}$ such that $\cA_{(i,\fraku)}$
has a progressive run if $\cA$ has an $(i,\fraku)$-starving
run. Moreover, we use the fact every $\DCPS$ that has a progressive
infinite run also has a shallow infinite run
(\cref{dcps-spawn-bounded}). This will allow us to turn a progressive
run of $\cA_{(i,\fraku)}$ into a shallow $(i,\fraku)$-starving run of
$\cA$, which must be starving by \cref{spawn-bounded-starving}. Let us
now see how to construct $\cA_{(i,\fraku)}$.

\subsection*{Freezing $\DCPS$}
For constructing $\cA_{(i,\fraku)}$, it is convenient to
have a simple locking mechanism available, which we call ``freezing.''
It will be easy to see that this can be implemented in $\DCPS$. In a
freezing $\DCPS$, there is one distinguished ``frozen'' thread in each
configuration. It cannot be resumed using the ordinary resume rules.
It can only be resumed using an unfreeze operation, which at the same
time freezes another thread. We use this to make sure that the
$e_{j+1}$ stays inactive between $c_j$ and $c_{j+1}$.

\newcommand{\freeze}{\text{{\tiny \SnowflakeChevron}}}
\newcommand{\freezerule}[4]{#1\mapsto #2\lhd #3 \mathbin{\freeze} #4}

Syntactically, a \emph{freezing $\DCPS$} is a tuple
$\cA=(G,\Gamma,\Delta,g_0,\gamma_0,\gamma_f)$, where
$(G,\Gamma,\Delta,g_0,\gamma_0)$ is a $\DCPS$, except that the rules
$\Delta$ also contain a set $\Deltau$ of \emph{unfreezing rules} of
the form $\freezerule{g}{g'}{\gamma}{\gamma'}$ and $\gamma_f$ is the
initial frozen thread with a single stack symbol.  The unfreezing
rules
allow the $\DCPS$ to unfreeze and resume a thread with top of stack
$\gamma$, while also freezing a thread with top of stack
$\gamma'$. A \emph{configuration} is a tuple in
$G\times(\Gamma^*\times\nats \cup
\{\#\})\times\multiset{\hat{\Gamma}^*\times\nats}$, where
$\hat{\Gamma}=\Gamma\cup{\Gamma^\freeze}$ and
${\Gamma^\freeze}=\{ {\gamma^\freeze} \mid \gamma\in \Gamma\}$. A
thread is \emph{frozen} if its top-of-stack belongs to
$\Gamma^\freeze$. It will be clear from the steps %
that in each reachable configuration, there is exactly one frozen
thread.

A freezing $\DCPS$ has the same steps as those of the corresponding
$\DCPS$. In particular, those apply only to top-of-stack symbols in
$\Gamma$. In addition, there is one more rule:
\[ \inferrule[Unfreeze]{
    \freezerule{g}{g'}{\gamma}{\gamma'}
  }{
    \langle g,\#,\mmap + \multi{{\gamma^\freeze} w,l} + \multi{\gamma' w',j}\rangle\mapsto \langle g',(\gamma w,l),\mmap + \multi{\gamma'^\freeze w',j}\rangle
  }
\]
Hence, the frozen thread $({\gamma^\freeze} w,l)$ is unfrozen and
resumes, while the thread $(\gamma'w',j)$ becomes the new frozen
thread.  Moreover, the initial configuration is
$\langle g_0,\#, \multi{(\gamma_0,0)} + \multi{({\gamma_f^\freeze},0)}\rangle$. 

Given these additional steps, progressive termination is defined as for
$\DCPS$. (In particular, the progressiveness condition also applies to frozen
threads.)
\begin{lemma}\label{liveness-freezing}
  Given a freezing $\DCPS$ $\cA$, it is decidable whether $\cA$ has a
  progressive run.  Moreover, if $\cA$ has a progressive run, then it
  has a shallow progressive run.
\end{lemma}
\Cref{liveness-freezing} can be shown using a straightforward
reduction to progressive termination of ordinary $\DCPS$.  The
freezing is realized by introducing stack symbols
$\Gamma^\freeze=\{\gamma^\freeze\mid \gamma\in\Gamma\}$.  An unfreeze
rule $\freezerule{g}{g'}{\gamma}{\gamma'}$ for a thread
$(\gamma^\freeze w,l)$ is then simulated by a simple locking mechanism
using a bounded number of context switches: It turns a thread with
stack $\gamma'w'$ into one with stack $\gamma'^\freeze w'$ (using
context switches) and then resumes $(\gamma^\freeze w,j)$, where
initially, $\gamma^\freeze$ is replaced with $\gamma$. Other than
that, for threads with top of stack in $\Gamma^\freeze$, there are no
resume rules.  Since each thread in a freeze $\DCPS$ can only be
frozen and unfrozen at most $K$ times, the constructed $\DCPS$ uses at
most $2K +1$ context-switches to simulate a run of the freeze $\DCPS$.

\subsection*{Reduction to Progressive Runs in Freezing $\DCPS$}
We now reduce starvation to progressive runs in freezing $\DCPS$. We first guess
a pair $(i,\fraku)$ with $i\in[1,K]$ and a $\frakS_{\cA}$-consistent
$\fraku\in\powerset{[1,B]^{\Lambda}}^{\cT}$, $\fraku=(U_t)_{t\in\cT}$,
and construct a freezing $\DCPS$ $\cA_{(i,\fraku)}$ so that $\cA$ has an
$(i,\fraku)$-starving run if and only if $\cA_{(i,\fraku)}$ has a
progressive run.  Moreover, if $\cA_{(i,\fraku)}$ has a progressive
run, then $\cA$ even has a shallow $(i,\fraku)$-starving run.
Therefore, $\cA$ has a starving run if and only if for some choice of
$(i,\fraku)$, $\cA_{(i,\fraku)}$ has a progressive run.

Intuitively, we do this by tracking for each thread execution the
multiset $\alpha_B(\mmap)$, where $\mmap$ is its production. Using
frozen threads, we make sure that every progressive run in $\cA$
contains executions $e_1,e_2,\ldots$ to witness
$(i,\fraku)$-starvation. To verify the $(i,\fraku)$-starvation, we also
track each thread's type and current context-switch number.  Hence,
we store a tuple $(t,j,\bar{\mmap},\bar{\nmap})$, where (i)~$t$ is the
type, (ii)~$j$ is the current context-switch number,
(iii)~$\bar{\mmap}$ is the guess for $\alpha_B(\mmap)$, where
$\mmap\in\multiset{\Lambda}$ is the entire production of the execution,
and (iv)~$\bar{\nmap}$ is $\alpha_B(\nmap)$, where
$\nmap\in\multiset{\Lambda}$ is the multiset spawned so far.

While a thread is inactive, the extra information is stored on the top
of the stack, resulting in stack symbols
$(\gamma,t,j,\bar{\mmap},\bar{\nmap})$. In particular, when we spawn a
new thread, we immediately guess its type $t$ and the abstraction
$\bar{\mmap}$, and we set $j=0$ and $\bar{\nmap}=\emptyset$. The
freezing and unfreezing works as follows.  Initially, we have the
frozen thread $\gamma_\dagger$ (where $\gamma_\dagger$ is a fresh stack
symbol).  To unfreeze it, we have to freeze a thread of some type
$t$ where $\bar{\mmap}$ belongs to $U_t$ (recall that this is a component of $\fraku$):
\[ g\mapsto g'\lhd \gamma_\dagger\mathbin{\freeze} (\gamma,t,i,\bar{\mmap},\bar{\nmap}) \]
for every $g,g'\in G$, $t\in\cT$, $\bar{\mmap}\in U_t$. To unfreeze
(and thus resume) a thread with top of stack
$(\gamma,t,i,\bar{\mmap},\bar{\nmap})$, we have to freeze a thread
$(\gamma',t',i,\bar{\mmap}',\bar{\nmap}')$ with
$\bar{\mmap}'\in U_{t'}$. Unfreezing requires
context-switch number $i$, because the executions $e_1,e_2,\ldots$
must be in segment $i$ in 
$c_1,c_2,\ldots$:
\[ g\mapsto \widehat{g'}\lhd (\gamma,t,i,\bar{\mmap},\bar{\nmap})\mathbin{\freeze} (\gamma',t',i,\bar{\mmap}',\bar{\nmap}') \]
for each resume rule $g\mapsto g'\lhd \gamma$, $t$, $\bar{\mmap}$,
and
$\bar{\nmap}$, provided that $g$ is the state specified in $t$ to
enter from in the $i$th segment.  Here, $\widehat{g'}$ is a
decorated version of $g'$, in which the thread can only transfer the
extra information related to $t,i,\bar{\mmap},\bar{\nmap}$
back to the global state.
Symmetrically, when interrupting a thread that information is
transferred back to the stack and the segment counter $j$ is
incremented.
To resume an ordinary (i.e.\ unfrozen) inactive thread, we have a
resume rule $g\mapsto g'\lhd (\gamma,t,j,\bar{\mmap},\bar{\nmap})$
for each resume rule $g\mapsto g'\lhd \gamma$ and each $t$, $j$,
$\bar{\mmap}$, and $\bar{\nmap}$ --- if $g$ is specified as
the entering global state for segment $j$ in $t$.  While
a thread is active, it keeps $\bar{\nmap}$ up to date by
recording all spawns (and reducing via $\alpha_B$).  Finally, when a
thread terminates, it checks that the components $\bar{\mmap}$ and
$\bar{\nmap}$ agree. 

It is clear from the construction that $\cA$ has a
$(i,\fraku)$-starving run if and only if $\cA_{(i,\fraku)}$ has a
progressive run. Moreover, \cref{liveness-freezing} tells us that if
$\cA_{(i,\fraku)}$ has a progressive run, then it has a shallow
progressive run.  This shallow progressive run clearly yields a
shallow $(i,\fraku)$-starving run of $\cA$. According to
\cref{spawn-bounded-starving}, this implies that $\cA$ has a starving
run. This establishes the following lemma, which implies that starvation is
decidable for $\DCPS$.
\begin{lemma}
  $\cA$ has a starving run if and only if for some $i\in[1,K]$ and
  some $\frakS_{\cA}$-consistent $\fraku\in\powerset{[0,B]^\Lambda}^{\cT}$, the
  freezing $\DCPS$ $\cA_{(i,\fraku)}$ has a progressive run.
\end{lemma}

\subsection*{Proving \cref{abstraction-consistency,consistency-decidable}}
It remains to prove
\cref{abstraction-consistency,consistency-decidable}.  We will use a
structural description of the sets $S_t$
(\cref{stack-vector-rational}), which requires some terminology.  An
\emph{automaton over $\Gamma^*\times\multiset{\Lambda}$} is a tuple
$\cM=(Q,E,q_0,q_f)$, where $Q$ is a finite set of \emph{states},
$E\subseteq Q\times \Gamma^*\times\multiset{\Lambda}\times Q$ is a
finite set of \emph{edges}, $q_0\in Q$ is its \emph{initial state},
and $q_f\in Q$ is its \emph{final state}. We write
$p\autstep[u|\mmap] q$ if there is a sequence
$(p_0,u_1,\mmap_1,p_1),
(p_1,u_2,\mmap_2,p_2),\ldots,(p_{n-1},u_n,\mmap_n,p_n)$ of edges in
$\cM$ with $p=p_0, q=p_n$, $u=u_1\cdots u_n$, and
$\mmap=\mmap_1+\cdots+\mmap_n$.  The set \emph{accepted by $\cM$} is
the set of all $(w,\mmap)\in\Gamma^*\times\multiset{\Lambda}$ with
$q_0\autstep[w|\mmap]q_f$. A subset of
$\Gamma^*\times\multiset{\Lambda}$ is \emph{rational} if it is accepted
by some automaton over $\Gamma^*\times\multiset{\Lambda}$.
\begin{lemma}\label{stack-vector-rational}
For every $t\in\cT$, the set $S_t$ is effectively rational.
\end{lemma}
This can be deduced from \cite[Lemma 6.2]{Zetzsche2013a}. Since the
latter would require introducing a lot of machinery, we include a
direct proof in the full version.
Both proofs are slight extensions of B\"{u}chi's proof of regularity of the set of
reachable stacks in a pushdown automaton~\cite[Theorem
1]{buchi1964regular}. The only significant difference is the following:
While \cite{buchi1964regular} essentially introduces shortcut edges
for runs that go from one stack $w$ back to $w$, we glue in a finite
automaton that produces the same output over $\Lambda$ as such
runs. This is possible since the set of the resulting multisets is
always semi-linear by Parikh's theorem~\cite[Theorem 2]{Parikh66}.

Because of \cref{stack-vector-rational}, the following immediately
implies \cref{consistency-decidable}:
\begin{lemma}
  Given a tuple $\frakS=(S_1,\ldots,S_k)$ of rational subsets
  $S_j\subseteq\Gamma^*\times\multiset{\Lambda}$ and a tuple
  $\fraku=(U_1,\ldots,U_k)$ of finite subsets
  $U_j\subseteq\multiset{\Lambda}$, it is decidable whether $\fraku$ is
  $\frakS$-consistent.
\end{lemma}
\begin{proof}
  Since $S_j$ is rational, for each $\mmap\in\multiset{\Lambda}$ and
  $j\in[1,k]$, we can compute a finite automaton for the language
  $T_{j,\mmap}=\{w\in\Gamma^* \mid \mmap\in \wdcl{S_j}{w}\}$. Then
  $\fraku$ is $\frakS$-consistent if and only if the intersection
  $\bigcap_{j\in[1,k]}\bigcap_{\mmap\in U_j} T_{j,\mmap}$
  of regular languages is non-empty, which is clearly decidable.
\end{proof}

Moreover, because of \cref{stack-vector-rational},
\cref{abstraction-consistency} is a direct consequence of the
following.
\begin{restatable}{proposition}{abstractionConsistencyRational}\label{abstraction-consistency-rational}
  Given rational subsets
  $S_1,\ldots,S_k\subseteq\Gamma^*\times\multiset{\Lambda}$, we can
  compute a bound $B$ such that for the tuple $\frakS=(S_1,\ldots,S_k)$,
  the following holds: If $\frakV=(V_1,\ldots,V_k)$ is a tuple of
  \femph{finite} subsets $V_j\subseteq\multiset{\Lambda}$, then $\frakV$
  is $\frakS$-consistent if and only if $\alpha_B(\frakV)$ is
  $\frakS$-consistent.
\end{restatable}
Thus, it remains to prove \cref{abstraction-consistency-rational},
which is the purpose of the rest of this section.  Note that in
\cref{abstraction-consistency-rational}, the requirement that the
$V_j$ be \emph{finite} is crucial. For example, suppose $k=1$ and
$S=\{(\mathtt{a}^n, n\cdot\multi{\mathtt{b}}) \mid n\in\nats\}$ and
$\frakS=(S)$.  Then a set $V\subseteq\multiset{\{\mathtt{b}\}}$ is
$\frakS$-consistent if and only if $V$ is finite. Hence, there is no
bound $B$ such that $\alpha_B(V)$ reflects $\frakS$-consistency of any
$V$.
%
For \cref{abstraction-consistency-rational}, we use Ramsey's theorem (see \cref{ramsey}) to
prove the following pumping
  lemma.
\begin{lemma}\label{pump-dcl}
  Given a tuple $\frakS=(S_1,\ldots,S_k)$ of rational subsets
  $S_j\subseteq\Gamma^*\times\multiset{\Lambda}$, we can compute a
  bound $M$ such that the following holds.  In a word $w$ with $M$
  marked positions, we can pick two marked positions so that for the
  resulting decomposition $w=xyz$, we have
  $xyz\le_\frakS xy^\ell z$ for every $\ell\ge 1$.
\end{lemma}
\begin{proof}
  Let $\cM_j$ be an automaton for $S_j$ with state set $Q_j$ for
  $j\in[1,k]$.  We may assume that the sets $Q_j$ are pairwise
  disjoint and we define $Q=\bigcup_{j=1}^k Q_j$ and $n=|Q|$.  To each
  $u\in\Gamma^*$, we assign a subset
  $\kappa(u)\subseteq Q\times\powerset{\Lambda}$, where
  $(q,\Theta)\in Q_j\times\powerset{\Lambda}$ belongs to $\kappa(u)$
  if there is a cycle in $\cM_j$ that starts (and ends) in $q$, reads
  $u$, and reads a multiset with support $\Theta$. Hence,
  for $q\in Q_j$ and $\Theta\subseteq\Lambda$, we have
  $(q,\Theta)\in\kappa(u)$ if and only if $q\autstep[u|\mmap]q$ in
  $\cM_j$ for some $\mmap\in\multiset{\Lambda}$ with
  $\sup(\mmap)=\Theta$.

  We specify $M$ later. Suppose $M$ positions $s_1,\ldots,s_M$ are
  marked in $w$.  We build a colored graph on $M$ vertices and we
  label the edge from $j$ to $j'$ by the set $\kappa(u)$, where $u$ is
  the infix of $w$ between $s_j$ and $s_{j'}$.  Hence, the graph is
  $r$-colored, where $r=2^{|Q|\cdot 2^{|\Lambda|}}$ is the number of
  subsets of $Q\times\powerset{\Lambda}$.  We now apply Ramsey's
  theorem. We compute $M$ so that $M\ge R(r;n+1)$,
  e.g. $M=r^{r(n-1)+1}$. Then our graph must contain a monochromatic
  subset of size $n+1$. Let $t_1,\ldots,t_{n+1}$ be the corresponding
  positions in $w$. Moreover, let $w=xy_1\cdots y_{n}z$ be the
  decomposition of $w$ such that $y_j$ is the infix between $t_j$ and
  $t_{j+1}$. We claim that with $y=y_1\cdots y_n$, we have indeed
  $xyz\le_\frakS xy^\ell z$ for every $\ell\ge 1$.

  Consider a word $xy^\ell z$ and some multiset
  $\mmap\in \wdcl{S_j}{xyz}$.  We have to show that
  $\mmap\in \wdcl{S_j}{xy^\ell z}$. Since $\mmap\in\wdcl{S_j}{xyz}$,
  there is a run of $\cM_j$ reading $(xyz,\mmap')$ for some
  $\mmap'\succeq_1\mmap$. Since $\cM_j$ has $\le n$ states, some state
  must repeat at two borders of the decomposition $y=y_1\cdots
  y_n$. Suppose our run reads $(y_f\cdots y_g,\bar{\mmap})$ on a cycle on $q\in Q_j$
  for some $\bar{\mmap}$ in $\cM_j$. By monochromaticity, we
  know that $\kappa(y_f\cdots y_g)=\kappa(y_h)$ for every $h\in[1,n]$.
  Observe that we can write
 \begin{equation}
 xy^\ell z = x y_1\cdots y_{f-1} (y_f\cdots y_n y_1\cdots y_{f-1})^{\ell-1} y_f\cdots y_n z.
  \label{pump-decomposition}
  \end{equation}
  For every $h\in[1,n]$, we have $\kappa(y_f\cdots y_g)=\kappa(y_h)$,
  and hence $q\autstep[y_h|\mmap_h]q$ in $\cM_j$ for some
  $\mmap_h\in\multiset{\Lambda}$ with
  $\sup(\mmap_h)=\sup(\bar{\mmap})$. Then, in particular,
  $\sup(\mmap_h)\subseteq\sup(\mmap')=\sup(\mmap)$. Therefore,
  \cref{pump-decomposition} shows that $\cM_j$ accepts
  $(xy^\ell z,\mmap'+(\ell-1)\sum_{h=1}^n \mmap_h)$.  Since we now
  have $\mmap\preceq_1\mmap'+(\ell-1)\sum_{h=1}^n \mmap_h$, this
  implies $\mmap\in \wdcl{S}{xy^\ell z}$.
\end{proof}

Using \cref{pump-dcl}, we can obtain the final ingredient of
\cref{abstraction-consistency-rational}:
\begin{restatable}{lemma}{extendConsistency}\label{extend-consistency}
  Given a tuple $\frakS=(S_1,\ldots,S_k)$ of rational subsets
  $S_j\subseteq\Gamma^*\times\multiset{\Lambda}$, we can compute a
  bound $B$ such that the following holds. Let $\frakV=(V_1,\ldots,V_k)$
  be a $\frakS$-consistent tuple and suppose $\alpha_B(\mmap)\in
  V_j$. Then adding $\mmap$ to $V_j$ preserves $\frakS$-consistency.
\end{restatable}
The idea is the following. Let $\frakV'=(V'_1,\ldots,V'_k)$ be
obtained from $\frakV$ by adding $\mmap$ to $V_j$.  Let us say that
$w\in\Gamma^*$ \emph{covers} some $\nmap\in V'_j$ if and only if
$\nmap\in \wdcl{S_j}{w}$. Since $\frakV$ is $\frakS$-consistent, there
is a $w\in\Gamma^*$ that covers all elements of $\frakV$. Then $w$
covers $\alpha_B(\mmap)$. Moreover, $\mmap$ agrees with
$\alpha_B(\mmap)$ on all coordinates where $\mmap$ is $<B$.
We now have to construct a $w'$ that covers $\mmap$ in the remaining coordinates.
A simple pumping argument for each coordinate of $\mmap$ over an
automaton for $S_j$ (say, with $B$ larger than the number of states)
would yield a word $w'$ that even covers $\mmap$. However,
this might destroy coverage of all the other multisets in
$\frakV$. Therefore, we use \cref{pump-dcl}. It allows us to choose
$B$ high enough so that pumping to $w'=xy^\ell z$ covers $\mmap$, but
also guarantees $w\le_\frakS w'$. The latter implies that going from $w$ to
$w'$ does not lose any coverage. 

Finally, \cref{abstraction-consistency-rational} follows from \cref{extend-consistency} by induction. 

\section{Conclusion}
\label{sec:conclusion}

We have shown decidability of verifying liveness for $\DCPS$ in the context-bounded case.
Our results imply that fair termination for $\DCPS$ is $\Pi_1^0$-complete when
each thread is restricted to context switch a finite number of times.
Our result extends to liveness properties that can be expressed as a B\"uchi condition. 
We can reduce to fair non-termination by simply adding the states of a B\"uchi automaton to the global states via a product construction. 
From there, the B\"uchi acceptance condition can be simulated by using a special thread that forces a visit to a final state when scheduled, 
and then reposts itself before terminating. 
Scheduling this thread fairly along an infinite execution thus results in infinitely many visits to final states.

While we have focused on termination- and liveness-related questions, our techniques also imply further decidability results on
commonly studied decision questions for concurrent programs. 
A run of a $\DCPS$ is \emph{bounded} if 
there is a bound $B\in\nats$ so that 
the number of pending threads in every configuration along the run is at most $B$.
The $K$-bounded boundedness problem asks if every $K$-context bounded run is bounded.
Since boundedness is preserved under downward closures, our techniques for non-termination
can be modified to show the problem is also $\TWOEXPSPACE$-complete. 

The $K$-bounded \emph{configuration reachability} problem for $\DCPS$ asks if a given configuration is reachable.
Our reductions from $\DCPS$ to $\VASSB$, and the decidability of reachability for $\VASSB$, imply
$K$-bounded configuration reachability is decidable for $\DCPS$.

Thus, combined with previous results on safety verification \cite{AtigBQ2009j} and the case $K=0$ \cite{GantyM12}, 
our paper closes the decidability frontier for all commonly studied $K$-bounded verification problems 
for all $K\geq 0$.

\begin{acks}                            %
This research was sponsored in part by
the Deutsche Forschungsgemeinschaft project 389792660 TRR 248--CPEC
and by the European Research Council under the
Grant Agreement 610150 (http://www.impact-erc.eu/) (ERC Synergy Grant ImPACT).
\end{acks}

\newpage

\bibliography{bibliography}

\appendix

\def \cAtilde{\widetilde{\cA}}
\def \Deltatilde{\widetilde{\Delta}}
\def \Deltatildec{\widetilde{\Delta}_{\mathsf{c}}}
\def \Deltatildei{\widetilde{\Delta}_{\mathsf{i}}}
\def \Deltatilder{\widetilde{\Delta}_{\mathsf{r}}}
\def \Deltatildet{\widetilde{\Delta}_{\mathsf{t}}}
\def \rhotilde{\widetilde{\rho}}

\section{Strengthening Fairness  to Progressive Runs} \label{sec:strengthening-fairness}

In this section we strengthen the notion of fairness for $\DCPS$ to progressiveness by proving \cref{lem:progDCPS}.

\subsubsection*{Idea}
To prove \cref{lem:progDCPS} we modify the $\DCPS$ $\cA$ by giving every thread a bottom of stack symbol $\bot$ and saving its context switch number in its top of stack symbol. We also save this number in the global state whenever a thread is active. This way we can still swap a thread out and back in again once it has emptied its stack, and we also can keep track of how often we need to repeat that, before we reach $K$ context switches and allow it to terminate.

Furthermore, we also keep a subset $G'$ of the global states of $\cA$ in our new global states, which restricts the states that can appear when no thread is active. This way we can guess that a thread will be ``stuck'' in the future, upon which we terminate it instead (going up to $K$ context switches first) and also spawn a new thread keeping track of its top of stack symbol in the bag. Then later we restrict the subset $G'$ to only those global states that do not have {\sc Resume} rules for the top of stack symbols we saved in the bag. This then verifies our guess of ``being stuck''.

\subsubsection*{Formal construction}
Let $K \in \nats$ and $\cA = (G,\Gamma,\Delta,g_0,\gamma_0)$ be a $\DCPS$. We construct the $\DCPS$ $\cAtilde = (\widetilde{G},\widetilde{\Gamma},\Deltatilde,(g_0,G),\gamma_0)$, where
\begin{itemize}
  \item $\widetilde{G} = \big(G \times \powerset{G}\big) \cup \big(G \times \{0,\ldots,K\} \times \powerset{G}\big) \cup \big(\bar{G} \times \{0,\ldots,K\} \times \powerset{G}\big)$, where $\bar{G} = \{\bar{g} | g \in G\}$ and $\powerset{G}$ is the powerset of $G$, i.e.\ the set of all subsets of $G$,
  \item $\widetilde{\Gamma} = \Gamma_\bot \cup (\Gamma_\bot \times \{0,\ldots,K\}) \cup \bar{\Gamma}$, where $\Gamma_\bot = \Gamma \cup \{\bot\}$ and $\bar{\Gamma} = \{\bar{\gamma}|\gamma\in\Gamma\}$,
  \item $\Deltatilde = \Deltatildec \cup \Deltatildei \cup \Deltatilder \cup \Deltatildet$ consists of the following transitions rules:
  \begin{enumerate}
    \item $(g_1,G') \mapsto (g_2,G') \lhd \gamma \in \Deltatilder$ for all $G' \supseteq \{g_1\}$ iff $g_1 \mapsto g_2\lhd \gamma \in \Deltar$.
    \item $(g_1,G') \mapsto (g_2,k,G') \lhd (\gamma,k) \in \Deltatilder$ for all $k \in \{1,\ldots,K\}$, $G' \supseteq \{g_1\}$ iff $g_1 \mapsto g_2\lhd \gamma \in \Deltar$.
    \item $(g,G')|\gamma \hookrightarrow (g,0,G')|(\gamma,0).\bot \in \Deltatildec$ for all $g \in G$, $\gamma \in \Gamma$.
    \item $(g,k,G')|\gamma \hookrightarrow (g,k,G')|(\gamma,k) \in \Deltatildec$ for all $k \in \{0,\ldots,K\}$, $g \in G$, $\gamma \in \Gamma_\bot$.
    \item $(g_1,k,G')|(\gamma,k) \hookrightarrow (g_2,k,G')|w \in \Deltatildec$ for all $k \in \{0,\ldots,K\}$ iff $g_1|\gamma \hookrightarrow g_2|w \in \Deltac$.
    \item $(g_1,k,G')|(\gamma_1,k) \hookrightarrow (g_2,k,G')|w \triangleright \gamma_2 \in \Deltatildec$ for all $k \in \{0,\ldots,K\}$ iff $g_1|\gamma_1 \hookrightarrow g_2|w \triangleright \gamma_2 \in \Deltac$.
    \item $(g_1,k,G')|(\gamma_1,k) \mapsto (g_2,G')|(\gamma_2,k+1) \gamma_3 \in \Deltatildei$ for all $k \in \{0,\ldots,K-1\}$, $G' \supseteq \{g_2\}$ iff $g_1|\gamma_1 \mapsto g_2|\gamma_2 \gamma_3 \in \Deltai$, where $\gamma_2 \in \Gamma$ and $\gamma_3 \in \Gamma \cup \{\varepsilon\}$.
    \item $(g_1,k,G')|(\bot,k) \hookrightarrow (\bar{g}_2,k,G')|(\bot,k) \in \Deltatildec$ for all $k \in \{0,\ldots,K\}$ iff $g_1 \mapsto g_2 \in \Deltat$
    \item $(\bar{g},k,G')|(\bot,k) \mapsto (\bar{g},k+1,G')|(\bot,k+1) \in \Deltatildei$ for all $k \in \{0,\ldots,K-1\}$, $g \in G' \subseteq G$.
    \item $(\bar{g},k,G')|(\bot,k) \hookrightarrow (\bar{g},k,G')|(\bot,k) \in \Deltatildec$ for all $k \in \{0,\ldots,K-1\}$, $g \in G' \subseteq G$.
    \item $(\bar{g},k,G') \mapsto (\bar{g},k,G') \lhd (\bot,k) \in \Deltatilder$ for all $k \in \{0,\ldots,K\}$, $g \in G' \subseteq G$.
    \item $(\bar{g},K,G')|(\bot,K) \hookrightarrow (\bar{g},K,G')|\varepsilon \in \Deltatildec$ for all $g \in G' \subseteq G$.
    \item $(\bar{g},K,G') \mapsto (g,G') \in \Deltatildet$ for all $g \in G' \subseteq G$.
    \item $(g_1,k,G')|(\gamma,k) \hookrightarrow (\bar{g}_2,k,G')|\varepsilon \triangleright \bar{\gamma} \in \Deltatildec$ for all $k \in \{0,\ldots,K\}$, $g_2 \in G' \subseteq G$, $\gamma \in \Gamma$ iff $g_1|\gamma_1 \mapsto g_2|w \in \Deltai$ for some $w \in \Gamma^*$ and $g' \mapsto g \lhd \gamma \in \Deltar$ for some $g' \in G'$, $g \in G$. 
    \item $(g_1,k,G')|(\gamma,k) \hookrightarrow (\bar{g}_2,k,G')|\varepsilon \in \Deltatildec$ for all $k \in \{0,\ldots,K\}$, $g_2 \in G' \subseteq G$, $\gamma \in \Gamma$ iff $g_1|\gamma_1 \mapsto g_2|w \in \Deltai$ for some $w \in \Gamma^*$ and $g' \mapsto g \lhd \gamma \notin \Deltar$ for all $g' \in G'$, $g \in G$.
    \item $(\bar{g},k,G')|\gamma \hookrightarrow (\bar{g},k,G')|\varepsilon \in \Deltatildec$ for all $k \in \{0,\ldots,K\}$, $g \in G' \subseteq G$, $\gamma \in \Gamma$.
    \item $(\bar{g},k,G')|\bot \hookrightarrow (\bar{g},k,G')|(\bot,k) \in \Deltatildec$ for all $k \in \{0,\ldots,K\}$, $g \in G' \subseteq G$.
    \item $(g,G') \mapsto (g,0,G') \lhd \bar{\gamma} \in \Deltatilder$ for all $k \in \{1,\ldots,K\}$, $g \in G' \subseteq G$, $\gamma \in \Gamma$.
    \item $(g,0,G_1)|\bar{\gamma} \hookrightarrow (\bar{g},0,G_2)|(\bot,0) \in \Deltatildec$ for all $g \in G_2 \subseteq G_1 \subseteq G$, $\gamma \in \Gamma$, where $G_2 = \{g_1 \in G_1 | \forall g_2 \in G\colon g_1 \mapsto g_2 \lhd \gamma \notin \Deltar\}$.
  \end{enumerate}
\end{itemize}
We proceed by constructing a run of $\cAtilde$ from a fair run of $\cA$ and argue that this construction results in a progressive run and preserves starvation, which proves the if direction of the two points of \cref{lem:progDCPS}. For the only if direction we argue that we can also do this backwards, starting with a progressive run from $\cAtilde$ and constructing one from $\cA$.

Let $\rho$ be a infinite fair run of $\cA$. We begin by formalizing the notion threads reaching a certain local configuration and never progressing from there. Consider a local configuration $t = (w,i)$ of $\cA$ that over the course of $\rho$ is only removed from the bag finitely often (via applications of {\sc Resume}) and is added to the bag more often than it is removed (via applications of {\sc Swap}). Let $n_t$ be the number of times $t$ is removed from the bag in this way. Then from the $(n+1)$th applications of {\sc Swap} and onwards adding $t$ to the bag over the course of $\rho$, we say the local configuration $t$ added to the bag is \emph{stagnant}.

Now we can properly construct an infinite run $\rhotilde$ of $\cAtilde$. Whenever $\rho$ reaches a configuration with global state $g$, $\rhotilde$ mimics it by reaching a global state that includes $g$ in the state tuple. Similarly any non-stagnant local configuration $(\gamma w,i)$ occurring on $\rho$ corresponds to a local configuration $\big((\gamma,i) w,i\big)$ occurring on $\rhotilde$, whereas the empty stack corresponds to stack content $\bot$ or $(\bot,i)$. Newly spawned local configurations $(\gamma,0)$ look the same in both runs. For the initial configuration $\langle g_0, \bot, \multi{(\gamma_0,0)}\rangle$ of $\rho$ and $\langle (g_0,G), \bot, \multi{(\gamma_0,0)}\rangle$ of $\rhotilde$ this correspondence evidently holds. Let us now go over the consecutive thread step and scheduler step relations of $\rho$ and construct $\rhotilde$ appropriately:

In the following, for any step relation between two configurations of $\cA$, we write a number 1 to 17 above the arrow to denote which rule of $\Deltatilde$ is being applied.
\begin{description}
  \item [Case $\langle g_1, \#, \mmap + \multi{(\gamma w, i)} \rangle \mapsto \langle g_2, (\gamma w, i), \mmap \rangle$ due to $g_1 \mapsto g_2 \lhd \gamma \in \Deltar$:]\ \\
  If $i = 0$ (and therefore also $w = \varepsilon$) then
  \[\langle (g_1,G'), \#, \widetilde{\mmap} + \multi{(\gamma, 0)} \rangle \xmapsto{\,1\,} \langle (g_2,G'), (\gamma,0), \widetilde{\mmap} \rangle \xrightarrow{3} \langle (g_2,0,G'), \big((\gamma,0) \bot,0\big), \widetilde{\mmap} \rangle,\]
  otherwise (if $i \geq 1$)
  \[\langle (g_1,G'), \#, \widetilde{\mmap} + \multi{\big((\gamma,i) w \bot, i\big)} \rangle \xmapsto{\,2\,} \langle (g_2,i,G'), \big((\gamma,i) w \bot, i\big), \widetilde{\mmap} \rangle.\]
  
  \item [Case $\langle g_1, (\gamma w,i), \mmap \rangle \rightarrow \langle g_2, (w' w,i), \mmap \rangle$ due to $g_1|\gamma \hookrightarrow g_2|w' \in \Deltac$:]%
  \[\langle (g_1,i,G'), \big((\gamma,i) w \bot,i\big), \widetilde{\mmap} \rangle \xrightarrow{5} \langle (g_2,i,G'), (w' w \bot,i), \widetilde{\mmap} \rangle,\]
  from here, if $w' w = \gamma' w'' \neq \varepsilon$ for some $\gamma' \in \Gamma$, $\rhotilde$ continues with
  \[\langle (g_2,i,G'), (\gamma' w'' \bot,i), \widetilde{\mmap} \rangle \xrightarrow{4} \langle (g_2,i,G'), \big((\gamma',i) w'' \bot,i\big), \widetilde{\mmap} \rangle,\]
  otherwise (if $w' w = \varepsilon$) $\rhotilde$ continues with
  \[\langle (g_2,i,G'), (\bot,i), \widetilde{\mmap} \rangle \xrightarrow{4} \langle (g_2,i,G'), \big((\bot,i),i\big), \widetilde{\mmap} \rangle.\]
  
  \item [Case $\langle g_1, (\gamma w,i), \mmap \rangle \rightarrow \langle g_2, (w' w,i), \mmap + \multi{(\gamma',0)} \rangle$ due to $g_1|\gamma \hookrightarrow g_2|w' \triangleright \gamma' \in \Deltac$:]%
  \[\langle (g_1,i,G'), \big((\gamma,i) w \bot,i\big), \widetilde{\mmap} \rangle \xrightarrow{6} \langle (g_2,i,G'), (w' w \bot,i), \widetilde{\mmap} + \multi{(\gamma',0)} \rangle,\]
  if $w' w = \gamma'' w'' \neq \varepsilon$ for some $\gamma'' \in \Gamma$ then $\rhotilde$ continues with
  \[\langle (g_2,i,G'), (\gamma'' w'' \bot,i), \widetilde{\mmap} + \multi{(\gamma',0)} \rangle \xrightarrow{4} \langle (g_2,i,G'), \big((\gamma'',i) w'' \bot,i\big), \widetilde{\mmap} + \multi{(\gamma',0)} \rangle,\]
  otherwise (if $w' w = \varepsilon$) $\rhotilde$ continues with
  \[\langle (g_2,i,G'), (\bot,i), \widetilde{\mmap} + \multi{(\gamma',0)} \rangle \xrightarrow{4} \langle (g_2,i,G'), \big((\bot,i),i\big), \widetilde{\mmap} + \multi{(\gamma',0)} \rangle.\]
  
  \item [Case $\langle g_1, (\gamma w, i), \mmap \rangle \mapsto \langle g_2, \#, \mmap + \multi{(w' w, i+1)} \rangle$ due to $g_1|\gamma \mapsto g_2|w' \in \Deltai$:]\ \\
  By definition of $\Deltai$ we have $1 \leq |w'| \leq 2$ and therefore $w' = \gamma_1 \gamma_2$ for some $\gamma_1 \in \Gamma$, $\gamma_2 \in \Gamma \cup \{\varepsilon\}$. If $(w' w, i+1)$ is not stagnant here, then
  \[\langle (g_1,i,G'), \big((\gamma,i) w \bot,i\big), \widetilde{\mmap} \rangle \xmapsto{\,7\,} \langle (g_2,G'), \#, i\big), \widetilde{\mmap} + \multi{\big((\gamma_1,i+1) \gamma_2 w \bot, i+1\big)} \rangle.\]
  Otherwise (if $(w' w, i+1)$ is stagnant here), if $G'$ still contains states that allow {\sc Resume} rules on $\gamma_1$, we first save this top of stack symbol in the bag (as $\bar{\gamma}_1$):
  \[\langle (g_1,i,G'), \big((\gamma,i) w \bot,i\big), \widetilde{\mmap} \rangle \xrightarrow{14} \langle (\bar{g}_2,i,G'), (w \bot,i), \widetilde{\mmap} + \multi{(\bar{\gamma}_1,0)} \rangle.\]
  If $G'$ does not allow for any {\sc Resume} rules on $\gamma_1$ we do not save this top of stack symbol:
  \[\langle (g_1,i,G'), \big((\gamma,i) w \bot,i\big), \widetilde{\mmap} \rangle \xrightarrow{15} \langle (\bar{g}_2,i,G'), (w \bot,i), \widetilde{\mmap} \rangle.\]
  In both of these stagnant cases we continue by almost emptying the stack of the active thread in $\rhotilde$. Let $\widetilde{\mmap}' = \widetilde{\mmap} + \multi{(\bar{\gamma}_1,0)}$ in the first case and $\widetilde{\mmap}' = \widetilde{\mmap}$ in the second case:
  \[\langle (\bar{g}_2,i,G'), (w \bot,i), \widetilde{\mmap}' \rangle \xrightarrow{16} \ldots \xrightarrow{16} \langle (\bar{g}_2,i,G'), (\bot,i), \widetilde{\mmap}' \rangle \xrightarrow{17} \langle (\bar{g}_2,i,G'), \big((\bot,i),i\big), \widetilde{\mmap}' \rangle.\]
  From here $\rhotilde$ continues with what we call an \emph{extended thread termination}, where we force the active thread to make its remaining context switches (up to $K$) and then terminate:
  \begin{align*}
  \langle (\bar{g}_2,i,G'), \big((\bot,i),i\big), \widetilde{\mmap}' \rangle &\xmapsto{\,9\,} \langle (\bar{g}_2,i+1,G'), \#, \widetilde{\mmap}' + \multi{\big((\bot,i+1),i+1\big)} \rangle \\
  &\xmapsto{11} \langle (\bar{g}_2,i+1,G'), \big((\bot,i+1),i+1\big), \widetilde{\mmap}' \rangle \\
  &\xrightarrow{10} \langle (\bar{g}_2,i+1,G'), \big((\bot,i+1),i+1\big), \widetilde{\mmap}' \rangle \\
  &\;\;\;\vdots\;\;\;\big((K - i)\text{ repetitions of the rule sequence (9), (11), (10)}\big) \\
  &\xrightarrow{10} \langle (\bar{g}_2,K,G'), \big((\bot,K),K\big), \widetilde{\mmap}' \rangle \\
  &\xrightarrow{12} \langle (\bar{g}_2,K,G'), (\varepsilon,K), \widetilde{\mmap}' \rangle \\
  &\xmapsto{13} \langle (g_2,G'), \#, \widetilde{\mmap}' \rangle.
  \end{align*}
  Of course if $i = K$ only the last two steps are part of $\rhotilde$. We also note that rule (10) is only necessary to facilitate alternation between thread step and scheduler step relations, as required by runs of $\DCPS$.
  
  \item [Case $\langle g_1, (\varepsilon, i), \mmap \rangle \mapsto \langle g_2, \#, \mmap \rangle$ due to $g_1 \mapsto g_2 \in \Deltat$:]%
  \[\langle (g_1,i,G'), \big((\bot,i),i\big), \widetilde{\mmap} \rangle \xrightarrow{8} \langle (\bar{g}_2,i,G'), \big((\bot,i),i\big), \widetilde{\mmap}' \rangle,\]
  from here $\rhotilde$ continues with an extended thread termination, as explained previously.
  
  \item [Once a particular configuration $\langle g, \#, \mmap \rangle$ is reached:]\ \\
  Let $\Ginf \subseteq G$ be the set of global states that occur infinitely often on configurations of $\rho$ with no active thread. Let $\langle g, \#, \mmap \rangle$ be a configuration upon which only global states in $\Ginf$ occur on configurations with no active thread on $\rho$ (including for this configuration itself, meaning $g \in \Ginf$). Furthermore, no new top of stack symbols appear among stagnant threads after this configuration has occurred. Now we handle all the threads with top of stack symbol in $\bar{\Gamma}$ spawned over the course of $\rhotilde$:
  \[\langle (g,G_1), \#, \widetilde{\mmap} + \multi{(\bar{\gamma},0)} \rangle \xmapsto{18} \langle (g,0,G_1), (\bar{\gamma},0), \widetilde{\mmap} \rangle \xrightarrow{19} \langle (\bar{g},0,G_2), \big((\bot,0),0\big), \widetilde{\mmap} \rangle.\]
  Here $G_2 = \{g_1 \in G_1 | \forall g_2 \in G\colon g_1 \mapsto g_2 \lhd \gamma \notin \Deltar\}$, i.e.\ the subset of $G_1$ that no longer contains any global states that allow {\sc Resume} rules for $\gamma$. We continue with an extended thread termination, upon which we repeat these steps until no more local configurations with top of stack symbol in $\bar{\Gamma}$ are in the bag.
\end{description}
The run $\rhotilde$ is sound: Firstly, the correspondence of local configurations we mentioned earlier is upheld throughout the run. Secondly, we always save the cs-number of the active thread in the state tuple and we have no cs-number in the state tuple if there is no active thread, which we assumed to hold at the start of all step sequences we added to $\rhotilde$. Finally, the subset $G' \subseteq G$ in each state tuple restricts which global states can occur while there is no active thread, but we update it in such a way that none of the steps we want to make are disallowed: As long as we visit global states $\notin \Ginf$ in configurations with no active thread $G' = G$ holds. After such visits stop happening, we update $G'$ based on all the stagnant threads that occurred. However, the top of stack symbols of stagnant threads cannot allow for the application of {\sc Resume} rules from global states in $\Ginf$, because that would violate the fairness condition of $\rho$. Therefore $\Ginf \subseteq G'$ still holds after we finish updating $G'$, which means all future steps are still allowed by $G'$.

The run $\rhotilde$ is also progressive: Firstly, every local configuration on $\rhotilde$ that corresponds to a non-stagnant thread of $\rho$ is resumed at some point, because its correspondent was handled fairly on $\rho$, which implies resumption for non-stagnant threads. Secondly, the local configurations of $\rhotilde$ corresponding to stagnant threads of $\rho$ do not have to be resumed, since they always terminate. Thirdly, all threads with top of stack in $\bar{\Gamma}$ are resumed eventually: Since we only update $G'$ once all different types of these threads were spawned, the update disallows any {\sc Resume} rules for all top of stack symbols of stagnant threads occurring afterwards. Therefore the only transition rule in $\Deltatildec$ that spawns these threads (rule 14) can no longer be applied (we apply rule 15 instead). Since these threads then no longer occur in the bag afterwards, they are of course do not have to be resumed. Finally, every thread that terminates does so after exactly $K$ context switches: This is due to the extended thread termination that is performed every time the bottom of stack symbol $\bot$ is reached.

A thread is starved by $\rhotilde$ iff a thread is starved by $\rho$: It is clear that stagnant threads cannot be starved, and non-stagnant threads of $\rho$ have their correspondents handled in same way on $\rhotilde$. The only mismatch in starvation can thus occur on the threads with top of stack symbol in $\bar{\Gamma}$. However, we already argued that after a certain point, none of these threads even occur as part of a configuration of $\rhotilde$ anymore. Therefore they cannot be starved either. This concludes the if-direction of the proof.

For the \emph{only if} direction let $\rhotilde$ be an infinite progressive run of $\cAtilde$. We observe that we can do most of the previous construction backwards to obtain an infinite run $\rho$ of $\cA$. Whenever we guess a thread to be stagnant by using rule (14) or (15) on $\rhotilde$, we instead keep its correspondent around in $\rho$ forever. Since each local configuration with top of stack symbol in $\bar{\Gamma}$ has to have a {\sc Resume} rule applied to it on $\rhotilde$, we eventually are restricted to global states that do not allow for {\sc Resume} rules on the corresponding top of stack symbols in $\Gamma$. This means the threads we keep around forever on $\rho$ are ``stuck'' and therefore handled fairly. Restricting the global states based on some symbol $\bar{\gamma}$ also disallows new threads with the same symbol to spawn during $\rhotilde$ afterwards. This means at some point no new threads with top of stack symbol in $\bar{\Gamma}$ are newly added to the bag, but the old threads disappear eventually because of progressiveness. Thus these threads cannot be exhibit starvation that then would not be present in $\rho$. Finally, all remaining threads of $\rho$ are handled the same way as their correspondents of $\rhotilde$, meaning the progressiveness of the latter implies fairness of the former.

\section{Proofs from Section~\ref{sec:vassb}} \label{appendix:dcps-to-vassb}

In this section we prove \cref{th:dcps-to-vassb} formally, which involves constructing a $\VASSB$ $\cV$ with the same progressiveness properties (when starting in configuration $c_0 = (q_0,\emptyset,\emptyset)$) as a given $\DCPS$ $\cA$.

\subsubsection*{Formal construction}
Let $K \in \nats$ and $\cA = (G,\Gamma,\Delta,g_0,\gamma_0)$ be a $\DCPS$ with $\Gamma = \{\gamma_0,\ldots,\gamma_l\}$. Construct a $\VASSB$ $\cV = (Q,P,\bar{Q},\bar{P},E)$, where
\begin{itemize}
  \item $\bar{Q} = \big(\cT(\cA,K) \times \{0,\ldots,K\} \times \mathsf{SL}(\cA,K) \times (\Gamma \cup \{\varepsilon\})\big) \cup \{\bot\}$,
  \item $\bar{P} = \Gamma \times \{0,\ldots,K\}$,
  \item $Q = G \cup (G \times \bar{Q} \times G) \cup \{q_0\}$,
  \item $P = \Gamma$,
  
  \item $E$ contains the following edges of the form $q\autstep[\op]q'$:
  \begin{enumerate}
    \item $\op=\delta\in\integs^P$ with $\delta(p) = 0$ for all $p \in P$, iff $q = q_0$, and $q' = (g_0,\bar{q},g_1)$ with $\bar{q} = (t,0,\mathsf{sl}(t),\varepsilon)$ where $t \in \cT(\cA,K)$ with a first segment going from $g_0$ to $g_1$. \label{VASSwBitem:init}
    \item $\op=\delta\in\integs^P$ with $\delta(\gamma) = -1$ and $\delta(p) = 0$ for all $p \neq \gamma$, iff there is a rule $g_1\mapsto g_2\lhd \gamma \in \Deltar$, $q = g_1 \in G$, and $q' = (g_2,\bar{q},g_3)$ with $\bar{q} = (t,0,\mathsf{sl}(t),\varepsilon)$ where $t \in \cT(\cA,K)$ with a first segment going from $g_2$ to $g_3$. \label{VASSwBitem:newT}
    \item $\op=\newb(\bar{q},S)$, iff $q = (g_1,\bar{q},g_2)$ where $\bar{q} = (t,0,S,\varepsilon)$, and $q' = (g_1,\bar{q}',g_2)$ where $\bar{q}' = (t,0,S,\gamma_0)$, for some $t \in \cT(\cA,K)$, $S \in \mathsf{SL}(\cA,K)$, $g_1,g_2 \in G$. \label{VASSwBitem:newB}
    \item $\op=\deflateb(\bar{q},\bar{q}', \bar{p},p)$, iff $p = \gamma_i$, $\bar{p} = (\gamma_i,j)$, $\bar{q} = (t,j,S,\gamma_i)$, $\bar{q}' = (t,j,S,\gamma_{i+1})$, $q = (g_1,\bar{q},g_2)$, and $q' = (g_1,\bar{q}',g_2)$, for some $j \in \{0,\ldots,K\}$, $t \in \cT(\cA,K)$, $S \in \mathsf{SL}(\cA,K)$, $i \in \{0,\ldots,l-1\}$, $g_1,g_2 \in G$. \label{VASSwBitem:popmid}
    \item $\op=\deflateb(\bar{q},\bar{q}', \bar{p},p)$, iff $p = \gamma_l$, $\bar{p} = (\gamma_l,j)$, $\bar{q} = (t,j,S,\gamma_l)$, $\bar{q}' = (t,j+1,S,\varepsilon)$, $q = (g_1,\bar{q},g_2)$, and $q' = g_2$, for some $j \in \{0,\ldots,K-1\}$, $t \in \cT(\cA,K)$, $S \in \mathsf{SL}(\cA,K)$, $g_1,g_2 \in G$. \label{VASSwBitem:popend}
    \item $\op=\deflateb(\bar{q},\bar{q}', \bar{p},p)$, iff $p = \gamma_l$, $\bar{p} = (\gamma_l,K)$, $\bar{q} = (t,K,S,\gamma_l)$, $\bar{q}' = \bot$, $q = (g_1,\bar{q},g_2)$, and $q' = (g_1,\bot,g_2)$, for some $t \in \cT(\cA,K)$, $S \in \mathsf{SL}(\cA,K)$, $g_1,g_2 \in G$. \label{VASSwBitem:popendlimit}
    \item $\op=\burstb(\bot)$ iff $q = (g_1,\bot,g_2)$ and $q=g_2$ for some $g_1,g_2 \in G$. \label{VASSwBitem:float}
    \item $\op=\delta\in\integs^P$ with $\delta(p) = 0$ for all $p \in P$, iff there is a rule $g_1\mapsto g_2\lhd \gamma \in \Deltar$, $q = g_1 \in G$ and $q' = (g_2,\bar{q},g_3)$ with $\bar{q} = (t,j,S,\gamma_0)$ where $t \in \cT(\cA,K)$ with a $j$th context switch on symbol $\gamma$ followed by a $j$th segment going from $g_2$ to $g_3$, for some $S \in \mathsf{SL}(\cA,K)$, $j \in \{1,\ldots,K\}$. \label{VASSwBitem:resB}
  \end{enumerate}
\end{itemize}
Intuitively, the edge in (\ref{VASSwBitem:init}) spawns the initial thread, the edges in (\ref{VASSwBitem:newT}) remove a thread with cs-number $0$ when it becomes active for the first time and the edges in (\ref{VASSwBitem:newB}) convert it to a balloon, the edges in (\ref{VASSwBitem:popmid}) apply the spawns during the $j$th segment of a thread by transferring them from the balloon to the regular places, the edges in (\ref{VASSwBitem:popend}) and (\ref{VASSwBitem:popendlimit}) finalize this application process with the edges in (\ref{VASSwBitem:popendlimit}) handling the special case of the $K$th segment of a thread, the edges in (\ref{VASSwBitem:float}) burst balloons containing no more spawns, and the edges in (\ref{VASSwBitem:resB}) initialize the application of spawns for a thread with cs-number $\geq 1$ that is already being represented by a balloon.

\begin{lemma}
  The $\DCPS$ $\cA$ has an infinite, progressive, $K$-context switch bounded run iff the $\VASSB$ $\cV$ has an infinite progressive run.
\end{lemma}
\begin{proof}
  For the \emph{only if direction} let $\rho$ be an infinite, progressive, $K$-context switch bounded run of $\cA$. We want to decompose $\rho$ into parts corresponding to a single thread each. To do this we look at the continuous \emph{segments} of $\rho$ where a thread is active, starting from the first configuration with an active thread and ending each segment at a configuration with no active thread, upon which a new segment begins at the very next configuration. We can now group segments together into \emph{executions} of single threads:
  \begin{itemize}
    \item Start with the earliest not yet grouped segment where the active thread has context switch number $0$. Note the local configuration that gets added to the bag at the end of this segment.
    \item Take the very next ungrouped segment in $\rho$ where the active thread at the beginning matches the previously noted local configuration. Again, note the local configuration that gets added to the bag at the end of this segment.
    \item Repeat the previous step until an empty stack occurs.
  \end{itemize}
  Repeating this infinitely often decomposes $\rho$ into infinitely many thread executions of $K+1$ segments each. This is because due to progressiveness all local configurations have to be resumed and all thread terminations occur after exactly $K$ context switches. Each execution also corresponds to a thread type $t$, which can be determined by taking note of the global state and top of stack symbol at the start of each segment and the global state at the end of each segment of the execution. Since all executions end in the empty stack, $L_t \neq \varnothing$ for all thread types $t$ that occur on $\rho$ in this way. Let us now construct from $\rho$ an infinite run of $\cV$ starting from $c_0 = (q_0,\emptyset,\emptyset)$:
  \begin{itemize}
    \item Start with the edge from (\ref{VASSwBitem:init}) to create a token on the place corresponding to the initial stack symbol $\gamma_0$.
    \item Go over all consecutive segments of $\rho$. If a segment is the first segment of an execution:
    \begin{itemize}
      \item Start with an edge from (\ref{VASSwBitem:newT}) to remove a token from place $\gamma$, which is the top of stack symbol this segment starts with. Use as $g_2$ the global state this segment starts with, as $g_3$ the global state this segment ends in, and as $t$ the thread type this segment's execution corresponds to. This leads to a state of $\cV$ that contains the semi-linear set $s(t)$.
      \item Next, take an edge from (\ref{VASSwBitem:newB}) to create a balloon whose contents characterize exactly the spawns made by this segment's execution. This is possible because we kept track of $s(t)$ in the state, and the appropriate balloon contents are in this semi-linear set by construction.
      \item From here we need to deflate the balloon $|\Gamma|$ times to transfer all the spawns from this segment to the places of $\cV$. This is done using the edges from (\ref{VASSwBitem:popmid}) for each $i$ from $0$ to $l-1$ on the balloon places $(\gamma_0,0)$ to $(\gamma_{l-1},0)$. The $0$ as the second component is because this is the first segment of an execution.
      \item If this is not the last segment of its execution, we use an edge from (\ref{VASSwBitem:popend}) to perform the last deflation on the balloon place $(\gamma_l,0)$ and go to the global state that this segment ends in, using that state as $g_2$.
      \item If it is the last segment of its execution we instead use an edge from (\ref{VASSwBitem:popendlimit}) perform the same deflation, but go to a state $(g_1,\bot,g_2)$ instead. From there we use the edge from (\ref{VASSwBitem:float}) to burst the balloon and finally go to state $g_2$, which this segment ends in.
    \end{itemize}
    \item If a segment is not the first segment of its execution:
    \begin{itemize}
      \item Start with an edge from (\ref{VASSwBitem:resB}) using as $g_2$ the global state this segment starts with, as $g_3$ the global state this segment ends in, and as $\bar{q}$ the balloon state we ended in after handling the previous segment of this execution.
      \item From here continue like for a first segment of an execution, deflating the balloon $|\Gamma|$ times in the next few steps. Note that we operate on balloon places $\Gamma \times \{j\}$, where this is the $j$th segment within its execution, instead of $\Gamma \times \{0\}$.
    \end{itemize}
  \end{itemize}
  Since each thread execution of $\rho$ has exactly $K$ segments, any balloon will eventually reach the $K$th segment while going over $\rho$. This results in an edge from (\ref{VASSwBitem:float}) being used to burst the balloon. Therefore any balloon state occurring in a configuration of $\cV$ on the constructed run will be used for a deflate or burst later. Furthermore, each place $\gamma \in \Gamma$ corresponds to a local configuration $(\gamma,0)$ of $\cA$, and we remove a token from it whenever we switch in a thread with that local configuration. Since such local configurations always have to be resumed on $\rho$, the corresponding places will always have tokens removed from them as well. Thus the handling of all balloons and places fulfils the conditions for a progressive run.
  
  The constructed run of $\cV$ is also sound: For each segment of $\rho$ it has a sequence of transitions that make the same state change as that segment, and also apply the same spawns. The latter is due to initially choosing the contents during the creation of a balloon correctly. A token representing the initial thread is also created in the beginning of the run. Since the transitions from one segment to the next require the same resumption rules $\Deltar$ as in $\rho$ (see edges (\ref{VASSwBitem:newT}) and (\ref{VASSwBitem:resB})) these are also possible in $\cV$.
  
  For the \emph{if direction} we do a very similar construction, but backwards. Starting with an infinite progressive run of $\cV$ from $c_0$, we decompose it into segments that end in a state in $G$. Afterwards we group those segments together by matching the local balloon configurations in the same way we did for local thread configurations before. Then by construction of $L_t$ we know that there is a thread execution of type $t$ that makes the same spawns in each segment as the group of one balloon does. Those executions will all terminate after $K$ context switches, guaranteeing progressiveness, and we can interleave them in the same way as the balloon segments to form an infinite progressive run of $\cA$.
\end{proof}

\section{Proofs for Section \ref{sec:configReach} }
\label{appendix-sec-five-ndag}
\def\sset{\mathsf{set}}

\paragraph{Balloons with Identity}
In order to prove Lemma \ref{new-witness-prog}, we will work with runs which
contain a unique identity for every balloon i.e. with 
\emph{balloons-with-id}. Formally,
a balloon-with-id is a tuple $(b,i)$ where $b$ is a balloon and $i
\in \nats$. A configuration-with-id $d$ is a tuple $d=(q,\mmap,\nu)$
where
$\nu$ is a multiset of balloons-with-id. For a $\VASSB$ $\cV=
(Q,P,\bbq,\bbp,E)$, the edges in $E$ define a transition relation on
configurations-with-id. For an edge $q\autstep[\op]q'$, and
configurations-with-id
$d=(q,\mmap,\nmap)$ and $d'=(q',\mmap',\nmap')$, we define
$d \autstep[\op] d'$:
\begin{itemize}
\item If $\op=\delta\in\integs^P$ and 
$\mmap'=\mmap+\delta$ and $\nmap'=\nmap$.
  
\item If $\op=\newb(\sigma,S)$ and
  $\mmap'=\mmap$ and $\nmap'=\nmap+\multi{((\sigma,\kmap),i)}$ for some
  $\kmap\in S$ and $i \in \nats$.  
  That is, we create a new balloon with state
  $\sigma$, multiset $\kmap$ for some $\kmap\in S$ and arbitrary id
  $i$.
  
\item If $\op=\deflateb(\sigma,\sigma',\pi,p)$ and 
  there is a balloon-with-id $(b,i)=((\sigma,\kmap),i)\in \bbq\times
  \multiset{\bbp}$ with $\nmap(b,i)\ge 1$ and
  $\mmap'=\mmap+\kmap(\pi)\cdot\multi{p}$ and
  $\nmap'=\nmap-\multi{(b,i)}+\multi{((\sigma',\kmap'),i)}$,
  where $\kmap'(\pi) = 0$ and $\kmap'(\pi') =\kmap(\pi')$ for all $\pi'\in\bbp\setminus\set{\pi}$.
  That is, we pick a balloon-with-id $((\sigma,\kmap),i)$ from 
  $\nmap$, transfer the contents in place $\pi$ from $\kmap$
  to place $p$ in $\mmap$, and update the state $\sigma$ to $\sigma'$
  while retaining its id.
  Here we say the balloon-with-id $((\sigma,\kmap),i)$ was 
  \emph{deflated}.
  
\item If $\op=\burstb(\sigma)$ and
  there is a balloon-with-id $(b,i)=((\sigma,\kmap),i)\in \bbq\times
  \multiset{\bbp} \times \nats$ with $\nmap(b,i) \geq 1$ and
  $\mmap'=\mmap$ and $\nmap' = \nmap \ominus \multi{(b,i)}$. 
  This means we pick some balloon-with-id $(b,i)$ with state $\sigma$
  from our multiset $\nmap$ of balloons-with-id and remove it,
  making any tokens still contained in its balloon places disappear as well.
  Here we say the balloon-with-id $(b,i)$ is \emph{burst}.
\end{itemize}
We often simply say `the balloon $i$' was burst (or deflated) when
balloon
identities are unique.
A \emph{run-with-id} $\tau=d_0 \xrightarrow{\op_1} d_1 
\xrightarrow{\op_1} d_2 \cdots$ is a finite or infinite sequence of
configurations-with-id. We note that a
semiconfiguration-with-id is just a semiconfiguration since there are
no balloons.

\begin{definition}
	We associate a \emph{canonical run-with-id} $\tau=d_0 \xrightarrow{\op_1} d_1 
\xrightarrow{\op_1} d_2 \cdots$ to any given run $\rho=c_0 
\xrightarrow{\op_1} c_1 \xrightarrow{\op_2} \cdots$ as follows:
\begin{itemize}
	\item If $\op_i$ creates a balloon $b$ in $\rho$, then it creates $
	(b,i)$ in $\tau$. This implies that balloons are assigned unique id's
	since every operation creates only one balloon. 
	\item If $\op_i$ deflates (resp. bursts) a balloon $b$ in $\rho$,
	then it deflates (resp. bursts)
	the balloon-with-id $(b,j)$ in $\tau$ where $j=min\{ k \mid d_
	{i-1}.\nu(b,k)
	\ge 1\}$.
\end{itemize}
\end{definition}
We observe that the multiset $d_i.\nu$ of any configuration-with-id
$d_i$ in a canonical run-with-id $\tau$ is infact a set. The set of
id's in $\tau$ is denoted $I(\tau)$. A collection of
balloons-with-id with the same id $i_0$ may be viewed as a particular
balloon
which
undergoes a sequence of operations $\seq_{i_1}=\op_{i_0},\op_
{i_1},\op_{i_2},
\cdots$.
The balloons $
(b_j,i_0)$
resulting from the operations $\op_{i_j}$ are associated with the id
$i_0$ and we will perform surgery on a run-with-id $\tau$ resulting in
a modified run $\tau'$ by
replacing the balloon $(b_1,i_0)$ at its point of inflation $\op_
{i_0}$ by $(b'_1,i_0)$ in $\tau'$, with the
implicit assumption that the sequence of operations $\seq'_{i_0}$ in
the modified run-with-id $\tau'$ now act on $(b'_1,i_0)$ instead. Note
that we use the notation $\seq'_i$ to represent the sequence of
operations on id $i$ occuring in the modified run $\tau'$.
The
sequences $\seq_{i_0}$ and $\seq_{j_0}$ are disjoint for $i_0 \neq
j_0$ and thus the set of sequences $\{ \seq_i \mid i \in \tau\}$ form
a partition
of the indices $I_B(\tau)=\{ i \mid \op_i \text{ is a balloon
operation }\}$.
We will write $i.\barq,i.\kmap$ in short for $(b,i).\barq$ and $
(b,i).\kmap$ where $(b,i)$ is created by $\op_i$. The linear set
used to create a balloon with the id $i$ is denoted $L_i$.
A run-with-id $\tau$ \emph{corresponds} to a run $\rho$ if $\rho$ is
obtained from $\tau$ by removing id's. For a balloon $b$ in $\rho$, the
set $\{
(b,i) \mid i \in I(\tau) \}$ of balloons-with-id in a corresponding
$\tau$ is called the
\emph{balloon class} of $b$. 

The fact that reachability remains preserved whether considering runs
or runs-with-id immediately follows from the definition of a
canonical run-with-id:
\begin{proposition}
\label{prop-runs-with-id}
	Given $\cV$ and its semiconfigurations $s_1,s_2$, there exists a run
	$\rho=s_1 \autstep[*] s_2$ of $\cV$ iff there exists a
	corresponding run-with-id
	$\tau=s_1 \autstep[*] s_2$.
\end{proposition}
The fact that progressivenss is also preserved is also not difficult
to see. We first define the obvious notion of progressiveness for a
run-with-id.
\begin{definition}
	A progressive run-with-id $\tau=d_0 \xrightarrow{\op_1} d_1 
\xrightarrow{\op_1} d_2 \cdots$ is one such that:
\begin{itemize}
 	\item for every
	balloon-with-id $(b,i) \in d_j$ for some $d_j \in \tau$, there exists
	$d_k$ with $k>j$ such that $d_k \xrightarrow{\op_k} d_{k+1}$ where
	$\op_k$ is a deflate or
	burst operation on id $i$, and
	\item for every $p \in P$ if $d_i.\mmap(p) >0$ then there exists
	$j>i$ such that $d_j \xrightarrow{\delta} d_{j+1}$ where $\delta(p)
	<0$.
 \end{itemize} 
\end{definition}

The following proposition follows from the above definition.
\begin{proposition}
\label{prop-prog-id}
	In a progressive run-with-id $\tau$, for any id $i$, one of the
	following is true:
	\begin{enumerate}
		\item $\seq_i$ is finite and the final operation in $\seq_i$ is a
		burst operation, or
		\item $\seq_i$ is infinite.
	\end{enumerate}
\end{proposition}
If $\seq_i$ is finite for all $i$, we say `every balloon is burst in
$\tau$'.
The existence of progressive runs and progressive runs-with-id also
coincide.
\begin{proposition}
\label{prog-run-with-id}
	For any given $\cV$ and its semiconfiguration $s$, there exists a
	progressive run $\rho$ from $s$ iff there exists a corresponding
	progressive run-with-id $\tau'$ from $s$.
\end{proposition}
\begin{proof}
Clearly a progressive run $\rho$ can be obtained from a progressive
run-with-id $\tau$ by id-removal. For the converse direction, consider
the canonical $\tau=d_0 \xrightarrow{\op_1} d_1 
\xrightarrow{\op_1} d_2 \cdots$
	corresponding to $\rho=c_0 \xrightarrow{\op_1} c_1 \xrightarrow{\op_2} c_2 \cdots$.
	 The canonical $\tau$ satisfies all progressiveness conditions
	except for one special case: a balloon $b$ in $\rho$ which undergoes
	an infinite sequence of trivial deflates at $j_1,j_2,\cdots$
	may result in $\tau$ which contains multiple balloons-with-id in the
	balloon class of $b$ at every configuration $d_{j_1},d_{j_2},\cdots$.
	By construction, $\tau$ always
	chooses the least id and hence
	$\tau$ is not progressive for other id's in the balloon
	class of $b$.
	The progressive $\tau'$ is obtained by modifying $\tau$ to `dovetail'
	through all possible choices in the balloon class of $b$ (i.e. if
	the id's in the balloon class are $i_1 <i_2 < i_3 \cdots$, then the
	choices made are $i_1,i_1,i_2,i_1,i_2,i_3,i_1 \cdots$). We note
	that $\tau'$ agrees with $\tau$ on all inflate and burst operations.
\end{proof}

We now consider the two properties assumed in the proof of Lemma 
\ref{new-witness-prog}, namely that of being zero-base and being typed
and show that it is possible to construct a $\VASSB$ $\cV'$ with these
properties from a given $\VASSB$ $\cV$ while preserving reachability
and progressiveness.

\paragraph{Balloon Types and Surgeries on Balloons}
A deflate operation transferring a
non-zero number of tokens is called a \emph{non-trivial} deflate,
otherwise it is a trivial deflate. Clearly, a non-trivial deflate
moving tokens from a balloon place $\barp$ must be the first deflate
which transfers tokens from $\barp$, motivating the following
definition.
For an id $i$, define
\begin{flalign*}
 	\typeseq_i= & \{ (\barp,p,j)
\mid \exists \barq,\barq' \in \bbq\; \exists j \in \seq_i \colon
\op_j=\deflateb
(\barq,\barq',\barp,p) \text{ and } \\
 &\forall \barq'',\barq''' \in
\bbq \;\forall p' \in P\; \forall j'\in \seq_i, j' <j \colon \op_{j'}
\neq \deflateb
(\barq'',\barq''',\barp,p') \}
 \end{flalign*} 
 We write $i <_t j$ for id's $i,j$ if $i <j$, $L_i=L_j$ where $L_i$
 (resp. $L_j$) is
the linear set used to inflate balloon $i$ (resp. $j$) and for every $
(\barp,p) \in
\bbp \times P$, there exists $k \in \seq_i$ such that $(\barp,p,k)
\in \typeseq_i$ iff there exists $k' \in \seq_j$ with $k' >k$ such that $
(\barp,p,k') \in \typeseq_j$. 

A \emph{deflate-sequence} $S=(\barp_1,p_1),(\barp_2,\p_2), \cdots, 
(\barp_n,p_n)$ is a finite sequence of elements from $\bbp \times P$
which satisfies the property that for all $i, j \in [1,n]$, we have
$\barp_i \neq \barp_j$. We write $n=|S|$ and also write $\sset(S)$
for
the set of tuples $(\barp_i,p_i)$ obtained from $S$ by ignoring the
order. A marked deflate-sequence
$M=(S,i)$ is a tuple
consisting of a deflate-sequence $S$ and $i \in [0,\cdots,|S|]$. We
write $M.S$ and $M.i$ for the two components of $M$. If $M'=(S,i+1)$
and $M=(S,i)$, we write $M'=M+1$. The
set of all marked deflate-sequences is denoted by $\cM$.

 Let $\typeseq_i\hspace{-3 pt}\downharpoonright\;=(\barp_1,p_1),
 (\barp_2,p_2),\cdots,(\barp_n,p_n)$ be
the deflate-sequence obtained by
projecting $\typeseq_i$ to the first two
components, in increasing order of the third component. We write
$k=\typeseq_i(j)$ if $(\barp_j,p_j,k) \in \typeseq_i$. 
We write $i \sim_t j$ if $\typeseq_i\hspace{-3 pt}\downharpoonright
\;=
\typeseq_j\hspace{-3 pt}\downharpoonright$ and $L_i=L_j$. A tuple $t=(L_i,\typeseq_i\hspace{-3 pt}
\downharpoonright)$ is called a \emph{type} and only depends on $\cV$.
The set of finitely many types associated with $\cV$ is denoted $\cT
_{\cV}$; we  drop the subscript when $\cV$ is clear from the context.

Given any $\VASSB$ $\cV=(Q,P,\bbq,\bbp,E_p \cup E_n \cup E_d \cup E_b)$, 
its \emph{typed extension} $\cV'=(Q,P,\bbq',\bbp,E_p \cup E'_n
\cup E'_d \cup E'_b)$ 
is given by $\bbq'=\bbq \times \cM$ and set of balloon edges is given
as:
\begin{enumerate}
		\item For each $e=q \xrightarrow{\newb(\barq,L)} q'$ in $E_n$, for
		each $M \in \cM$ such that $M.i=0$,  add\\
				$q \xrightarrow{\newb((\barq,M),L)} q'$ to $E'_n$,
	\item for each $e=q \xrightarrow{\deflateb(\barq,\barq',\barp,p)} q'$
		in $E_d$ and for each $M \in \cM$ such that $(\barp,p)=(\barp_j,p_j)
		\in M$ for some $j \in
			[1,\cdots, |M.S|]$ 
		\begin{enumerate}
			\item if $j \leq M.i$ then  add

		$q \xrightarrow{\deflateb((\barq,M),(\barq',M),\barp,p)} q'$ to
		$E'_d$, 
		\item else if $j =(M.i)+1$ then  add
		$q \xrightarrow{\deflateb((\barq,M),(\barq',M+1),\barp,p)} q'$ to
		$E'_d$, and
		\end{enumerate}
		   
	\item for each $e=q \xrightarrow{\burstb(\barq)} q'$ in $E_b$ and 
	for each $M \in \cM$ with $M.i=|M.S|$,  add
		$q \xrightarrow{\burstb(\barq,M)} q'$ to $E'_b$.
	\end{enumerate}

\begin{lemma}
\label{lem-typed-vassb}
	Given any $\VASSB$ $\cV$ and two semiconfigurations $s_1,s_2$
	of $\cV$,
	 its typed extension $\VASSB$ $\cV'$ satisfies the following
	 properties:
	\begin{enumerate}
		\item $\cV$ has a progressive run starting from $s_1$ iff $\cV'$
		has a progressive run starting from $s_1$ and
		\item $(\cV,s_1,s_2) \in \REACH$ iff $(\cV',s_1,s_2) \in \REACH$.
	\end{enumerate}
\end{lemma}
\begin{proof}
	Given a run $\rho'=c'_0 \autstep[\op'_1] c'_1 \autstep[\op'_2] \cdots$ of
	$\cV'$ starting from a semiconfiguration $c'_0$, it is clear that there
	exists a run $\rho=c'_0 \autstep[\op_1]
	c_1 \autstep[\op_2] \cdots$ of $\cV$ where for each $k \ge 1$,
	\begin{itemize}
		\item if $\op'_k \in E_p$ then $\op_k=\op'_k$,
		\item else if $\op'_k$ creates a balloon $b'$ with state $
		(\barq,M)$, then $\op_k$ creates a
		balloon $b$ such that $b.\kmap=b'.\kmap$ and $b.\barq=\barq$,
		\item otherwise $\op'_k$ is a deflate or burst operation on a
		balloon $b'$ with state $(\barq,M)$ and $\op_k$
		applies the corresponding deflate or burst operation to a balloon
		$b$ with $b.\kmap=b'.\kmap$ and $b.\barq=\barq$.
	\end{itemize}
	By induction on length of the run, we see that $
	c_k.\mmap=c'_k.\mmap$ for
	each $k \ge 0$ and further, there is a bijection between $c_k.\nu$
	and $c'_k.\nu$ which preserves the balloon contents with the
	state of a balloon in $c_k$ obtained by projecting the state of the
	corresponding balloon to the first component. If $\rho'$ is a
	finite run ending in a semiconfiguration, this implies that $\rho$
	also ends in the same semiconfiguration. Since $\rho'$ is progressive
	with respect to balloon states of the form $(\barq,M)$, clearly
	$\rho$ is progressive with respect to the balloon states $\barq$.

	For the converse direction, we argue using runs-with-id.
	Let $\tau=d_0 \autstep[\op_1] d_1 \autstep[\op_2] \cdots$ be the
	canonical
	run-with-id of $\cV$ starting from a semiconfiguration $d_0$ for a
	given $\rho$ of $\cV$. We
	define the run-with-id $\tau'=d_0 \autstep[\op'_1] d_1 \autstep
	[\op'_2] \cdots$ of
	$\cV'$ where for each $k \ge 1$,
	\begin{itemize}
		\item if $\op_k \in E_p$ then $\op'_k=\op_k$,
		\item if $\op_k \in E_n$ creates $(b,i)$ then $\op'_k$ creates $
		(b',i)$ where $b'.\barq=(b.\barq,M)$ with $M=(\typeseq_i
		\downharpoonright,0)$ and $b'.\kmap=b.\kmap$, and
		\item if $\op_k \in E_d \cup E_b$ and $\op_k \in \seq_i$ for some id
		$i$, then $\op'_k \in \seq'_i$. Here $\seq'_i$(resp. $\seq_i$)
		refers to the
		sequence of operations on id $i$ in $\tau'$ (resp. $\tau$).
	\end{itemize}

	If $\tau$ is a finite run between semiconfigurations $s_1$ and $s_2$,
	then so is
	$\tau'$ and thus we have (2) of the lemma.
	 If $\tau$ is a
	progressive run-with-id from $s_1$ then so is
	$\tau'$ and by Proposition \ref{prog-run-with-id}, we conclude (1)
	of the lemma.
\end{proof}
In the remainder of this section, we assume that any $\VASSB$ is
typed.

A $\VASSB$ $\cV$ is said to be \emph{zero-base} if every linear set in
$\cV$ has base vector $\emptyset$.
\begin{lemma}
\label{lem-remove-base}
	Given any $\VASSB$ $\cV$, we can construct a zero-base $\VASSB$
	$\cV'$ such that for any two semiconfigurations $s_1,s_2$ of $\cV$,
	\begin{enumerate}
		\item $\cV$ has a progressive run starting from $s_1$ iff $\cV'$
		has a progressive run starting from $s_1$ and
		\item $(\cV,s_1,s_2) \in \REACH$ iff $(\cV',s_1,s_2) \in \REACH$.
	\end{enumerate}
\end{lemma}
\begin{proof}
We may assume that there exists a unique linear set $L_
{\barq}$ associated with any given balloon state $\barq$ of $\cV$ such
that any
balloon $b$ inflated with state $\barq$ has $b.\kmap \in L_{\barq}$,
by applying the following modification to $\cV$.
Let $\mathcal{L}$ be the finite set of linear sets used in $\cV$. We
replace the set of balloon states $\bbq$ by 
the cartesian product $\bbq \times \mathcal{L}$ and for each inflate
operation $\newb(\barq,L)$ we use $\newb((\barq,L),L)$ in the modified
$\VASSB$. This modification is easily seen to preserve
progressiveness and reachability. We also assume that the given
$\cV$ is typed and instead of writing balloon states as tuples $
(\barq,M)$, we assume that every state $\barq$ comes with an
associated $M_{\barq}$.

Let $\cV=(Q,P,\bbq,\bbp,E_p \cup E_n
	\cup E_d \cup E_b)$ and $\bmap_{\barq}$ be the base vector of a
	balloon state $\barq \in
\bbq$.
$\cV'=(Q',P,\bbq',\bbp,E'_p \cup E'_n \cup E'_d \cup E_b)$
is constructed from $\cV$ as follows:
\begin{enumerate}
	\item $Q'=Q \cup E_d \times P$,
	\item $\bbq'=\bbq$ contains a state $\barq'$ for each
	$\barq$ such that
			$L_{\barq'}$ is obtained from $L_{\barq}$ by removing the base
			vector $\bmap_{\barq}$,
	\item $E'_p=E_p \cup E''_p$, where $E''_p$ contains for each edge 
			$e=q \xrightarrow{\deflateb
			(\barq_1,\barq_2,\barp,p)} q'$ in $E_d$ with $M_{\barq_1} \neq M_
			{\barq_2}$, the edge
			$(e,p) \xrightarrow{\delta}
			q'$ where $\delta(p)=\bmap_{\barq}(p)$ and $\delta(p')=0$ for $p'
			\neq
			p$,
	\item $E'_n$ contains an edge $q \xrightarrow{\newb(\barq',L_
	{\barq'})} q'$
	for
			each edge $q \xrightarrow{\newb(\barq,L_{\barq})} q'$ in $E_n$, and
			\item $E'_d$ contains the following edges for each $e=q 
			\xrightarrow{\deflateb
			(\barq_1,\barq_2,\barp,p)} q'$ in $E_d$:
			\begin{enumerate}
			 	\item if $M_{\barq_1} \neq M_{\barq_2}$ then we add $q 
			\xrightarrow{\deflateb
			(\barq'_1,\barq'_2,\barp,p)} (e,p)$ to $E'_d$,
			\item else we add $q 
			\xrightarrow{\deflateb
			(\barq'_1,\barq'_2,\barp,p)} q'$ to $E'_d$.
			 \end{enumerate} 
\end{enumerate}
Given a finite run $\rho=c_0 \xrightarrow{} c_1 \autstep[] c_2 \cdots
c_k$
of $\cV$ between semiconfigurations,
there exists a finite run $\rho'=c'_0 \xrightarrow{*} c'_1 \autstep[*]
c'_2 \cdots c'_k$
where for
every $\newb(\barq,L)$
operation
in
$\rho$, we perform a $\newb(\barq',L)$ operation in $\rho'$, creating
a
balloon which is identical except for the base vector. The missing
base
vector tokens are then transferred to the appropriate place via edges
of the form $(e,p) \xrightarrow{\delta} q'$ from (2). Note that the
extra base-vector tokens are only added during non-trivial deflates,
which is information that can be obtained from the type information in
$M_{\barq}$.

 Conversely if
$\rho'$ is a a finite run of $\cV'$ between semiconfigurations of
$\cV$, we
observe that the states use to add the missing tokens are sandwiched
between the states of a $\deflateb$ operation and we obtain a run
$\rho$ from $\rho'$ by replacing the $\newb(\barq',L)$ operations by
$\newb(\barq,L)$ operations and replacing every chain of edges $q
\rightarrow (e,p) \rightarrow q'$ by the corresponding $\deflateb$
edge $q \rightarrow q'$ of $\cV$. Thus reachability is preserved. It is
easy to see that the
construction also preserves progressiveness since any infinite run of
$\cV'$
can also be decomposed into segments which are runs between
semiconfigurations of $\cV$.
 \end{proof}
Henceforth we assume that any $\VASSB$ is zero-base.

\paragraph{Token-shifting Surgery}
Fix a run-with-id $\tau=d_0 \autstep[\op_1] d_1 \autstep[\op_2] d_2
\cdots$ of a zero-base $\VASSB$
$\cV$. Let $i$ be
an id and $I$ a finite set of id's with the property that for each $j
\in I$
we have $i <_t j$. Let $N=\sum_{\kmap \in K} \kmap$ where $K=\{ \kmap
\mid \exists j \in I,
\op_j \text{ creates balloon } j \text{ with } \kmap=j.\kmap\}$.
 The token-shifting surgery $\cS_{i \leftarrow I}$
on $\tau$ yielding $\cS_{i \leftarrow I}(\tau)=\tau'=d'_0 \autstep
[\op'_1] d'_1 \autstep[\op'_2] d'_2
\cdots$ is defined as, for each $k \ge 1$:
\begin{enumerate}
	\item If $k \not\in I \cup \{ i\}$, then $\op'_k=\op_k$,
	\item else if $k=i$ then if $\op_i$ creates $(b,i)$ then $\op'_i$
	creates $(b',i)$ where $b'.\barq=b.\barq$ and $b'.\kmap= b.\kmap
	+ N$,
	\item else $k \in I$ and if $\op_k$ creates $(b_k,k)$ then $\op'_k$
	creates $(b',k)$ where $b'.\barq=b.\barq$ and $b'.\kmap =\emptyset$.
\end{enumerate}
In other words, the only difference between $\tau'$ and $\tau$ is the
content of balloons created at $I \cup \{i\}$ and the resulting
changes in the configurations.

\begin{proposition}
\label{valid-token-shifting}
	For any $\tau = d_0 \autstep[\op_1] d_1 \autstep[\op_2] d_2
\cdots$ and any $\cS_{i \leftarrow I}$, the run $\cS_{i
	\leftarrow I}(\tau)=\tau'=d'_0 \autstep[\op'_1] d'_1 \autstep[\op_2]
	d'_2
\cdots$ is a valid run. Further, if $\tau$ is a
	progressive run from $s$ then so is $\cS_{i
	\leftarrow I}(\tau)$. If $\tau=s \xrightarrow{*} s'$ is a finite run
	between two semiconfigurations, then $\cS_{i \leftarrow I}(\tau)$ is
	also a run from $s$ to $s'$.
\end{proposition}
\begin{proof}
	The fact that $\cV$ is zero-base ensures that $b'.\kmap$ as defined
	in (2)
	belongs to the linear set $L_i$. The condition $i <_t j$ for each $j
	\in I$ implies that the balloon $i$ makes its non-trivial deflates
	before any of the balloons with an id from $I$. This implies that for
	each $k \ge 1$ we have $d_k.\mmap \leq d'_k.\mmap$ since the total
	number of tokens transferred from $\barp$ to $p$ for any $\barp \in
	\bbp,p \in P$ in $\tau'$ remains greater than that in $\tau$ for
	every prefix $\tau[0,k]$. By monotonicity, every operation remains
	enabled in $\tau'$. Progressiveness is preserved since the only
	change is in the content of some of the newly created balloons. Since
	the total number of tokens transferred in the entire run is the same
	in $\tau$ and $\tau'$ for each $(\barp,p)$ in a finite run
	between two semiconfigurations $s,s'$, and there are no balloons in
	$s'$, this implies that $\tau'$ reaches the same semiconfiguration
	$s'$.
\end{proof}

\subsection{Proof of Lemma \ref{new-witness-prog}}
We have now introduced all of the machinery required to prove Lemma
\ref{new-witness-prog}. We will construct $\cV'$ from $\cV$ such that
the following are
equivalent:
\begin{enumerate}
	\item[(S1)] there exists a progressive run-with-id $\tau=s
	\rightarrow d_1
	\rightarrow
d_2 \cdots$ of $\cV$,
\item[(S2)] there exists a progressive run-with-id $\tau'=s
\rightarrow d'_1
\rightarrow
d'_2 \cdots$ of $\cV'$ which bursts every balloon, and
\item[(S3)] there exists an
	$A,B$-witness
	$\rho'_{A,B}=s \autstep[*] c \autstep[*] c'$ of $\cV'$ for some $A
	\subseteq P', B \subseteq \bbq'$.
\end{enumerate}

We will first show (S1) iff (S2) and then show (S2) iff (S3).
By Lemmas \ref{lem-remove-base} and \ref{lem-typed-vassb} we assume
that $\VASSB$ $\cV$ is zero-base and typed, and
that every $\barq \in \bbq$
comes with an associated marked deflate-sequence $M_{\barq}$ and
do not write balloon states explicitly as $(\barq,M_{\barq})$. From
$\cV=(Q,P,\bbq,\bbp,E_p \cup E_n
\cup E_d \cup E_b)$, we
construct
$\cV'=(Q',P',\bbq,\bbp,(E_p \cup E''_p) \cup E_n \cup E'_d
\cup E'_b)$ as
follows:
\begin{itemize}
	\item $Q'=Q \cup (E \times \{0,1\})$,
	\item $P'=P \cup (\bbq \times \{0,1\})$,
	\item $E'$ is defined as follows:
	\begin{enumerate}
	\item For each $e=q \xrightarrow{\deflateb(\barq,\barq',\barp,p)}
		q'$ in $E_d$ where $(\barp,p)=(\barp_j,p_j) \in M_{\barq}$,
			\begin{enumerate}
				 	\item if $M_{\barq}.i \leq |M_{\barq}.S| -1$ then
				 	add,\\
							$q \xrightarrow{\deflateb(\barq,\barq',\barp,p)} q'$ to $E'_d$,
				\item else if $M_{\barq}.i = |M_{\barq}.S|$,
						\begin{enumerate}
							\item $q \xrightarrow{\deflateb(\barq,
											\barq',\barp,p)} q'$ to $E'_d$,
								 \item $q \xrightarrow{\deflateb(\barq,
											\barq',\barp,p)} (e,0)$
											to $E'_d$,
								 \item $(e,0) \xrightarrow{\burstb(\barq')} (e,1)$
								 	to $E'_b$,
								\item $(e,1) \xrightarrow{\barq'^+} q'$ to $E''_p$.
									\item if $\barq = \barq'$, for $(i,j) \in \{ (0,1),(1,0)\}$ add $q 
										\xrightarrow{(\barq,i)^+(\barq,j)^-} q'$ to
										$E''_p$,
								\item else for each $i, j \in \{0,1 \}$ add $q 
								\xrightarrow{(\barq,i)^-(\barq',j)^+} q'$ to $E''_p$
						\end{enumerate}
				\end{enumerate}
	\item For each $e=q \xrightarrow{\burstb(\barq)}
				q'$ in $E_b$, add $q \xrightarrow{\burstb(\barq)}
				q'$ to $E'_b$.
	\end{enumerate}
\end{itemize}

For $k,k' \in \seq_i$, we write $k \lessdot_i k'$ to denote the
successor relation in $\seq_i$.
 Let $m$ be the smallest id in $\tau$ such
that $\seq_m$ is infinite.
Since there can be at most finitely many indices in $\seq_m$ where a
non-trivial deflate occurs, there exists $n \in \seq_m$
which is the last such index.
Intuitively, $\cV'$ uses the edges in (1 b ii) to (1 b vi) to burst the
balloon
with id $m$ at operation $n$ where the last non-trivial deflate occurs
and
uses a $\VASS$ token to simulate the effect of the infinite suffix of
$\seq_m$ which performs only trivial deflates. We use the notation
$p^+$(resp. $p^-$) to indicate that the place $p$ is incremented 
(resp. decremented). 
For example, for $p_1, p_2\in P$, we write $p_1^+$, $p_2^-$, and $p_1^-p_2^+$ for the vectors
mapping 
$p_1$ to $+1$ and all other places to $0$,
$p_2$ to $-1$ and all other places to $0$,
and $p_1$ to $-1$, $p_2$ to $+1$, and all other places to $0$, respectively.

 From the given progressive $\tau=s \xrightarrow{\op_1}
d_1
\xrightarrow{\op_2} d_2 \cdots$ of $\cV$, we construct a
 progressive $\tau'_1=d'_0 \xrightarrow{*}
d'_{i_1}
\xrightarrow{*} d'_{i_2} \cdots$ for $\cV'$. We set $d'_0=s$ and
for every $k$, the run $d'_{i_{k-1}} \autstep[*] d'_{i_k}$ is defined
as
follows:
\begin{enumerate}[(I)]
\item If $k \not\in \seq_m$ or $k<n$ then
\begin{enumerate}
					\item if $\op_k=\delta$ then $i_k=i_{k-1}+1$ and $d'_{i_{k-1}}
					\autstep[\delta] d'_{i_k}$,
				 	\item if $\op_k=\newb(\barq,L)$ creates balloon $(b,k)$, 
								then $i_k=i_{k-1}+1$ and\\
								 $d'_{i_{k-1}}
								\xrightarrow{\newb(\barq,L)} d'_{i_k}$
								creates $(b',i_k)$ with $b'.\kmap=b.\kmap$,
				\item else if $\op_k=\deflateb
								(\barq,\barq',\barp,p)$ deflates balloon $k'$, then $i_k=i_
								{k-1}+1$ and\\
								 $d'_{i_{k-1}}
								\xrightarrow{\deflateb(\barq,
								\barq',\barp,p)} d'_{i_k}$
								deflates balloon $
								i_{k'}$,
				\item else $\op_k=\burstb
								(\barq)$ bursts balloon $k'$, then $i_k=i_{k-1}+1$ and\\
								 $d'_{i_
								{k-1}}
								\xrightarrow{\burstb(\barq)} d'_{i_k}$
								bursts balloon $
								i_{k'}$,
 \end{enumerate} 

 \item otherwise we have $k \in \seq_m$ and $k \ge n$, in which case,
 \begin{enumerate}
 		\item if $k =n$ and $\op_k=\deflateb(\barq,\barq',\barp,p)$
 		deflates $m$,
 		 then $i_n=i_{n-1}+3$ and\\
 		  $d'_{i_{n-1}} 
				\xrightarrow{\op'_{i_{n-1}+1}} d'_{i_{n-1}+1}
				\autstep[\op'_{i_{n-1}+2}] d'_{i_{n-1}+2} \autstep[\op'_{i_
				{n-1}+3}]
				d'_{i_{n-1}+3}$ where 
				\begin{itemize}
					\item $\op'_{i_{n-1}+1}=\deflateb(\barq,\barq',\barp,p)$ deflates
					balloon $i_m$ where
					$n' \lessdot_m n$,
					\item  $\op'_{i_{n-1}+2}=\burstb(\barq)$, and 
					\item $\op'_{i_{n-1}+3}=(\barq,0)^+$,
				\end{itemize}
		\item else if $k>n$ and $\op_k=\deflateb(\barq,\barq',\barp,p)$ then
					$i_k=i_{k-1}+1$ and
					\begin{enumerate}
						\item if $\barq=\barq'$ and $\exists k' <_m k$ such that
						$\op'_{i_{k'}}=(\barq,1)^+(\barq,0)^-$ then
						\\
						$\op'_{i_{k}}=(\barq,0)^+(\barq,1)^-$ else\\
						$\op'_{i_{k}}=(\barq,1)^+(\barq,0)^-$,
						\item else if $\barq \neq \barq'$ and $ k' \lessdot_m k$ is
						such that
						$\op'_{i_{k'}}=(\barq,0)^+(\barq',0)^-$
						then\\ 
						$\op'_{i_{k}}=(\barq,0)^+(\barq',0)^-$ else\\
						$\op'_{i_{k}}=(\barq,1)^+(\barq',1)^-$.
					\end{enumerate}
 \end{enumerate}
\end{enumerate}

We repeat the construction starting with $\tau_1$, picking the next
higher index $m_1>m$ such that $\seq(i_1)$ is infinite and obtain a
sequence of runs $\tau_1,\tau_2,\tau_3 \cdots$ which agree on
increasingly larger prefixes (since the id's are chosen in order).
Thus we may define $\tau'$ as the limit of this sequence having the
property that every balloon is burst in $\tau'$. This implies that
$\tau'$ satisfies the progressiveness conditions for all places $p
\in P$ and all balloons. Our choices in (II b i) and (II b ii) imply that $\tau'$ also
satisfies progressiveness conditions for the places $\bbq \times \{ 0,
1\}$.

Conversely, consider an arbitrary progressive run-with-id $\tau'$ of
$\cV'$
which bursts all balloons. It may
be decomposed as
$\tau'=d'_0 \xrightarrow{*} d'_{i_1} 
\xrightarrow{*} d'_{i_2} \cdots$ where each segment $d'_{i_{k-1}} 
\xrightarrow{*} d'_{i_k}$ is either an edge from $E_p$,
(1 a),(1 b i),(1 b v),(1 b vi) or (2), or a sequence
of edges
(1 b ii,1 b iii,1 b iv). Choosing the
balloon with least index
$m$ such that
there exists a burst operation (1 b iii), we know there exists an
infinite sequence of corresponding
token transfers (1 b v,1 b vi) by progressiveness. By construction,
there
exist trivial deflates (1 b i) corresponding to these token transfers
and by replacing the transfers by (1 b i) edges, we obtain $\tau'_1$.

Repeating the construction with $\tau'_1$ we obtain a sequence of runs
$\tau'_1,\tau'_2,\cdots$ which agree on increasingly larger prefixes
and
we obtain a progressive run $\tau''$ of $\cV'$ as the limit of these
runs. Since no edges from (1 b ii) to (1 b vi) exist in $\tau''$, this
implies that $\tau''$ is infact a progressive $\tau$ of $\cV$. This
concludes the proof that (S1) iff (S2).

We now take up (S2) iff (S3), beginning with the only-if direction.
Let $\tau'=d'_0 \autstep[] d'_1 \autstep[]
d'_2
\cdots$ be a progressive run of $\cV'$ which bursts every
balloon. We explain the idea behind constructing the
required witness for (S3):
We pick a set of id's $I_1$ which contains one
balloon for each $\barq' \in \bbq'$. We then inductively pick sets $I_k$
for each $k >1$ such that all balloons in $I_{k-1}$ have been burst
before the creation of any balloon in $I_k$. This divides $\tau'$ into
segments such that the tokens of all balloons created in the $k^{th}$
segment excluding those from $I_k$ can be shifted to the balloons in
$I_{k-1}$. In particular, this creates a infinite sequence of
configurations which contain only empty balloons, which can be used to
construct the required witness. Recall that $I(\tau)$ denotes the
set of id's
in $\tau$ while $I_B(\tau)$ denotes the indices of the operations in
$\tau$ which are balloons operations.
Formally, for each $\barq' \in \bbq'$ we define\\
 $i_{1,\barq'}=\min\{ i
\in I(\tau') \mid \op_i=\newb(\barq',L_{\barq'})\}$,\\ $I_1=\{ i_
{1,\barq'} \mid \barq' \in \bbq'\}$ and\\
$B_1=\{ j_{1,\barq'} \in
I_B
(\tau') \mid
\exists i_{1,\barq'} \in I_1, \op_{j_{1,\barq'}}=\burstb(\barq'),
j_{1,\barq'} \in \seq_{i_{1,\barq'}}\}$.\\
 Inductively, we define for
$k>1$:\\
$i_{k,\barq'}=\min\{ l
\in I(\tau') \mid \op_l=\newb(\barq',L_{\barq'}), l > \max(B_
{k-1})\}$,\\
$I_k=\{ i_{k,\barq'} \mid \barq' \in \bbq'\}$ and\\
 $B_k=\{ j_
{k,\barq'}
\in I_B(\tau') \mid
\exists i_{k,\barq'} \in I_k, \op_{j_{k,\barq'}}=\burstb(\barq'),
j_{k,\barq'} \in \seq_{i_{k,\barq'}}\}$.

We also define, for $k \ge 1$,\\
 $I_{k,k+1,\barq'}=\{ l \in I(\tau')
\mid l \not \in I_{k+1}, \op_l=\newb(\barq',L_{\barq'}) \text{ and }
\max(B_k) <l <\max(B_{k+1})$ and\\
 $I_{k,k+1}=\bigcup_{\barq' \in
\bbq'}I_{k,k+1,\barq'}$.

We note that by construction, for each $k \ge 1$, the id $i_
{k,\barq'}$
and the set of id's $I_{k,k+1,\barq'}$ satisfy the
requirements for a token-shifting
surgery. Applying the surgeries $\cS_{i_{k,\barq'} \leftarrow I_
{k,k+1,\barq'}}$ to $\tau'$, we obtain $\tau''=d''_0 \autstep[] d''_1 \autstep[] d''_2 \cdots$
 which bursts every balloon. Therefore there exists $n \in \nats$ such
 that for
 any id $i$ satisfying $i < \max(B_1)$, balloon $i$ is burst before
 $d''_n$. By construction, for every $k > n$ the configuration $d''_
 {\max(B_k)}$ contains only empty balloons i.e. for every $(b,j)$ such
 that $d''_{\max(B_k)}.\nu(b,j)=1$, we have $(b,j).\kmap=\emptyset$.

 Let $S=d''_{j_1},d''_{j_2},\cdots$ be the infinite sequence of
 configurations in $\tau''$ such that they only contain empty
 balloons i.e. satisfy condition (3) of an $A,B$-witness. There exists
 an infinite subsequence $S'=d''_{k_1},d''_
 {k_2},\cdots$ of $S$ such that the corresponding id-removals
 $c''_{k_1},c''_{k_2},\cdots$ satisfy $c''_{k_1} \leq c''_
 {k_2} \leq\cdots$ due to the fact that $\leq$ is a
 Well-Quasi-Order. Further, by the Pigeonhole principle, we may assume
 that they all have the same set of occupied places and same set of
 balloon states $B$, thus satisfying condition (5) of an
 $A,B$-witness.
 We pick $A=P_\infty = \set{p\in P\mid \forall i\geq
 0\exists
 j>i\colon d''_j.\mmap(p) > 0}$, ensuring that $\sup(c''_{k_l})
 \subseteq A$ as required. By progressiveness, there
 exists
 $k_m,k_n$ with $k_m < k_n$ such that conditions (2) and (4) of an
 $A,B$-witness are satisfied for $c_{k_m}$ and $c_{k_n}$. Thus the run
 $\rho''=c_0 \autstep[*]c''_
 {k_m}
 \autstep[*] c''_{k_n}$ is the required $A,B$-witness. 

Conversely, given an $A,B$-witness $\rho''=c_0 \autstep[*]c
 \autstep[*] c'$ let $\rho'=c_0 \autstep[*]c_1
 \autstep[*] c_2 \autstep[*] c_3 \cdots$ be its `unrolling' where
 $c_1=c,c_2=c'$ and for each $k \ge 3$, the configuration $c_k$ is
 obtained from $c_{k-1}$ by applying the sequence of operations $c_1
 \autstep[*]c_2$. Clearly the unrolling satisfies all the
 progressiveness conditions. Consider a corresponding progressive
 run-with-id
 $\tau'=d'_0\autstep[*]d'_1
 \autstep[*] d'_2 \autstep[*] d'_3 \cdots$ given by Proposition \ref{prog-run-with-id}.
 By Proposition \ref{prop-prog-id}, every
 id created is either burst or, $\seq_i$ is infinite with $\seq_i$
 containting an infinite
 suffix of trivial deflates. By applying the construction in (S1) iff
 (S2),
 we replace every id $i$
 which has an infinite $\seq_i$ in $\tau'$ by a $\VASS$ token to
 obtain a
 progressive $\tau''$ where every balloon is burst. This concludes the
 proof of Lemma \ref{new-witness-prog}.

\subsection{Proof of Lemma \ref{lem:fairnterm-to-reach}}
\label{proof-fairnterm-to-reach}
The detailed construction of the $\VASSB$ $\cV(T)$
corresponding to an $A,B$-witness is given below. The $\VASSB$ $\cV
(T)$ consists of five total stages broken up into three main stages
and two auxiliary stages.
\emph{First Main Stage: Simulating two copies.} 
We define a $\VASSB$ $\cV_1$ that simulates two identical copies of the run of $\cV$.
The two copies share the global state in $\VASSB$ $\cV_1$.
It maintains two copies of the places, one for each run, 
and in addition, uses a pair of places $\bbq \times \set{1,2}$ for
each balloon state in order
to count the total
number of balloons of a given balloon state. An extra pair of places
$\bbq \times \set{3,4}$
for each balloon state remain unused in this stage and are used by
later stages for checking the emptiness of balloons.
Each step of $\cV$ is simulated by $\cV_1$ by updating two copies of the places and conceptually maintaining
two copies of the balloons.
The only non-trivial point is that we cannot maintain two balloons separately, because two different
steps executing $\newb(\barq, L)$ may pick two different contents from
the linear set $L$.
We avoid this by maintaining one balloon with two sets of balloon
places and extend the linear set $L$
so that it has two identical copies of the same multiset for each set of places.

We introduce some notation.
For a $v \in \nats^P$ (resp.\ $\delta \in \integs^P$), we write the ``doubling'' $v\odot v$ (resp.\ $\delta \odot \delta$) 
for the function in $\nats^{P\times\set{1,2}}$ such that $v\odot v (p,1) = v\odot v(p,2) = v(p)$
(resp.\ $\integs^{P\times\set{1,2}}$ such that $\delta\odot \delta(p,1) = \delta\odot\delta(p,2) = \delta(p)$).
For a linear set $L \subseteq \nats^P$, we similarly write $L \odot L \subseteq \nats^{P\times\set{1,2}}$ 
for the ``doubling'' of $L$:
$v\odot v \in L\odot L$ iff $v \in L$.
A representation of $L\odot L$ can be obtained from the representation of $L$ by doubling the base and period vectors.
We also write $\bfz \circ v$ (resp.\ $\bfz \circ L$) to denote the
extension of $v\in \nats^P$ (resp. $L\subseteq \nats^P$)
to $\nats^{P\times\set{1,2}}$: we define $\bfz\circ v(p,1) = 0$ and $\bfz\circ v(p,2) = v(p)$, and $\bfz\circ v\in \bfz\circ L$
iff $v\in L$.
Finally, we write a function $\delta\in\integs^P$ by explicitly mentioning the components that change.

Formally,
$\cV_1 = (Q \cup E \times\set{1,2}, P \times\set{1, 2} \cup \bbq\times
\set{1,2,3,4}, \bbq\times \bbq\cup\bbq, \bbp\times \set{1,2}, E_1)$,
where we add the following edges to $E_1$:
\begin{enumerate}
\item For each $e = q\autstep[\delta] q'$ in $E$ there is $q \autstep
[\delta
\odot \delta] q'$ in $E_1$.
\item For each $e= q \xrightarrow{\newb(\barq, L)} q'$ in $E$ the edges 
$q \xrightarrow{\newb((\barq,\barq), L\odot L)} (e,1)$ and $(e, 1) \xrightarrow{(\barq,1)^+(\barq,2)^+} q'$
are in $E_1$.
The second edge keeps track of the number of balloons in state $\barq$.
Note that the first stage only creates balloons with states in $\bbq\times\bbq$.
\item For each $e= q \xrightarrow{\deflateb(\barq,\barq',\barp,p)} q'$ in $E$, the following edges are in $E_1$:\\
$q \xrightarrow{\deflateb((\barq,\barq),(\barq',\barq), (\barp,1),(p,1))} (e,1)$,
$(e,1) \xrightarrow{\deflateb((\barq',\barq),(\barq',\barq'), (\barp,2),(p,2))}(e,2)$, 
$(e,2) \xrightarrow{(\barq,1)^-,(\barq',1)^+, (\barq,2)^-,
(\barq',2)^+} q'$.
The first two edges transfer the balloon tokens to $(p,1)$ and $(p,2)$, respectively.
The last edge tracks the new number of balloons in each balloon state.
\item For each $e = q \xrightarrow{\burstb(\barq)} q'$ in $E$, the edges $q \xrightarrow{\burstb{(\barq,\barq)}} (e,1)$ and
$(e,1) \xrightarrow{(\barq,1)^-(\barq,2)^{-}} q'$ are in $E'$.
The second edge maintains the count of the number of balloons with balloon state $\barq$.
\end{enumerate}

For an initial (semi-)configuration $c_0=(q_0, \mmap_0,\emptyset)$ of $\cV$, one can construct a configuration of $\cV_1$ that makes
two identical copies on to the places and initializes the places
$\barq\times\set{1,2,3,4}$ to zero.
A run of $\cV_1$ decomposes into two identical runs of $\cV$, and for any reachable configuration, the places $P\times \set{1}$ and
$P\times \set{2}$ have identical number of tokens.
So do the places $\bbq\times \set{1}$ and $\bbq\times \set{2}$ which
track the number of balloons in each copy for a given balloon state.
The places in $\bbq\times \set{3,4}$ remain empty.
Each balloon can be divided into two identical balloons by focusing on the two copies of the places.

\smallskip\noindent
\emph{First Auxilliary Stage: Checking Emptiness of Balloons in $c_1$
.}
The first auxiliary stage uses a $\VASSB$ $\cV_{1 \rightarrow 2}$
which is used to check balloons in $c_1$ for emptiness while at the
same time transferring tokens from the copy $\bbq \times \set{1}$ to
the copy $\bbq \times \set{3}$. The only places used in this stage are
those in $\bbq \times \set{1,3}$. The states used by $\cV_{1
\rightarrow
2}$ are two copies of the states of $\cV$ together with states of the
form $(q,\barq,\barp)$ used to perform the emptiness check. We also
have two marked copies of balloon states $\tilde{\bbq} \times 
\set{1,2}$. The
$\VASSB$ $\cV_{1 \rightarrow 2}$
picks a balloon $b$ from $c_1$ and performs
one deflate operation for
each $\barp \in \bbp$ on
$b$, sending all tokens to a special place $p_{\check}$. During the
check, the balloon is put into the first special marked copy
of its state to
ensure that all of the checks are performed on the balloon $b$ picked.
At
the end of the check, a token is transferred indicating that $b$ has
been checked for emptiness and the balloon state is put into the
second marked copy. The series of checks forms a loop starting
and ending at a state $(q,1)$ while passing through states of the form
$
(q,\barq,\barp)$, with one loop per balloon checked. In the event of a
correct check on all
balloons in $c_1$, the place $p_{\check}$ remains empty, all the
tokens have been transferred and control is passed to the second main
stage. In checking emptiness only the first copy of $Q$ is used. $\cV_
{1\rightarrow2}$
then non-deterministically guesses that all the balloons have been
checked for emptiness and moves from a state $(q,1)$ to the state $
(q,2)$ in the second copy of $Q$. It then converts the balloon states
of all the balloons from the second marked copy to the normal balloon
state. Note that this
conversion cannot be done before all the balloons are checked for
emptiness since we could otherwise make the mistake of checking only
one balloon repeatedly
for emptiness. 

Let $\bbp=\{\barp_1,\barp_2,\cdots,\barp_n\}$ be an
arbitrary enumeration.
Formally, $\cV_{1 \rightarrow 2} = (Q_{1 \rightarrow 2}, p_{\check} \cup
\bbq\times\set{1,2,3,4},((\bbq\times \bbq) \cup (\tilde{\bbq}\times 
\tilde{\bbq}\times\set{1,2})), \bbp\times \set{1,2},
E_{1 \rightarrow 2})$ where $\tilde{\bbq}$ is a decorated copy of
$\bbq$. The set global of states is $Q_{1 \rightarrow 2}=(Q\times
\set{1,2})
 \cup (Q \times \bbq\times 
\bbp)$ and
we add the
following edges to $E_{1 \rightarrow 2}$ for each
$q \in Q$ and $\barq \in \bbq$:
\begin{enumerate}
	\item $(q,1) \xrightarrow{\deflateb((\barq,\barq),
	(\tilde{\barq},\tilde{\barq},1),(\barp_1,1),p_{\check})}
	(q,\barq,\barp_1)$,

	\item for each $k \in [1,n-2]$ add \\
	$(q,\barq,\barp_k) \xrightarrow{\deflateb((
	\tilde{\barq},\tilde{\barq},1),
	(\tilde{\barq},\tilde{\barq},1),(\barp_{k+1},1),p_{\check})} 
	(q,\barq,\barp_
	{k+1})$, 

	\item $(q,\barq,\barp_{n-1}) \xrightarrow{\deflateb((
	\tilde{\barq},\tilde{\barq},1),
	(\tilde{\barq},\tilde{\barq},2),(\barp_n,1),p_{\check})} 
	(q,\barq,\barp_
	n)$ and

	\item $(q,\barq,\barp_n) \xrightarrow{(\barq,1)^-,(\barq,3)^+} 
	(q,1)$. A token is transferred from $(\bbq,1)$ to $(\bbq,3)$.

	\item $(q,1) \autstep[\bfz] (q,2)$,
	\item
	$(q,2) \autstep[\deflateb((\tilde{\barq},
	\tilde{\barq},2),(\barq,\barq),(\barp_n,1),p_{\check})] (q,2)$.

\end{enumerate}

\smallskip\noindent
\emph{Second Main Stage: Simulating $c_1 \xrightarrow{*} c_2$.}
At this point, we note that the first copy of $c_1$ resides in the
first copy of places and the third copy of balloon states.
The second main stage is given by a $\VASSB$ $\cV_2$ which keeps the
first copy of $c_1$ frozen and
simulates the $\VASSB$ $\cV$ on the second copy.
It additionally performs verification of the progress constraints.
Thus, a state consists of a triple $(q, q', M) \in Q\times Q \times 2^{A\cup B}$, where $q$ keeps the state
of $c_1$, $q'$ is the current state of the simulation, and $M$ is a
subset of $A\cup B$.
When control moves from state $q$ in $\cV_{1 \rightarrow 2}$ to $\cV_2$,
we start at $(q, q, \emptyset)$,
where the $\emptyset$ denotes none of the progress constraints have
been checked.

The simulation only updates the second copy of the places.
New balloons now use the states of $\cV$ and we continue to track the
number of balloons for each balloon state
 but only in the second copy i.e. the places $\bbq\times\set{2}$. The
 other copies $\bbq\times\set{1,3,4}$ remain untouched.
A $\deflateb$ or a $\burstb$ operation may be performed on double-balloons with state $\barq\times \barq$ or on normal balloons
with state $\barq$.
On double-balloons, $\deflateb$ works on the second
component only. 

Formally, $\cV_2 = (Q_2, P\times\set{1,2}\cup \bbq\times\set{1,2,3,4},
\bbq\times \bbq \cup \bbq, \bbp\times \set{1,2}, E_2)$
consists of states
$Q_2 = \set{ (q, q', M) \mid q\in Q, q' \in Q \cup E, M \subseteq A\cup B}$, and the following edges in $E_2$:
\begin{enumerate}
 \item 
Let $e= q \xrightarrow{\delta} q'$.
Let $P_{\delta}= \{p \in A \mid \delta(p) <0 \}$. If $\sup(\delta)
\subseteq A$, then
for each $q_1 \in Q$, $M \subseteq A\cup B$, and $p\in P$, add
 			$(q_1,q,M) \xrightarrow{\bfz\circ \delta}(q_1,q',M)$ if $p\in P\setminus P_\delta$ and
 			 $(q_1,q,M) \xrightarrow{}(q_1,q',M\cup P_{\delta})$ if $p\in
 			 P_\delta$.
The $\delta$ edges are restricted to those whose support is in $A$. We
track the place progress condition.
\item For each $e= q \xrightarrow{\newb(\barq, L)} q'$ in $E$, for each
$q_1 \in Q$, $M \subseteq A \cup B$, add
 	$(q_1,q,M) \xrightarrow{\newb(\barq, \bfz\circ L)}(q_1,e,M)$ and 
 	$(q_1,e,M) \xrightarrow{(\barq,2)^+} (q_1,q',M)$.

  \item For each $e= q \xrightarrow{\deflateb(\barq,\barq',\barp,p)} q'$ in $E$, 
for each $q_1 \in Q$, $M \subseteq A\cup B$, $\barq'' \in \bbq$ add:
$(q_1, q, M) \xrightarrow{\deflateb((\barq'',\barq),(\barq'',\barq'), 
(\barp, 2),(p,2))} (q_1,e,M')$,
$(q_1, q,M) \xrightarrow{\deflateb(\barq,\barq',(\barp,2),(p,2))} 
(q_1,e,M)$,
$(q_1,e,M) \xrightarrow{(\barq,2)^-,(\barq',2)^+}(q_1, q', M')$,
where $M' = M\cup\set{\barq}$ if $\barq\in B$ and $M'=M$ otherwise.
The first edge deflates doubled balloons left over from the first main
stage, but only to the second copy of $P$. We
 	track the balloon progress condition only for balloons from the
 	first main stage.

 \item For each $e= q \xrightarrow{\burstb(\barq)} q'$ in $E$, 
 			for each $q_1 \in Q$, $M \subseteq A \cup B$, $\barq' \in \bbq$,
 			add
 			 $(q_1, q, M) \xrightarrow{\burstb((\barq',\barq))}(q_1,e,M')$,
 			 $(q_1, q, M) \xrightarrow{\burstb(\barq)}(q_1,e,M)$, and
 			 $(q_1,e,M) \xrightarrow{(\barq,2)^-}(q_1,q',M)$,
			where $M' = M\cup\set{\barq}$ if $\barq\in B$ and $M' =M$ otherwise.
 		\end{enumerate}

\smallskip\noindent
\emph{Second Auxiliary Stage: Checking Emptiness of Balloons in
$c_2$.} 	
The difference between the first auxiliary stage and second auxiliary
stage is that
the emptiness must be checked for both the balloons produced in the
first main stage as well as the second main stage in the latter, since
there may be
a double-balloon
that continues to exist at $c_2$. The emptiness check transfers
tokens from the copy $\bbq \times \set{2}$ to the copy $\bbq \times 
\set{4}$.	As in the first auxiliary stage, we
pick an enumeration
$\set{\barp_1,\barp_2,\cdots,\barp_n}$ of $\bbp$.
Formally, $\cV_{2 \rightarrow 3} = (Q_{2 \rightarrow 3}, p_{\check} \cup
\bbq\times\set{1,2,3,4},(\bbq\times \bbq \cup \bbq) \cup ((
\tilde{\bbq}\times \tilde{\bbq} \cup \tilde{\bbq})\times\set{1,2}),
\bbp\times
\set{1,2},
E_{2 \rightarrow 3})$ where $\tilde{\bbq}$
is a
decorated
copy of $\bbq$. The set of global states is $Q_{2 \rightarrow
3}=(Q\times
\set{1,2})
 \cup
Q\times\bbp\times(\bbq \cup \bbq\times\bbq)$  and
we add the
following edges to $E_{2 \rightarrow 3}$ for each
$q \in Q$, $(\barq,\barq') \in \bbq \times \bbq$ and $\barq \in \bbq$:
\begin{enumerate}
	\item $(q,1) \xrightarrow{\deflateb((\barq,\barq'),
	(\tilde{\barq},\tilde{\barq'},1),\barp_1,p_{\check})}(
	q,\barp_1,\barq,\barq')$,\\
	 $q \xrightarrow{\deflateb(\barq,
	(\tilde{\barq},1),\barp_1,p_{\check})}(q,\barp_1,\barq)$,

	\item for each $k \in [1,n-2]$ add \\
	$(q,\barp_k,\barq,\barq') \xrightarrow{\deflateb(
	(\tilde{\barq},\tilde{\barq'},1),
	(\tilde{\barq'},\tilde{\barq'},1),\barp_{k+1},p_{\check})} (q,\barp_
	{k+1},\barq,\barq')$,\\
	$(q,\barp_k,\barq) \xrightarrow{\deflateb((\tilde{\barq},1),(\tilde{
	\barq},1),\barp_{k+1},p_{\check})} (q,\barp_{k+1},\barq)$,

	 \item $(q,\barp_{n-1},\barq,\barq') \xrightarrow{\deflateb(
	(\tilde{\barq},\tilde{\barq'},1),
	\tilde{\barq},\tilde{\barq'},2),\barp_n,p_{\check})} (q,\barp_
	n,\barq,\barq')$,\\
	$(q,\barp_{n-1},\barq) \xrightarrow{\deflateb((\tilde{\barq},1)
	(\tilde{\barq},2),\barp_n,p_{\check})} (q,\barp_n,\barq)$ and

	\item $(q,\barp_n,\barq,\barq') \xrightarrow{(\barq',2)^-,
	(\barq',4)^+}
	(q,1)$,
	$(q,\barp_n,\barq) \xrightarrow{(\barq,2)^-,(\barq,4)^+} 
	(q,1)$
	. A token is transferred from $(\bbq,2)$ to $(\bbq,4)$.

	\item $(q,1) \autstep[\bfz] (q,2)$,
	\item $(q,2) \autstep[\deflateb((\tilde{\barq},\tilde{\barq'},2),
	(\barq,\barq'),\barp_n,p_{\check})] (q,2)$,\\
	$(q,2) \autstep[\deflateb((\tilde{\barq},2),
	\barq,\barp_n,p_{\check})] (q,2)$.
\end{enumerate}

\smallskip\noindent
\emph{Third Main Stage: Verification.}
At this point, the first copy of $c_1$ uses balloon places $\bbq
\times \set{3}$ while $c_2$ uses the balloon places $\bbq \times 
\set{4}$.
In the third main stage, verification non-deterministically removes
the same number of tokens from the two copies for each place in $A$
and for each place in $B$ in the copies $\bbq\times \set{3,4}$.
Additionally, all balloons are burst.
Finally, all places in $P \times \set{2}$ and $\bbq\times \set{4}$ are
emptied.
Formally, $\VASSB$ $\cV_3 = (\set{q_\ver, q_f}, P\times\set{1,2} \cup
\bbq\times\set{1,2,3,4}, \bbq \cup \bbq\times \bbq, \bbp\times
\set{1,2},
E_3)$
contains the following edges:
\begin{enumerate}
	\item For each $p \in A$, there is an edge $q_{\ver} \xrightarrow{(p, 1)^- (p, 2)^-} q_{\ver}$ and
 		for each $\barq \in B$, the edge $q_{\ver} \xrightarrow{(\barq,
 		3)^- (\barq, 4)^-} q_{\ver}$.
For each $\barq \in \bbq$, the edge $q_{\ver} 
 	\xrightarrow{\burstb(\barq)} q_{\ver}$ and for each $(\barq,\barq')
 	\in \bbq \times \bbq$ the edge $q_{\ver} 
 	\xrightarrow{\burstb(\barq,\barq')} q_{\ver}$
\item $q_{\ver} \xrightarrow{\mathbf{0}} q_f$ and
for each $p \in A$, the edge $q_f \xrightarrow{ (p,2)^-} q_f$ and  
		for each $\barq \in B$, the edge $q_f \xrightarrow{(\barq,4)^-}
		q_f$.
\end{enumerate}
The edges in (1) ensure $\pc(c_1) \leq_p \pc(c_2)$ in the
places $A\cup B$ and that all remaining balloons are removed.
The edges in (2) remove extra tokens in case $\pc(c_1) <_p \pc (c_2)$.

\smallskip\noindent
\emph{Overall $\VASSB$ and Transitions between Stages.}
The $\VASSB$ $\cV(T)$ is a composition of the five stages $\cV_1$, $\cV_{1
\rightarrow 2}$, $\cV_2$, $\cV_{2 \rightarrow 3}$ and $\cV_3$.
A set of glue transitions connect the five stages: these
non-deterministically guess when to move
from $\cV_1$ to $\cV_{1
\rightarrow 2}$, $\cV_{1
\rightarrow 2}$ to $\cV_2$, $\cV_2$ to $\cV_{2\rightarrow 3}$ and from $\cV_{2\rightarrow 3}$
to $\cV_3$. Transfers between stages do not perform any operations and
use the following edges:
\begin{itemize}
	\item From $\cV_1$ to $\cV_{1 \rightarrow 2}$, for each $q\in Q$ of
	$\cV_1$ to the corresponding first copy $(q,1)$ in $\cV_{1
	\rightarrow 2}$,
	\item from $\cV_{1 \rightarrow 2}$ to $\cV_2$, for each $(q,2) \in Q
	\times\{2\}$ of $\cV_
	{1 \rightarrow 2}$ to the corresponding $(q,q,\emptyset)$ of $\cV_2$,
	\item from $\cV_2$ to $\cV_{2 \rightarrow 3}$, for each $(q,q,A \cup B)$
	of $\cV_2$ to the corresponding $(q,1)$ of $\cV_{2 \rightarrow 3}$,
	and
	\item from $\cV_{2 \rightarrow 3}$, for each $(q,2) \in Q\times\{
	2\}$ of $\cV_{2
	\rightarrow 3}$ to $q_{\ver}$ of $\cV_3$.
\end{itemize}

\smallskip\noindent
\emph{Correctness.}
A correct simulation of the progressiveness witness results in $\cV(T)$
reaching
the configuration $(q_f,\emptyset,\barz)$. 
Conversely, suppose $\cV(T)$ reaches $(q_f,\emptyset,\barz)$.
This implies that it reaches a configuration $(q_{\ver},\mmap,\nu)$
where $\mmap$ only contains tokens in the second copy of $A$ and the
fourth copy of
$B$. This implies that all the balloons of $c_2$ were correctly
checked for emptiness in $\cV_{2 \rightarrow 3}$ since no tokens remain
in the second copy of $\bbq$ and there are no tokens in $p_{\check}$.
Similarly, the emptiness check performed by $\cV_{1 \rightarrow 2}$ is
also correct.
Thus, the configuration $c_1$ at the end of $\cV_1$ and $c_2$ at the end of $\cV_2$ 
satisfy $\pc(c_1) \leq \pc(c_2)$ and they both only contain empty
balloons with state $B$
and have non-zero tokens in places in $A$.
Further, the transition between $\cV_2$ and $\cV_{2 \rightarrow 3}$
ensures the progressivenes conditions on $A$ and $B$ have been
checked.

\section{Proofs for Section \ref{sec:reach2reach}}
The proofs of Lemmas \ref{lem:N-bal-bdd} and  \ref{lem:reach2reach}
both assume that the $\VASSB$ considered is both zero-base (by Lemma
\ref{lem-remove-base}) and typed (by Lemma \ref{lem-typed-vassb}).
Further, the notion of runs-with-id introduced in Appendix
\ref{appendix-sec-five-ndag} is used in the proof of Lemma \ref{lem:N-bal-bdd}.
\subsection{Proof of Lemma \ref{lem:N-bal-bdd}}

In any finite run-with-id $\tau$, let $I_{\barq}(\tau)$ be the set of
id's of non-empty
balloons inflated with state $\barq$. Note that by definition, $M_
{\barq}.i=0$ for any newly created balloon.
Let $\tau'=s_1 \autstep[*] s_2=d'_0 \autstep[\op'_1] d'_1 \autstep
[\op'_2] d'_2 \cdots$ be an arbitrary canonical run-with-id
which is
$N$-balloon bounded i.e. inflates at most $N$ non-empty balloons. We
equivalently assume that the bound applies to each balloon
state since this implies a bound of $|\bbq|N$ on the total number
of balloons. Suppose there exists a state $\barq_0$ with $|I_{\barq_0}
(\tau')| > N$. We will show
that
there exists $\tau=s_1 \autstep[*] s_2$ of $\cV$ with $|I_{\barq}
(\tau)|=|I_{\barq}
(\tau')|$ for $\barq \neq \barq_0$ and $|I_{\barq_0}(\tau)| < |I_
{\barq_0}
(\tau')|$. Clearly, this suffices to prove the lemma.\\
For any deflate sequence $S=(\barp_1,p_1),\cdots,(\barp_n,p_n)$, let
$s(S)$ be the set $\set{(\barp_1,p_1),\cdots,(\barp_n,p_n)}$.
	For id's $i,j \in I_{\barq_0}(\tau')$ with $i <j$, let $\col_
	{i,j}:(s(M_{\barq_0}.S))
	\rightarrow \{
	\lgreen,\lred\}$ be defined by 
	\[ \col_{i,j}(\barp,p)=\begin{cases}
		\ltr{green} &\text{ if } \exists k<k', (\barp,p,k) \in \typeseq_i, 
		(\barp,p,k') \in \typeseq_j\\
		\ltr{red} &\text{ otherwise}
	\end{cases}\]
	The color $\lgreen$ (resp. $\lred$) indicates that a deflate $
	(\barp,p)$ of the
	balloon $i$ occurs before (resp. after) the corresponding deflate of
	balloon $j$.
	We write $k=\typeseq_i(\barp,p)$ if $(\barp,p,k) \in \typeseq_i$.
The set $\cC_{\barq_0}=\set{\col_{i,j} \mid i,j \in I_{\barq_0}
(\tau')}$ of functions is finite since both the domain and range of
the functions is finite. Let $G_{\barq_0}$ be the graph with $I_{\barq_0}$
as the set of vertices and  edges colored by $\cC_{\barq_0}$. Note
that the color $\col_{i,j}$ of the edge between vertices $i$ and $j$
is a finite word from $
\set{\lgreen,\lred}^*$. Let
$r=|\cC_{\barq_0}| \leq 2^{|\bbp|}$
and $n=|\bbq|^{|\bbp|}+1$. Then by Ramsey's theorem~\cite[Theorem
B]{Ramsey1930} there exists $R
(r;n) \in \nats$ such that for any $r$-colored graph with at least $R
(r;n)$ vertices, there exists a monochromatic subgraph of size at
least $n$. Further, by the result of Erd\H{o}s and
  Rado~\cite[Theorem~1]{ErdosRado1952} we know that $R(r;n) \leq r^{r(n-2)+1} \leq O
(\exp_4(|\cV|))$. Choosing
$N=r^{r(n-2)+1}$, this implies that $|I_{\barq_0}| > r^{r(n-2)+1}$ and
therefore there exists a monochromatic subgraph $G'$ of $G_{\barq_0}$
with
at least $n$ many vertices. Let $c$ be the color of every edge in the
graph induced
by $G'$. A
\emph{$\lred$-block} $c[i,j]$ of $c$ is a maximal contiguous
subsequence $
(\barp_i,p_i),(\barp_{i+1},p_{i+1}),\cdots,(\barp_j,p_j)$ of $M_
{\barq_0}.S$ such
that $c$
takes value $\lred$ on each $(\barp_k,p_k)$ for $i \leq k \leq j$ and
value $\lgreen$ on $(\barp_{j+1},p_{j+1})$ (if $j<n$) and $(\barp_
{i-1},p_
{i-1})$ (if $i>1$).
Let
$c[j_{1,1},j_{1,2}],c[j_{2,1},j_{2,2}], \cdots,c[j_{l,1},j_{l,2}]$ be
the $\lred$-blocks of $c$ with $j_{k,2} < j_{k+1,1}$ for each $k$.

 Intuitively, deflates of balloons from $G'$ in a particular red-block
happen in the reverse order as compared to their id's. That is, if
$i_0<j_0<k_0<l_0$ are id's which are present in $G'$, then in a
red-block, $l_0$ is deflated first, followed by $k_0$, then $j_0$ and
finally $i_0$. In Figure \ref{ramsey-fig}, one such red-block (called
$1$-block in the main body) is formed by the second and third deflates
where the magenta outlined balloon corresponding to $l_0$ is deflated
first at $l_2$ and $l_3$, followed by the green, orange and cyan
outlined
balloons.

 Formally, let $\min(G')$ (resp. $\max(G')$) be the
least
(resp. greatest) id in $G'$. For each $i \in G'$ and for each $k$
with $ 1
\leq k \leq l$, we have $\typeseq_{\max(G')}(j_{k,1}) \leq \typeseq_i
(j_{k,1})$ and $\typeseq_i(j_{k,2}) \leq \typeseq_{\min(G')}(j_{k,2})$
by construction. In other words, for any $\lred$-block $c[j_{k,1},j_
{k,2}]$ of $c$,
all non-trivial deflates made by balloons in $G'$ happen in $\tau'
[m_{k,1},m_{k,2}]$ where $m_{k,1}=\typeseq_{\max(G')}(j_{k,1})-1$ and
$m_{k,2}=\typeseq_{\min
(G')}
(j_{k,2})$ and further, no other non-trivial deflates of balloons in
$G'$
other than the ones
in $c[j_{k,1},j_{k,2}]$ occur in $\tau'[m_{k,1},m_{k,2}]$. 
Let $\tau'
[m_
{1,1},m_{1,2}],\tau'[m_{2,1},m_
{2,2}], \cdots, \tau'[m_{l,1},m_{l,2}]$ be the infixes of $\tau'$
corresponding to $c[j_{1,1},j_{1,2}],c[j_{2,1},j_{2,2}], \cdots, c[j_
{l,1},j_{l,2}]$. We observe that the total number of points $j_{1,1},j_{1,2},j_{2,1},
\cdots j_{l,2}$ is at most $|\bbp|$ since the $\barp$ values at all
these points is distinct and therefore, the same bound applies to $m_
{1,1},m_{1,2},m_{2,1},m_
{2,2}, \cdots, m_{l,1},m_{l,2}$. Let $M=\bigcup_{1 \leq k \leq
l}[m_
{k,1},m_
{k,2}]$.

By construction, $|G'|=n \geq |\bbq|^
{|\bbp|}+1$, which implies that there exist $i,j \in G'$ with $i<j$
such that for each $k \geq 1$ and each $l \in \{ 1,2\}$ we have $d'_
{m_{k,l}}.
i.\barq=d'_{m_{k,l}}.j.\barq$ i.e. the balloons $i$ and $j$
have the same balloon state at each of the configurations $d'_{m_
{k,l}}$. Let $\op'_i$ (resp. $\op'_j$) in $\tau'$ create a balloon $
(b',i)$ (resp. $(b'',j)$) with
$b'.\kmap=\kmap_i$ (resp. $b''.\kmap=\kmap_j$). 
We now
construct $\tau=d_0 \autstep[\op_1] d_1 \autstep[\op_2] d_2 \cdots$
from $\tau'$ by an \emph{id-switching surgery} in which all operations
in $\tau$ are retained by $\tau'$ with the exception of those in
$\seq'_i$ and $\seq'_j$, which are replaced by $\seq_i$ and
$\seq_j$
as follows: 
\begin{enumerate}
	\item $\op_i$ creates $(b_i,i)$ where $b_i.\kmap=\kmap_i +
	\kmap_j$ and $b_i.\barq=b'.\barq$
	and $\op_j$ creates $(b_j,j)$ where $b_j.\kmap=\emptyset$ and
	$b_j.\barq=b''.\barq$,
	\item for each $k \in \seq'_i$ (resp. $k \in \seq'_j$) where $k
	\in M$,
	 $k \in \seq_j$ (resp. $k \in \seq_i$) and
	 \item for each $k \in \seq'_i$ (resp. $k \in \seq'_j$) where $k
	\not \in M$,
	 $k \in \seq_i$ (resp. $k \in \seq_j$).
\end{enumerate}
Intuitively, all operations in $\tau'$ belonging to the segments in
$M$
on balloon $i$ are reinterpreted as being
performed on balloon $j$ in $\tau$ and vice versa, while operations
outside of the segments in $M$ are retained as before. First we
observe that this reinterpretation results in a valid sequence of
operations $\seq_i$ and $\seq_j$ in $\tau$ since we have chosen $i,j$
such that the balloon states match at every boundary point of $M$.
This implies
that $i <_t j$ in $\tau$ and hence (1) which is a token-shifting
operation ensures that $\tau$ is valid by Proposition 
\ref{valid-token-shifting} and is a run from $s_1$ to $s_2$. Since
balloon $j$ is empty in $\tau$, this implies $|I_{\barq_0}(\tau)| <
|I_
{\barq_0}(\tau')|$.

\subsection{Proof of Lemma \ref{lem:reach2reach}}

	Given $\VASSB$ $\cV=(Q,P,\balloonQ,\balloonP,E)$ and its semiconfigurations $s_1,s_2$, we
	construct a $\VASS$ $\cV'=(Q',P',E')$ and its configurations
	$s'_1,s'_2$ such
	that there exists an $N$-balloon-bounded run $\rho=s_1 \autstep[*]
	s_2$ of $\cV$ (where $N$ is given by Lemma \ref{lem:N-bal-bdd}) iff $
	(\cV',s'_1,s'_2) \in \REACH$.

We assume that $\cV'$ can have edges of the form $e=q 
\xrightarrow{+\Delta_L}
q'$ where $+\Delta_L$ is an operation which non-determinstically adds
an element from a
linear
set $L$ to the set of places, since reducing such an extension to a
normal VASS
is a standard construction. We similarly have operations
$-\Delta_L$ which remove a vector $\vmap \in L$ from the 
set
of
places.
For sets $\bbp=\{ \barp_1, \cdots, \barp_n\}$ and $\bbp'=\{
\barp'_1,\cdots,\barp'_n\}$, for
$L \subseteq \multiset{\bbp}$, we write $L(\bbp \leftarrow \bbp')
\subseteq
\multiset{\bbp'}$ for the corresponding linear set of vectors
$\vmap'$
where
$\vmap'(\barp'_i)=\vmap(\barp_i)$ for each $i$ where $\vmap \in L$.
For $L \subseteq \multiset{\set{\barp}}$ and $p \in P$, we also define
the multiset $L(\barp \odot (-p)) \subseteq \multiset{\set{\barp,p}}$
consisting of vectors $\vmap^- \odot
\vmap(p)$ obtained from $\vmap \in L$ as $(\vmap^- \odot
\vmap(p))(\barp)=-(\vmap^- \odot
\vmap(p))(p)=-\vmap(\barp)$.
The projection $L~\hspace{-0.1 cm}\downharpoonright_{\chi}$ of $L$ to
a subset $\chi
\subseteq \bbp$ is obtained by retaining the values of vectors for the
$\chi$ component and setting the rest to 0. All of the above defined
operations enable us to simulate balloon operations by addition and
subtraction operations on $\VASS$ places.

Let $H$ be the set of all functions $\eta:\bbq \times 
\set{1,\cdots,N} \rightarrow \bbq \cup \set{\#,\dagger}$. A function
$\eta \in H$ is used in the following construction to track the state
of a non-empty balloon of $\cV$ in the global state of $\cV'$. The
$\#$ symbol indicates that a balloon
has non
yet been inflated and $\dagger$ indicates that it has been burst. Let
$G$ be
the set of vectors $\wmap:\bbq \rightarrow \set{0,1,\cdots, N}$.
A
vector $\wmap \in G$ is used to store the number of non-empty balloons
of a particular balloon state in $\cV$ in the global state of $\cV'$.
Increments $p^+$ (resp. decrements $p^-$) to a place $p$ are used as
before. For a type sequence $S=(\barp_1,p_1), \cdots,(\barp_n,p_n)$,
we write $\sset_1(S)$ for the set $\set{\barp_1,\cdots,\barp_n}$.

$\cV'$ is defined as:
\begin{itemize}
	\item $Q'=Q \times G \times H \cup \set{q_f}$,
	\item $P'=P \cup \bbq\times\set{1,\cdots,N}\times\bbp \cup \bbq$,
	\item $E'$ contains the following edges:
	\begin{enumerate}
		\item for each $q \autstep[\delta] q'$ in $E$, for each $\wmap \in
		G$ and $\eta \in H$ add\\
		$(q,\wmap,\eta) \autstep[\delta] (q',\wmap,\eta)$ to $E'$,
		\item for each $q \autstep[\newb(\barq_1,L)] q'$ in $E$,
		\begin{enumerate}
			\item for each $\wmap \in
				G$ and $\eta \in H$ add $(q,\wmap,\eta) \autstep[\barq_1^+] 
				(q',\wmap,\eta)$ to $E'$,
			\item for each $\eta \in H$ and for each $\wmap \in G$ such that
					$\wmap(\barq_1)<N$, add\\
					$(q,\wmap,\eta) \autstep[\Delta_{\wmap,\barq_1}] 
					(q',\wmap',\eta')$ to $E'$
					where\\
					 $\Delta_{\wmap,\barq_1}$ adds a vector from $L_{\barq_1}
					(\bbp \leftarrow \{\barq_1\}\times\{\wmap
					(\barq_1)+1\}\times\bbp)$,\\
					 the vector
					$\wmap'$ satisfies $\wmap'(\barq_1)=\wmap(\barq_1)+1$ and
					$\wmap'(\barq)=\wmap(\barq)$ for $\barq \neq \barq_1$, and\\
					$\eta'$ satisfies $\eta'(\barq_1,\wmap(\barq_1)+1)=\barq_1$ and
					for all $(\barq,i) \neq (\barq_1,\wmap(\barq_1)+1)$ we have
					$\eta'(\barq,i)=\eta(\barq,i)$,
		\end{enumerate}
			\item for each $q \autstep[\deflateb(\barq,\barq',\barp,p)] q'$ in
			$E$,
				\begin{enumerate}
					\item for each $\wmap \in
				G$ and $\eta \in H$ add $(q,\wmap,\eta) \autstep[\barq'^+,\barq^-] 
				(q',\wmap,\eta)$ to $E'$, 
				\item for each $\wmap \in G$, for each $\eta \in H$ and $
				(\barq_1,i) \in \bbq \times \set{1,\cdots,N}$ such that $\eta
				(\barq_1,i)=\barq$,
				\begin{enumerate}
					\item if $M_{\barq'}.i=M_{\barq}.i+1$ then add\\
					$ (q,\wmap,\eta) \autstep[\Delta_{\barq_1,i,\barp,p}] 
				(q',\wmap,\eta')$
				to $E'$
				where\\
				 $\Delta_{\barq_1,i,\barp,p}$ adds a vector $\vmap^-
				 \odot\vmap(p)$ corresponding to some $\vmap \in L_
				 {\barq_1}~\hspace{-0.2 cm}\downharpoonright_{\barp}(\barp
				 \leftarrow \{\barq_1\} \times \{i\} \times \barp)$
				and\\
				$\eta'$ satisfies $\eta'(\barq_1,i)=\barq'$ and $\eta'
				(\barq,j)=\eta(\barq,j)$ for all $(\barq,j) \neq (\barq_1,i)$ ,
				\item else add\\
				$ (q,\wmap,\eta) \autstep[\bfz] 
				(q',\wmap,\eta')$
				to $E'$
				where\\
				$\eta'$ satisfies $\eta'(\barq_1,i)=\barq'$ and $\eta'
				(\barq_2,j)=\eta(\barq_2,j)$ for all $(\barq_2,j) \neq 
				(\barq_1,i)$ ,
				\end{enumerate}
				
				\end{enumerate}
			\item for each $q \autstep[\burstb(\barq)] q'$ in
			$E$,
				\begin{enumerate}
					\item for each $\wmap \in
				G$ and $\eta \in H$ add $(q,\wmap,\eta) \autstep[\barq^-] 
				(q',\wmap,\eta)$ to $E'$, 
				\item for each $\wmap \in G$, for each $\eta \in H$ and $
				(\barq_1,i) \in \bbq \times \set{1,\cdots,N}$ such that $\eta
				(\barq_1,i)=\barq$, add\\
				$ (q,\wmap,\eta) \autstep[-\Delta_{\barq_1,i,(\bbp \setminus \sset_1
				(M_
				{\barq_1}.S))}]
				(q',\wmap,\eta')$ to $E'$
				where \\
				 $\Delta_{\barq_1,i,(\bbp \setminus \sset_1(M_{\barq_1}.S))}$ adds
				 a
				 vector from $L_{\barq_1}~\hspace{-0.2cm}\downharpoonright_{(\bbp
				 \setminus
				 \sset_1(M_{\barq_1}.S))}((\bbp \setminus \sset_1(M_{\barq_1}.S))
				 \leftarrow \{\barq_1\} \times \{i\} \times (\bbp \setminus \sset_1(M_
				 {\barq_1}.S) ))$
				and\\
				$\eta'$ satisfies $\eta'(\barq_1,i)=\dagger$ and $\eta'
				(\barq_2,j)=\eta(\barq_2,j)$ for all $(\barq_2,j) \neq 
				(\barq_1,i)$,
				\end{enumerate}
			\item for each $\eta$ such that for all $(\barq,i) \in \bbq \times
			\set{1,\cdots,N}$ we have $\eta(\barq,i)=\#$ or $\eta
			(\barq,i)=\dagger$, add\\
			$(s_2.q,\wmap,\eta) \autstep[\bfz] q_f$.
	\end{enumerate}
\end{itemize}
The configuration $s'_1$ is given as $s'_1.q=(s_1.q,\bfz,\eta_0)$
where $\eta_0(\barq,i)=\#$ for all $(\barq,i) \in \bbq \times 
\set{1,\cdots,N}$, and $s'_1.\mmap
(p)=s_1.\mmap(p)$ for $p \in P$ and $s'_1.\mmap(p)=0$ otherwise. The
configuration $s'_2$ is given as $s'_2.q=q_f$, and $s'_2.\mmap
(p)=s_2.\mmap(p)$ for $p \in P$ and $s'_2.\mmap(p)=0$ otherwise.

$\cV'$ contains extra places $\bbq$ which are used to store the number
of balloons of a given state $\barq$ which were created empty and
remain empty throughout a run, as well as places $\bbq\times\set{1,\cdots,N}\times\bbp$
which are used to store the contents of non-empty balloons created
during an $N$-balloon-bounded run of $\cV$. The $G$-component of the
global state is
used to track the total number of non-empty balloons created while the
$H$-component tracks the state changes of every balloon which was
created non-empty. The initial configuration $s'_1$ of $\cV$ contains
$s_1$ in its $P$-places while other places remain empty. The
$G$-component of the global state is set to the vector $\bfz$ since 
no non-empty balloons have been created and the $H$-component is set
to the function $\eta_0$ which is the constant function taking value
$\#$. The edges in (1) are used to simulate the place edges of $\cV$
and hence do not modify either the $G$ or the $H$-component.
Similarly, the simulation of operations on balloons which were created
empty use edges from $(2 a),(3 a)$ and $(4 a)$ which also do not
modify either the $G$ or the $H$-component and simply increment or
decrement places from $\bbq$ as appropriate.

 An inflate
balloon operation in the run $\rho$ of $\cV$ which creates the $i^
{th}$
non-empty balloon (upto the bound $N$) with state $\barq$ is simulated
by using edges
from $(2 b)$ which
increment $\wmap \in G$ from $i-1$ to $i$ and set the value of $\eta
(\barq,i)$ to $\barq$ from its original value of $\#$, while
simultaneously populating the places in $(\barq,i,\bbp)$ with a
vector from $L_{\barq}$. Deflate operations on non-empty balloons
similarly are simulated by updating $\eta$ and removing tokens from
the appropriate place using the edges in $(3 b)$. Here, we make a
distinction between the first time that a place $\barp$ undergoes a
deflate, in which case $(3 b i)$ is used and for every later deflate
on the place $\barp$, we use $(3 b ii)$ which does not transfer any
tokens. This implies that that for every non-empty balloon created,
there is only one opportunity for $\cV'$ to transfer tokens away from
places in $\bbq\times\set{1,\cdots,N}\times\bbp$.

The edges in $(4 b)$ simulate a burst operation by allowing the
removal of tokens from places which have not been transferred away,
while setting the value of $\eta$ to $\dagger$ at the appropriate
place to indicate that the balloon has been burst.

Finally, we only allow a move to the state $q_f$ if the $H$-component
indicates that all non-empty balloons produced have been burst. 
A correct simulation ensures that no tokens remain in any of the
places $\bbq\times\set{1,\cdots,N}\times\bbp \cup \bbq$. Thus
$\cV'$ can reach $s'_2$ from $s'_1$ if $\cV$ can reach $s_2$
from $s_1$. 
Conversely, we note that the sequence of deflates on each non-empty
balloon is tracked in the global state and as mentioned,
there is only one opportunity for $\cV'$ to correctly transfer tokens
from the places $\bbq\times\set{1,\cdots,N}\times\bbp$ for each
non-empty balloon. Thus, from any run $s'_1 \autstep[*] s'_2$ of
$\cV'$ we obtain a run $s_1 \autstep[*] s_2$ of $\cV$.

\section{Proofs from Section~\ref{starvation}}\label{appendix-starvation}

\subsection{Proof of Lemma~\ref{starving-consistent}}\label{appendix-proof-starving-consistent}
\starvingConsistent*
\begin{proof}
  Clearly, a starving run of a DCPS has to be consistent. Suppose a
  DCPS has a consistent run $\rho$.  Then there are configurations
  $c_1,c_2,\ldots$ and thread executions $e_1,e_2,\ldots$ that produce
  $\mmap_1,\mmap_2,\ldots$ so that the consistency conditions are met.
  In particular, for $V_p=\{\mmap_j\mid \text{$e_j$ has type
    $t$}\}$, the tuple $\frakV=(V_t)_{t\in\cT}$ is
  $\frakS_{\cA}$-consistent. This means, there exists a stack content $w\in\Gamma^*$ such that
  we can choose for each $j\ge 1$, a
  thread execution $e'_j$ (not necessarily occurring in $\rho$) that produces a multiset $\mmap'_j$ such that:
  \begin{enumerate}
  \item $e'_j$ has the same type as $e_j$
  \item $\mmap_j\le_1\mmap'_j$, and
  \item $e'_j$ arrives in $w$ after executing $i$ segments.
  \end{enumerate}
  The idea is now to replace executions $e_j$ with $e'_j$ to obtain a
  starving run. Replacing any individual $e_j$ by $e'_j$ would yield
  an infinite run, because $e'_j$ has the same type (so that state
  transition would still be possible), and $\mmap'_j\ge_1\mmap_j$, so
  that the bag contents are supersets of the bag contents of the old
  run.

  However, the resulting run might not be progressive: $e'_j$ might
  spawn an additional thread $(\gamma,0)$ such that $(\gamma,0)$
  becomes active only finitely many times during $\rho$. In that case,
  the extra $(\gamma,0)$ would be in the bag forever, but never
  scheduled.

  To remedy this, we only replace all but finitely many of the
  $e_j$. Let $\Gamma_{\mathsf{inf}}$ be the set of those
  $\gamma\in\Gamma$ such that $(\gamma,0)$ is scheduled infinitely
  many times during $\rho$. Since $\rho$ is progressive, we know that for
  some $k\ge 1$, each $e_j$ with $k\ge j$ only spawns threads from
  $\Gamma_{\mathsf{inf}}$. Now since $\mmap'_j\ge_1\mmap_j$, this
  implies that for $j\ge k$, $\mmap'_j$ only contains threads $\gamma$
  (together with the context-switch number when they are produced)
  such that $\gamma\in\Gamma_{\mathsf{inf}}$.

  Therefore, we obtain a new run $\rho'$ of our DCPS by replacing each
  $e_j$ by $e'_j$ for all $j\ge k$. Then, the resulting run is progressive,
  because the additionally spawned threads all belong to
  $\Gamma_{\mathsf{inf}}$. Moreover, $\rho'$ is starving.
\end{proof}
\subsection{Proof of Lemma~\ref{spawn-bounded-starving}}
\spawnBoundedStarving*
\begin{proof}
  If $\cA$ has a starving run $\rho$, then it also has a consistent
  run by \cref{starving-consistent}. Hence, there are a number
  $i\in[0,K]$, configurations $c_1,c_2,\ldots$ and thread executions
  $e_1,e_2,\ldots$ that produce $\mmap_1,\mmap_2,\ldots$ such that the
  conditions for a consistent run are satisfied.  In particular, with
  $V_t=\{\mmap_j\mid \text{$e_j$ has type $t$}\}$, the tuple
  $\frakV=(V_t)_{t\in\cT}$ is $\frakS_{\cA}$-consistent. Then, by definition of
  $\frakS_{\cA}$-consistency, the tuple $\frakm=\alpha_B(\frakV)$ is also
  $\frakS_{\cA}$-consistent.  Therefore, the run $\rho$ is $(i,\frakm)$-starving.

  Conversely, suppose $\rho$ is a shallow $(i,\frakm)$-starving run
  for some $i\in[0,K]$ and some $\frakS$-consistent tuple
  $\frakm\in\powerset{[0,B]^{\Lambda}}^{\cT}$. Then we have configurations
  $c_1,c_2,\ldots$, and thread executions $e_1,e_2,\ldots$ that
  produce $\mmap_1,\mmap_2,\ldots$ such that the conditions of
  $(i,\frakm)$-starvation are satisfied.

  Let $V_t=\{\mmap_j\mid \text{$e_j$ has type $t$}\}$ and
  $\frakV=(V_t)_{t\in\cT}$.  Since $\rho$ is shallow, we know that
  each $V_t$ is a finite set.  Moreover, since $\frakm=\alpha_B(\frakV)$ is
  $\frakS_{\cA}$-consistent, \cref{abstraction-consistency} tells us that
  $\frakV$ is $\frakS_{\cA}$-consistent as well. Thus, $\rho$ is
  consistent. Therefore $\cA$ also has a starving run by
  \cref{starving-consistent}.
\end{proof}

\subsection{Freezing $\DCPS$}

In this section we prove \cref{liveness-freezing}, which states that it is decidable whether a freezing $\DCPS$ $\cA$ has a $K$-context bounded progressive run, and if such a run exists then we can always assume it to be shallow. To this end we construct a $\DCPS$ $\cA$ that simulates $\cA$ in a way that preserves progressiveness and shallowness, while using at most $2K + 1$ context switches. The stated properties for freezing $\DCPS$ then follow from our results on $\DCPS$ (\cref{thm:fair-nterm}, \cref{lem:progDCPS}, and \cref{dcps-spawn-bounded}).

\subsubsection*{Idea}
The only differences between $\DCPS$ and freezing $\DCPS$ are the presence of frozen threads (marked with a top of stack symbol in $\Gamma^\freeze$) in configurations, and the rule {\sc Unfreeze}. To simulate the former we just add $\Gamma^\freeze$ to the stack alphabet and introduce a few new transition rules that spawn the initial frozen thread at the start of each run. To simulate {\sc Unfreeze} for a rule $\freezerule{g}{g'}{\gamma}{\gamma'}$, we follow these steps starting in global state $g$:
\begin{itemize}
  \item Resume a thread with $\gamma'$ as top of stack symbol, change it to $\gamma'^\freeze$ then swap out this thread.
  \item Resume a thread with $\gamma^\freeze$ as top of stack symbol, change it to $\gamma$ and go to $g'$ with this thread staying active.
\end{itemize}
While these steps change the stack symbols correctly and make the right thread active, there are two things of note here: Firstly, if $\gamma = \gamma'$ then the very thread swapped out at the end of the first step could be resumed at the start of the second step, which is undesired behavior. To remedy this, we add a second copy of $\Gamma^\freeze$ to our stack alphabet, which we call $\bar{\Gamma}^\freeze$. We modify the simulation so that each time a thread with top of stack symbol in $\Gamma^\freeze$ is swapped out in the first step, a thread with top of stack symbol in $\bar{\Gamma}^\freeze$ is resumed in the second step and vice-versa. Thus the two steps can no longer act on the very same thread, which fixes the issue.

Secondly, \emph{freezing} a thread in the first step causes it to make a context switch, which does not match the specification of {\sc Unfreeze}. However, \emph{unfreezing} a thread makes it active and therefore adds a context switch once the thread is swapped out again. Thus, a thread can be frozen and unfrozen a total number of $K + 1$ times each in the freezing $\DCPS$ $\cA$. In the $\DCPS$ $\cA'$ such a thread would then make $2K + 2$ context switches, with the last one happening right after it was active from being unfrozen the $(K + 1)$th time. This means that the highest context switch number it reaches while being active is $2K + 1$, which is exactly what we chose as the context switch bound for $\cA'$.

With this increased context switch bound, we need to make sure that threads do not make more than $K + 1$ context switches in $\cA'$, if they are resumed without being frozen. To this end we just artificially increase the context switch number by forcing an additional context switch every time a non-frozen thread is resumed. This mirrors the extra context switch caused by freezing, and therefore makes these threads behave correctly with a context switch bound of $2K + 1$ as well.

\subsubsection*{Formal construction}
Let $\cA = (G, \Gamma, \Delta, g_0, \gamma_0, \gamma_f)$ with $\Delta = \Deltac \cup \Deltai \cup \Deltar \cup \Deltat \cup \Deltau$ be a freezing $\DCPS$. We construct the $\DCPS$ $\cA' = (G', \Gamma', \Delta', g_0', \gamma_0')$, where:
\begin{itemize}
  \item $\Gamma' = \{\gamma_0'\} \cup \Gamma \cup \bar{\Gamma} \cup \Gamma^\freeze \cup \bar{\Gamma}^\freeze$ with $\bar{\Gamma} = \{\bar{\gamma} | \gamma \in \Gamma\}$, $\Gamma^\freeze = \{\gamma^\freeze | \gamma \in \Gamma\}$, and $\bar{\Gamma}^\freeze = \{\bar{\gamma}^\freeze | \gamma \in \Gamma\}$.
  \item $G' = G \cup \bar{G} \cup (G \times \Gamma^\freeze) \cup (G \times \bar{\Gamma}^\freeze)$ with $\bar{G} = \{\bar{g} | g \in G\}$
  \item $\Delta' = \Deltac' \cup \Deltai' \cup \Deltar' \cup \Deltat'$ consists of the following transition rules:
  \begin{enumerate}
    \item For the initial configuration:
    \begin{enumerate}
      \item $g_0' \mapsto g_0' \lhd \gamma_0' \in \Deltar'$.
      \item $g_0'|\gamma_0' \hookrightarrow g_0'|\gamma_f^\freeze \triangleright \gamma_0 \in \Deltac'$.
      \item $g_0'|\gamma_f^\freeze \mapsto g_0|\gamma_f^\freeze \in \Deltai'$.
    \end{enumerate}
    \item For each rule $g_1 \mapsto g_2 \lhd \gamma \in \Deltar$:
    \begin{enumerate}
      \item $g_1 \mapsto \bar{g}_2 \lhd \gamma \in \Deltar'$.
      \item $\bar{g}_2|\gamma \hookrightarrow \bar{g}_2|\bar{\gamma} \in \Deltac'$.
      \item $\bar{g}_2|\bar{\gamma} \mapsto \bar{g}_2|\bar{\gamma} \in \Deltai'$.
      \item $\bar{g}_2 \mapsto g_2 \lhd \bar{\gamma} \in \Deltar'$.
      \item $g_2|\bar{\gamma} \hookrightarrow g_2|\gamma \in \Deltac'$.
    \end{enumerate}
    \item For each rule $\freezerule{g_1}{g_2}{\gamma_1}{\gamma_2} \in \Deltau$:
    \begin{enumerate}
      \item $g_1 \mapsto (g_2,\gamma_1^\freeze) \lhd \gamma_2 \in \Deltar'$.
      \item $(g_2,\gamma_1^\freeze)|\gamma_2 \hookrightarrow (g_2,\gamma_1^\freeze)|\gamma_2 \in \Deltac'$.
      \item $(g_2,\gamma_1^\freeze)|\gamma_2 \mapsto (g_2,\gamma_1^\freeze)|\bar{\gamma}_2^\freeze \in \Deltai'$.
      \item $(g_2,\gamma_1^\freeze) \mapsto (g_2,\gamma_1^\freeze) \lhd \gamma_1^\freeze \in \Deltar'$.
      \item $(g_2,\gamma_1^\freeze)|\gamma_1^\freeze \hookrightarrow g_2|\gamma_1 \in \Deltac'$.
    \end{enumerate}
    \item Furthermore, for each rule $\freezerule{g_1}{g_2}{\gamma_1}{\gamma_2} \in \Deltau$:
    \begin{enumerate}
      \item $g_1 \mapsto (g_2,\bar{\gamma}_1^\freeze) \lhd \gamma_2 \in \Deltar'$.
      \item $(g_2,\bar{\gamma}_1^\freeze)|\gamma_2 \hookrightarrow (g_2,\bar{\gamma}_1^\freeze)|\gamma_2 \in \Deltac'$.
      \item $(g_2,\bar{\gamma}_1^\freeze)|\gamma_2 \mapsto (g_2,\bar{\gamma}_1^\freeze)|\gamma_2^\freeze \in \Deltai'$.
      \item $(g_2,\bar{\gamma}_1^\freeze) \mapsto (g_2,\bar{\gamma}_1^\freeze) \lhd \bar{\gamma}_1^\freeze \in \Deltar'$.
      \item $(g_2,\bar{\gamma}_1^\freeze)|\bar{\gamma}_1^\freeze \hookrightarrow g_2|\gamma_1 \in \Deltac'$.
    \end{enumerate}
    \item Finally, $\Delta'$ contains some unaltered rules of $\Delta$:
    \begin{enumerate}
      \item $r \in \Deltac'$ for each rule $r \in \Deltac$.
      \item $r \in \Deltai'$ for each rule $r \in \Deltai$.
      \item $r \in \Deltat'$ for each rule $r \in \Deltat$.
    \end{enumerate}
  \end{enumerate}
\end{itemize}

\begin{lemma}
  The freezing $\DCPS$ $\cA$ has a (shallow) progressive run that respects the context switch bound $K$ iff the $\DCPS$ $\cA'$ has a (shallow) progressive run that respects the context switch bound $2K + 1$.
\end{lemma}

\begin{proof}
  From a given infinite run of $\cA$ we construct an infinite run of $\cA'$ in such a way that progressiveness and shallowness are preserved. Then we argue that this construction can also be done backwards, starting with a run of $\cA'$.
  
  Let $\rho$ be an infinite run of $\cA$. We construct the infinite run $\rho'$ of $\cA'$ by induction on the length of a prefix of $\rho$. Say $\rho$ reaches the configuration $c = \langle g, t_a, \mmap \rangle$ after finitely many steps. Intuitively we want $\rho'$ to reach a configuration $c'$ that contains threads with the same stack contents as $c$, except for the top of stack symbol of the frozen thread, but with different context switch numbers. Any thread with cs-number $k$ in $c$ has cs-number $2k + 1$ in $c'$ if it was an active or frozen thread, and cs-number $2k$ otherwise. Formally $c' = \langle g, t_a', \mmap \rangle$, where:
  \begin{itemize}
    \item If $t_a = \#$ then $t_a' = \#$. If $t_a = (w, k)$ then $t_a' = (w, 2k + 1)$.
    \item For all $k \in \nats$, $w \in \Gamma^*$: $\mmap'(w,2k) = \mmap(w,k)$.
    \item For all $k \in \nats$, $w \in \Gamma^*$, $\gamma \in \Gamma$: $\mmap'(\gamma^\freeze w, 2k+1) = \mmap(\gamma^\freeze w, k)$ and $\mmap'(\bar{\gamma}^\freeze w, 2k+1) = 0$, or $\mmap'(\bar{\gamma}^\freeze w, 2k+1) = \mmap(\gamma^\freeze w, k)$ and $\mmap'(\gamma^\freeze w, 2k+1) = 0$.
    \item For all $k \in \nats$, $w \in \Gamma^*$, $\gamma \in \Gamma$: $\mmap'(w, 2k+1) = \mmap'(\gamma^\freeze w, 2k) = \mmap'(\bar{\gamma}^\freeze w, 2k) = 0$.
  \end{itemize}
  The second to last bullet point means that for a frozen thread in $c$ with top of stack symbol $\gamma^\freeze$, the corresponding thread in $c'$ has either $\gamma^\freeze$ or $\bar{\gamma}^\freeze$ as its top of stack symbol. The last bullet point simply ensures that $c'$ contains no additional threads (that do not correspond to a thread of $c$).
  
  Now to construct $\rho'$ inductively. If $c = \langle g_0, \#, \multi{(\gamma_0, 0)} + \multi{(\gamma_f^\freeze, 0)} \rangle$ is the initial configuration of $\cA$ then we use the following transitions in $\rho'$ starting with the initial configuration of $\cA'$:
  \begin{align*}
    \langle g_0', \#, \multi{(\gamma_0', 0)} \rangle &\xmapsto{\,\text{1a}\,} \langle g_0', (\gamma_0',0), \emptyset \rangle\\
    &\xrightarrow{\text{1b}} \langle g_0', (\gamma_0',0), \multi{(\gamma_0, 0)} \rangle\\
    &\xmapsto{\,\text{1c}\,} \langle g_0, \#, \multi{(\gamma_0, 0)} + \multi{(\gamma_f^\freeze, 1)} \rangle = c'.
  \end{align*}
  Here and for all following transitions of $\cA'$ we label each arrow with the corresponding transition rule of $\Delta'$. Since we reach the correct $c'$, this concludes the base case.
  
  For the inductive case assume that we reach $c$ after at least one step on $\rho$. We make a case distinction regarding the transition $\widetilde{c} \rightarrow c$ or $\widetilde{c} \mapsto c$ of $\rho$ that reaches $c$:

  \begin{description}
    \item [Case $\langle g_1, \#, \mmap + \multi{(\gamma w, i)} \rangle \mapsto \langle g_2, (\gamma w, i), \mmap \rangle$ due to $g_1 \mapsto g_2 \lhd \gamma \in \Deltar$:]\ \\
    We add the following transitions to $\rho'$, starting from the configuration obtained by applying the induction hypothesis to $\widetilde{c}$:
    \begin{align*}
      \langle g_1, \#, \mmap' + \multi{(\gamma w, 2i)} \rangle &\xmapsto{\,\text{2a}\,} \langle \bar{g}_2, (\gamma w, 2i), \mmap' \rangle\\
      &\xrightarrow{\text{2b}} \langle \bar{g}_2, (\bar{\gamma} w, 2i), \mmap' \rangle\\
      &\xmapsto{\,\text{2c}\,} \langle \bar{g}_2, \#, \mmap' + \multi{(\bar{\gamma} w, 2i+1))} \rangle\\
      &\xmapsto{\,\text{2d}\,} \langle g_2, (\bar{\gamma} w, 2i+1), \mmap' \rangle\\
      &\xrightarrow{\text{2e}} \langle g_2, (\gamma w, 2i+1), \mmap' \rangle = c'.
    \end{align*}

    \item [Case $\langle g_1, \#, \mmap + \multi{\gamma_1^\freeze w_1, i} + \multi{\gamma_2 w_2, j} \rangle \mapsto \langle g_2, (\gamma_1 w_1, i), \mmap + \multi{\gamma_2^\freeze w_2, j} \rangle$]
    \item [due to $\freezerule{g_1}{g_2}{\gamma_1}{\gamma_2} \in \Deltau$:]\ \\
    We again start from the configuration obtained by applying the induction hypothesis to $\widetilde{c}$. However, we need to make a case distinction based on the top of stack symbol of the thread in this configuration that corresponds to the frozen thread of $\widetilde{c}$. If this symbol is in $\Gamma^\freeze$, we add the following transitions to $\rho'$:
    \begin{align*}
      &\langle g_1, \#, \mmap' + \multi{\gamma_1^\freeze w_1, 2i + 1} + \multi{(\gamma_2 w_2, 2j)} \rangle\\
      &\xmapsto{\,\text{3a}\,} \langle (g_2,\gamma_1^\freeze), (\gamma_2 w_2, 2j), \mmap' + \multi{\gamma_1^\freeze w_1, 2i + 1} \rangle\\
      &\xrightarrow{\text{3b}} \langle (g_2,\gamma_1^\freeze), (\bar{\gamma}_2^\freeze w_2, 2j), \mmap' + \multi{\gamma_1^\freeze w_1, 2i + 1} \rangle\\
      &\xmapsto{\,\text{3c}\,} \langle (g_2,\gamma_1^\freeze), \#, \mmap' + \multi{\gamma_1^\freeze w_1, 2i + 1} + \multi{(\bar{\gamma}_2^\freeze w_2, 2j + 1)} \rangle\\
      &\xmapsto{\,\text{3d}\,} \langle (g_2,\gamma_1^\freeze), (\gamma_1^\freeze w_1, 2i + 1), \mmap' + \multi{(\bar{\gamma}_2^\freeze w_2, 2j + 1)} \rangle\\
      &\xrightarrow{\text{3e}} \langle g_2, (\gamma_1 w_1, 2i + 1), \mmap' + \multi{(\bar{\gamma}_2^\freeze w_2, 2j + 1)} \rangle = c'.
    \end{align*}
    Otherwise, if the symbol is in $\bar{\Gamma}^\freeze$, we add the following transitions to $\rho'$:
    \begin{align*}
      &\langle g_1, \#, \mmap' + \multi{\bar{\gamma}_1^\freeze w_1, 2i + 1} + \multi{(\gamma_2 w_2, 2j)} \rangle\\
      &\xmapsto{\,\text{4a}\,} \langle (g_2,\bar{\gamma}_1^\freeze), (\gamma_2 w_2, 2j), \mmap' + \multi{\bar{\gamma}_1^\freeze w_1, 2i + 1} \rangle\\
      &\xrightarrow{\text{4b}} \langle (g_2,\bar{\gamma}_1^\freeze), (\gamma_2^\freeze w_2, 2j), \mmap' + \multi{\bar{\gamma}_1^\freeze w_1, 2i + 1} \rangle\\
      &\xmapsto{\,\text{4c}\,} \langle (g_2,\bar{\gamma}_1^\freeze), \#, \mmap' + \multi{\bar{\gamma}_1^\freeze w_1, 2i + 1} + \multi{(\gamma_2^\freeze w_2, 2j + 1)} \rangle\\
      &\xmapsto{\,\text{4d}\,} \langle (g_2,\bar{\gamma}_1^\freeze), (\bar{\gamma}_1^\freeze w_1, 2i + 1), \mmap' + \multi{(\gamma_2^\freeze w_2, 2j + 1)} \rangle\\
      &\xrightarrow{\text{4e}} \langle g_2, (\gamma_1 w_1, 2i + 1), \mmap' + \multi{(\gamma_2^\freeze w_2, 2j + 1)} \rangle = c'.
    \end{align*}
  \end{description}
  The remaining four cases use the same transition rules for both $\rho$ and $\rho'$, as the rules defined in (5) are ones already present in $\Delta$:
  \begin{description}
    \item [Case $\langle g_1, (\gamma w,i), \mmap \rangle \rightarrow \langle g_2, (w_2 w_1,i), \mmap \rangle$ due to $g_1|\gamma \hookrightarrow g_2|w_2 \in \Deltac$:]%
    \[\langle g_1, (\gamma w,2i+1), \mmap' \rangle \xrightarrow{\text{5a}} \langle g_2, (w_2 w_1,2i+1), \mmap' \rangle = c'.\]
    
    \item [Case $\langle g_1, (\gamma_1 w_1,i), \mmap \rangle \rightarrow \langle g_2, (w_2 w_1,i), \mmap + \multi{(\gamma_2,0)} \rangle$ due to $g_1|\gamma_1 \hookrightarrow g_2|w_2 \triangleright \gamma_2 \in \Deltac$:]%
    \[\langle g_1, (\gamma_1 w_1,2i+1), \mmap' \rangle \xrightarrow{\text{5a}} \langle g_2, (w_2 w_1,2i+1), \mmap' + \multi{(\gamma_2,0)} \rangle = c'.\]

    \item [Case $\langle g_1, (\gamma w_1, i), \mmap \rangle \mapsto \langle g_2, \#, \mmap + \multi{(w_2 w_1, i+1)} \rangle$ due to $g_1|\gamma \mapsto g_2|w_2 \in \Deltai$:]%
    \[\langle g_1, (\gamma w_1, 2i+1), \mmap' \rangle \xmapsto{\,\text{5b}\,} \langle g_2, \#, \mmap' + \multi{(w_2 w_1, 2i+2)} \rangle = c'.\]
    
    \item [Case $\langle g_1, (\varepsilon, i), \mmap \rangle \mapsto \langle g_2, \#, \mmap \rangle$ due to $g_1 \mapsto g_2 \in \Deltat$:]%
    \[\langle g_1, (\varepsilon, 2i+1), \mmap' \rangle \xmapsto{\,\text{5c}\,} \langle g_2, \#, \mmap' \rangle = c'.\]
  \end{description}
  Each $c'$ reached in any of the cases correctly corresponds to $c$ as described above. Spawn boundedness is equivalent for $\rho$ and $\rho'$ since corresponding segments where a thread is active make the same amount of spawns in both runs. Progressiveness is also equivalent: Every transition that makes a thread active in $\rho$ (either by resuming or unfreezing) results in a corresponding thread becoming active in $\rho'$. Conversely every two resume transitions due to rules from (2) in $\rho'$ correspond to a single resume transition in $\rho$, and each resume transition due to (3d) or (4d) corresponds to an unfreezing transition in $\rho$. Resume transitions due to (3a) or (4a) cause a thread to gain a top of stack symbol from $\Gamma^\freeze \cup \bar{\Gamma}^\freeze$, which can be the case for only one thread at a time. Thus, if $\rho'$ is progressive, this very thread has to be resumed, causing a later resume transition due to (3d) or (4d), which we already discussed. Finally, termination occurs at context switch number $k$ in $\rho$ iff it occurs at $2k + 1$ in $\rho'$, which matches the relationship between the two context switch bounds for these runs.
  
  For the other direction we start with an infinite run $\rho'$ of $\cA'$. Observe that $\rho'$ decomposes into infixes that have the same form as the transition sequences obtained during the construction for the other direction. This is because for each of the alphabets $\bar{\Gamma}$, $\Gamma^\freeze$, and $\bar{\Gamma}^\freeze$ there is always at most one thread that has a top of stack symbol in that alphabet; which causes the transition rules defined for $\cA'$ to not allow for any other types of behavior. Thus, we can do the previous construction backwards and obtain a run $\rho$ of $\cA$ that has equivalent progressiveness and shallowness properties when compared to $\rho'$.
\end{proof}

\subsection{Detailed construction of $\cA_{(i,\fraku)}$}\label{appendix-starvation-detailed-construction}
We construct a freezing DCPS
$\cA_{(i,\fraku)}=(G',\Gamma',\Delta',g'_0,\gamma'_0,\gamma_\dagger)$.  Here,
$g'_0$ is a fresh state and $\gamma'_0$ are fresh stack symbols, which
are both used for initialization. Moreover, $\gamma_\dagger$ is a fresh stack
symbol that will be the top of stack of the initially frozen thread.

In addition to the stack of a thread in $\cA$, a thread in
$\cA_{(i,\fraku)}$ tracks some extra information
$(t,j,\bar{\mmap},\bar{\nmap})$, where
\begin{itemize}
\item $t$ is the type of the current thread execution,
\item $j$ is the number of
  segments that have been completed,
\item $\bar{\mmap}$ is the guess for $\alpha_B(\mmap)$,
  where $\mmap$ is the production of the thread execution,
\item $\bar{\nmap}$ is $\alpha_B(\nmap)$, where
  $\nmap\in\multiset{\Lambda}$ is the multiset that has been produced
  so far.
\end{itemize}
When a thread is inactive, this extra information is stored on its top
of stack symbol.  For this, we need the new alphabet
$\Gamma'=\Gamma\cup\tilde{\Gamma}\cup\{\gamma'_0,\gamma_\dagger\}$, where
\[ \tilde{\Gamma}=\{(\gamma,t,j,\bar{\mmap},\bar{\nmap}) \mid \gamma\in\Gamma,~\text{$t$ is a type},~\bar{\mmap},\bar{\nmap}\in[0,B]^{\Lambda} \}. \]
While a thread is active, this extra information is stored in the
global state.  This makes it easier to update it, e.g.\ when the stack
is popped. Moreover, we need a global state $\widehat{g}$ for each
$g$, which enforce that the information is transfered from the stack
to the global state. In order to execute the initially frozen thread with top-of-stack $\gamma_\dagger$,
we need global states $g^{\dagger,j}$ for $g\in G$ and $j\in[0,K]$.

Finally, we need special global states $g'_j$ for $j\in[0,K]$ to
execute an initial helper thread with top-of-stack $\gamma'_0$. Thus,
we have
$G'=G\cup \tilde{G}\cup \widehat{G}\cup\{g'_j,g^{\dagger,j}\mid
j\in[0,K]\}$, where
\begin{align*}
  \tilde{G}&=\{(g,t,j,\bar{\mmap},\bar{\nmap}) \mid g\in G,~\text{$t$ is a type},~j\in[0,K],~\bar{\mmap},\bar{\nmap}\in[0,B]^{\Lambda}\}, \\
  \widehat{G}&=\{\widehat{g}\mid g\in G\}.
\end{align*}
We now describe the transition rules of $\cA_{(i,\fraku)}$.

\subsubsection*{Rules for initialization}
We begin the description of a few rules that serve to initialize our
freezing DCPS. The initial configuration of our freezing DCPS is
$\langle g'_0,\bot,\multi{(\gamma'_0,0),
  (\gamma_\dagger^\freeze,0)}\rangle$. Since in our simulation, we
need that each thread in the bag is already annotated with its extra
information, we use a helper thread with top-of-stack $\gamma'_0$ to (i)~spawn a thread simulating
$\gamma_0$ and also (ii)~guess its extra information. Thus we have a
resume rule
\[ g'_0\mapsto g'_1\lhd \gamma'_0 \]
to resume $(\gamma'_0,0)$. Then in this new thread, we spawn $(\gamma_0,0)$, but with extra information. Thus, we have a creation rule
\[ g'_0|\gamma'_0\hookrightarrow g'_1|\gamma'_0\rhd (\gamma_0,t,0,\bar{\mmap},\emulti) \]
for every type $t$ and $\bar{\mmap}\in[0,B]^{\Lambda}$. After this, our helper thread has to
complete $K$ context-switches. This means, it has a interrupt rules
\[ g'_j|\gamma'_0\mapsto g'_{j+1}|\gamma'_0 \]
for $j\in[1,K-1]$ and a rule to remove $\gamma'_0$:
\[ g'_K|\gamma'_0\hookrightarrow g'_K|\varepsilon \]
Then, from $g'_K$ with an empty stack, we can only use the termination rule
\[ g'_K\mapsto g_0 \]
which enters the global state that corresponds to the initial global state $g_0$ of $\cA$.
\subsubsection*{Creation rules}
An internal action of a thread is simulated in the obious way. For each rule $g|\gamma\hookrightarrow g'|w'$, we have a rule
\[ (g,t,j,\bar{\mmap},\bar{\nmap})|\gamma\hookrightarrow (g,t,j,\bar{\mmap},\bar{\nmap})|w' \]
for every type $t$, $j\in[0,K]$, and $\bar{\mmap},\bar{\nmap}\in[0,B]^{\Lambda}$.

When our DCPS spawns a new thread, it immediately guesses its type
and production abstraction $\bar{\mmap}$. Moreover, it sets its
segment counter to $0$ and sets $\bar{\nmap}$ to
$\emulti$. Hence, for every creation rule
$g|\gamma\hookrightarrow g'|w'\rhd\gamma'$, we have a rule
\[ (g,t,j,\bar{\mmap},\bar{\nmap})|\gamma\hookrightarrow (g,t,j,\bar{\mmap},\alpha_B(\bar{\nmap}+\multi{(\gamma,j)}))|w'\rhd (\gamma',p',0,\bar{\mmap}',\emulti) \]
for types $t,t'$, $j\in[0,K]$, and
$\bar{\mmap},\bar{\mmap}',\bar{\nmap}\in[0,B]^{\Lambda}$. Note that
changing $\bar{\nmap}$ to $\alpha_B(\bar{\nmap}+\multi{(\gamma,j)})$
records that the current thread has spawned a thread $\gamma$ in segment $j$.
Note that we perform this guessing of extra information with every
newly spawned thread, not just those that will be frozen and hence
yield the executions $e_1,e_2,\ldots$. Therefore, there is no
requirement here that $\bar{\mmap}$ belong to $U_p$ where
$\fraku=(U_t)_{t\in\cT}$.

\subsubsection*{Interruption rules}
When we interrupt a thread, then its extra information is transferred from the global state to the top of stack and the segment counter $j$ is incremented.
Thus, for every interrupt rule $g|\gamma\mapsto g'|w'$ of $\cA$, we write $w'=\gamma''w''$ (recall that $1\le |w'|\le 2$) and include rules
\[ (g,t,j,\bar{\mmap},\bar{\nmap})|\gamma\mapsto g|(\gamma'',t,j+1,\bar{\mmap},\bar{\nmap})w'' \]
for every type $t$, $j\in[0,K-1]$, and $\bar{\mmap},\bar{\nmap}\in[0,B]^{\Lambda}$.

\subsubsection*{Resumption rules}
In order to simulate a resumption rule $g\mapsto g'\lhd\gamma$, we resume some thread with $\gamma$ (and extra information) as top of stack. The transfer of the extra finformation cannot be done in the same step, so we have an additional state $\widehat{g'}$ in which this transfer is carried out. Hence, for every resumption rule $g\mapsto g'\lhd \gamma$, we have rules
\[ g\mapsto \widehat{g'}\lhd (\gamma,p,j,\bar{\mmap},\bar{\nmap}) \]
for each type $t$, $j\in[0,K]$, $\bar{\mmap},\bar{\nmap}\in[0,B]^{\Lambda}$. In $\widehat{g'}$, we then transfer the extra information into the global state. Thus, we have
\[ \widehat{g}|(\gamma,t,j,\bar{\mmap},\bar{\nmap})\hookrightarrow (g,t,j,\bar{\mmap},\bar{\nmap})|\gamma, \]
which exist for each $g\in G$, $\gamma\in\Gamma$, $\bar{\mmap},\bar{\nmap}\in[0,B]^{\Lambda}$, $t\in\cT$, and $j\in[0,K]$.

\subsubsection*{Termination rules}
When we terminate a thread, we check that the two components $\bar{\mmap}$ and $\bar{\nmap}$ in the extra information match. Hence, for each termination rule $g\mapsto g'$, we include a rule
\[ (g,t,j,\bar{\mmap},\bar{\mmap})\mapsto g' \]
for each type $t$, $j\in[0,K]$, and $\bar{\mmap}\in[0,B]^{\Lambda}$.
\subsubsection*{Unfreezing rules}
Using freezing and unfreezing, we make sure that there are thread
executions $e_1,e_2,\ldots$ that satisfy the conditions of a
$(i,\fraku)$-starving run.  This works as follows.  Initially, we have
the frozen thread $\gamma_\dagger$. To satisfy progressiveness, a run of
$\cA_{(i,\fraku)}$ must at some point unfreeze (and switch to)
$\gamma_\dagger$. During this unfreeze, we make sure hat there exists a
thread that can play the role of $e_1$.

Thus, to unfreeze $\gamma_\dagger$, we have to freeze a thread of some type
$t$ where $\bar{\mmap}$ belongs to $U_t$:
\[ g\mapsto g^{\dagger,0}\lhd \gamma_\dagger\mathbin{\freeze} (\gamma,t,i,\bar{\mmap},\bar{\nmap}) \]
for every $g\in G$, $t\in\cT$, $\bar{\mmap}\in U_p$. The state $g^{\dagger,0}$ is a copy of $g$
in which we can only complete the execution of the $\gamma_\dagger$ thread and then return to $g$.
This means, we have interrupt rules
\[ g^{\dagger,j}|\gamma_\dagger\mapsto g^{\dagger,j+1}|\gamma_\dagger \]
for $j\in[0,K-1]$ and $g\in G$, and resume rules
\[ g^{\dagger,j}|\gamma_\dagger\mapsto g^{\dagger,j}|\gamma_\dagger \]
for $j\in[1,K]$ and $g\in G$, a rule to empty the stack in the last segment:
\[ g^{\dagger,K}|\gamma_\dagger\hookrightarrow g^{\dagger,K}|\varepsilon \]
for $g\in G$ and finally a termination rule
\[ g^{\dagger,K}\mapsto g \]
so that the simulation of $\cA$ can continue.

After this, the new frozen thread with top of stack
$(\gamma,t,i,\bar{\mmap},\bar{\nmap})$ has to be resumed (and thus
unfrozen) at some point.  To make sure that at that point, there is a
thread that can play the role of $e_2$.  Therefore, to unfreeze (and
thus resume) a thread with top of stack
$(\gamma,t,i,\bar{\mmap},\bar{\nmap})$, we freeze a thread
$(\gamma',t',i,\bar{\mmap}',\bar{\nmap}')$ with
$\bar{\mmap}'\in U_{t'}$. Note that unfreezing always happens with
context-switch number $i$, because the executions $e_1,e_2,\ldots$
have to be in their $i$-th segment in the configurations
$c_1,c_2,\ldots$. Thus, we have
\[ g\mapsto \widehat{g'}\lhd (\gamma,t,i,\bar{\mmap},\bar{\nmap})\mathbin{\freeze} (\gamma',t',i,\bar{\mmap}',\bar{\nmap}') \]
for each resume rule $g\mapsto g'\lhd \gamma$, type $t$,
$\bar{\mmap}\in U_p$, and $\bar{\nmap}\in[0,B]^{\Lambda}$, provided
that $g$ is the state specified in $t$ to be entered in the $i$-th
segment.

\subsection{Proof of Lemma~\ref{stack-vector-rational}}\label{appendix-starvation-rational}
In this section, we prove \cref{stack-vector-rational}. It will be
convenient to use a slightly modified definition of pushdown automata
for this.

\paragraph{Pushdown automata with output}
If $\Gamma$ is an alphabet, we define
$\bar{\Gamma}=\{\bar{\gamma} \mid \gamma\in\Gamma\}$.  Moreover, if
$x=\bar{\gamma}$, then we define $\bar{x}=\gamma$. For a word
$v\in(\Gamma\cup\bar{\Gamma})^*$, $v=v_1\cdots v_n$,
$v_1,\ldots,v_n\in\Gamma\cup\bar{\Gamma}$, we set
$\bar{v}=\bar{v_n}\cdots\bar{v_1}$.  A \emph{pushdown automaton with
  output} is a tuple $\cA=(Q,\Gamma,\Lambda,E,q_0,q_f)$, where $Q$ is a
finite set of \emph{states}, $\Gamma$ is its \emph{stack alphabet},
$\Lambda$ is its \emph{output alphabet},
$E\subseteq Q\times(\Gamma\cup\bar{\Gamma}\cup\{\varepsilon\})\times
\multiset{\Lambda}\times Q$ is a finite set of \emph{edges}, $q_0\in Q$
is its \emph{initial state}, and $F\subseteq Q$ is its set of
\emph{final states}. A \emph{configuration} of $\cA$ is a triple
$(q,w,\mmap)$ with $q\in Q$, $w\in\Gamma^*$, and
$\mmap\in\multiset{\Lambda}$.  For configurations $(q,w,\mmap)$ and
$(q',w',\mmap')$, we write $(q,w,\mmap)\autstep(q',w',\mmap')$ if
there is an edge $(q,u,\nmap,q')$ in $\cA$ such that
$\mmap'=\mmap+\nmap$ and (i)~if $v=\varepsilon$, then $w'=w$, (ii)~if
$v\in\Gamma$, then $w'=wv$ and (iii)~if $v=\bar{\gamma}$ for
$\gamma\in\Gamma$, then $w=w'\gamma$. By $\autsteps$, we denote the
reflexive transitive closure of $\autstep$.

For a pushdown automata with output $\cA$ and a state $q$, we define
\[ S_{\cA,q} = \{(w,\mmap)\in\Gamma^*\times\multiset{\Lambda} \mid (q_0,\varepsilon,\vec 0) \autsteps (q,w,\mmap')\autsteps(q_f,\varepsilon,\mmap)~\text{for some $\mmap'\in\multiset{\Lambda}$} \} \]
In other words, $S_{\cA,q}$ collects those pairs $(w,\mmap)$ such that
$\cA$ has a run that visits the state $q$ with stack content $w$, and
the whole run outputs $\mmap$. Clearly, \cref{stack-vector-rational}
is a consequence of the following:
\begin{lemma}\label{pushdown-output-rational}
  Given a pushdown automaton with output $\cA$ and a state $q$, the
  set $S_{\cA,q}$ is effectively rational.
\end{lemma}

In the proof of \cref{pushdown-output-rational}, it will be convenient
to argue about the dual pushdown automaton. If $\cA$ is a pushdown
automaton with output, then its \emph{dual automaton}, denoted
$\bar{\cA}$, is obtained from $\cA$ by changing each edge
$(p,u,\mmap,q)$ into $(q,\bar{u},\mmap,p)$, and switching the initial
and final state, $q_0$, and $q_f$. Moreover, we define
\[ I_{\cA,q}=\{(w,\mmap)\in\Gamma^*\times\multiset{\Lambda} \mid (q_0,\varepsilon)\autsteps (q,w,\mmap) \}, \]
which is a one-sided version of $S_{\cA,q}$: In the right component,
we only collect the multiset output until we reach $q$ and $w$.
However, using the dual automaton, we can construct $S_{\cA,q}$ from the sets
$I_{\cA,q}$. For subsets $S,T\subseteq\Gamma^*\times\multiset{\Lambda}$, we define
\[ S\otimes T = \{ (w,\mmap_1+\mmap_2)\in\Gamma^*\times\multiset{\Lambda} \mid (w,\mmap_1)\in S,~(w,\mmap_2)\in T\}. \]
Then clearly $S_{\cA,q} = I_{\cA,q}\otimes I_{\bar{\cA},q}$. The next lemma follows using
a simple product construction.
\begin{lemma}\label{rational-convolution}
  Given rational subsets
  $S,T\subseteq\Gamma^*\times\multiset{\Lambda}$, the set $S\otimes T$
  is effectively rational.
\end{lemma}

Because of $S_{\cA,q}=I_{\cA,q}\otimes I_{\bar{\cA},q}$,
\cref{pushdown-output-rational} is a direct consequence of the
following.
\begin{lemma}\label{pushdown-output-rational-oneside}
  Given $\cA$ and $q$, the set $I_{\cA,q}$ is effectively rational.
\end{lemma}
\begin{proof}
  Roughly speaking, we do the following. For each pair of states
  $p,p'$, we look at the set $K_{p,p'}\subseteq\multiset{\Lambda}$ of
  outputs that can be produced in a computation that goes from $(p,w)$
  to $(p',w)$ without ever removing a letter from $w$. Note that this
  set does not depend on $w$.  Since $K_{p,p'}$ is semi-linear, there
  is a finite automaton $\cB_{p,p'}$ that can produce $K_{p,p'}$.  We
  glue in between $p$ and $p'$ in $\cA$. In the resulting pushdown
  automaton with output $\cA'$, we can then observe that every
  configuration is reachable without ever performing a pop operation.
  Therefore, removing all pop operations from $\cA'$ and making $q$
  the only final state yields an automaton over
  $\Gamma^*\times\multiset{\Lambda}$ that accepts $I_{\cA,q}$.

  Let us do this in detail. Let $\cA=(Q,\Gamma,\Lambda,E,q_0,q_f)$ be a
  pushdown automaton with output. For each pair of states $p,p'\in Q$,
  we define
  \[ K_{p,p'}=\{\mmap\in\multiset{\Lambda} \mid
    (p,\varepsilon,\emulti)\autsteps (p',\varepsilon,\mmap) \}. \]
  It
  follows from Parikh's theorem that each set $K_{p,p'}$ is
  semi-linear. In particular, we can an automaton
  $\cB_{p,p'}$ over $\Gamma^*\times\multiset{\Lambda}$ that accepts $\{\varepsilon\}\times K_{p,p'}$.

  Let $\cA'$ be the pushdown automaton with output obtained from $\cA$
  by glueing in, between any pair $p,p'$ of states, the automaton
  $\cB_{p,p'}$. Observe that for any reachable configuration
  $(p,w,\mmap)$ of $\cA$, we can reach $(p,w,\mmap)$ in $\cA'$ without
  using pop transitions: If in a run there is a transition that pops
  some $\gamma$, we can replace the part of the run that pushes that
  $\gamma$, then performs other instructions, and finally pops
  $\gamma$, with a run in some $\cB_{p,p'}$. Conversely, any
  configuration reachable in $\cA'$ in a state that already exists in
  $Q$, is also reachable in $\cA$.

  Now let $\cA''$ be the pushdown automaton with output obtained from
  $\cA'$ by removing all pop transitions. According to our
  observation, $\cA''$ has the same set of reachable configurations in
  $Q$ as $\cA$. Since $\cA''$ has no pop transitions, it is in fact an
  automaton over $\Gamma^*\times\multiset{\Lambda}$. Hence, if we make $q$
  the final state, we obtain an automaton over
  $\Gamma^*\times\multiset{\Lambda}$ for the set $I_{\cA,q}$.
\end{proof}

\subsection{Proof of Lemma~\ref{extend-consistency}}\label{appendix-proof-extend-consistency}
\extendConsistency*
\begin{proof}
  Suppose the rational subsets are given by automata
  $\cA_1,\ldots,\cA_k$. By introducing intermediate states, we may
  clearly assume that every edge in these automata either reads a
  letter from $\Gamma$ or a singleton multiset from $\Lambda$, meaning every
  edge either belongs to $Q\times \Gamma\times\{\emulti\}\times Q$ or is of the form
  $(p,\varepsilon,\mmap,q)$ with states $p,q$ and $|\mmap|\le 1$.  Let $n$
  be an upper bound on the number of states of each $\cA_i$,
  $i\in[1,k]$.  Let $M$ be the bound from \cref{pump-dcl} and let
  $B=M(n+1)$.
  
  Now let $\frakV=(V_1,\ldots,V_k)$ be a tuple of subsets of
  $\multiset{\Lambda}$ that admits an $\frakS$-consistency witness
  $w\in\Gamma^*$. Moreover, let $\mmap\in \multiset{\Lambda}$ with
  $\alpha_B(\mmap)\in V_i$.  We prove the
  \lcnamecref{extend-consistency} by constructing a word
  $\bar{w}\in\Gamma^*$ with $w\le_\frakS \bar{w}$ such that
  $\mmap\in \wdcl{S_i}{\bar{w}}$. This implies the
  \lcnamecref{extend-consistency}, because $\bar{w}$ witnesses
  $\frakS$-consistency of the tuple $\frakV'=(V'_1,\ldots,V'_k)$ with
  $V'_i=V_i\cup\{\mmap\}$ and $V'_j=V_j$ for $j\ne i$.
 
  Since $\alpha_B(\mmap)\in V'_i$ and $w$ is a $\frakS$-consistency
  witness, we know that $\alpha_B(\mmap)\in \wdcl{S_i}{w}$. This
  means, there is a multiset $\mmap'\in\multiset{\Lambda}$ with
  $\mmap'\ge_1\alpha_B(\mmap)$ and $(w,\mmap')\in S_i$.

  Observe that if we had $\mmap'\ge_1\mmap$, we could just choose
  $w'=w$.  Moreover, in those coordinates $c\in\Lambda$ where
  $\mmap(c)<B$, we already know that $\mmap'(c)\ge\mmap(c)$,
  because $\mmap'\ge\alpha_B(\mmap)$.  For those coordinates $c$ with
  $\mmap(c)\ge B$, we will obtain $\bar{w}$ by pumping an infix in $w$.

  Let $c\in\Lambda$ with $\mmap(c)\ge B$. Then
  $\mmap'(c)\ge B$ and therefore the pair $(w,\mmap')$ is accepted on
  a run
  \[ q_0\autstep[(u_0,\mmap_0)] q_1\autstep[(u_1,\mmap_1)] \cdots \autstep[(u_B,\mmap_B)] q_B \autstep[(u_{B+1},\mmap_{B+1})] q_{B+1}\]
  where $w=u_0\cdots u_{B+1}$, $\mmap'=\mmap_0+\cdots+\mmap_{B+1}$, and
  $\mmap_j(c)\ge 1$ for each $j=1,\ldots,B$. Since $B=Mn$, there is a state  $p$ of
  $\cA_i$ that appears at least $M$ times in the sequence $q_1,\ldots,q_B$.
  This means, we have a run
  \[ q_0\autstep[(v_0,\nmap_0)] p\autstep[(v_1,\nmap_1)] p \autstep[(v_2,\nmap_2)] \cdots p \autstep[(v_{M},\nmap_{M})] q_{B+1} \]
  with $w=v_0\cdots v_{M}$, $\mmap'=\nmap_0+\cdots+\nmap_{M}$, and $\nmap_j(c)\ge 1$ for each $j=1,\ldots,M-1$.

  We claim that we can write $w=xyz$ such that
  $w=xyz\le_\frakS xy^\ell z$ for every $\ell$ and that there is a run
  $q_0\autstep[(x,\hat{\nmap}_1)] p\autstep[(y,\hat{\nmap}_2)] p
  \autstep[(z,\hat{\nmap}_3)] q_{B+1}$ with $\hat{\nmap}_2(c)\ge 1$.
  We distinguish two cases.
  \begin{enumerate}
  \item First, suppose that there is a $j\in\{1,\ldots,M-1\}$ with
    $v_j=\varepsilon$.  Then we can choose this $v_j$ as the $y$ in
    the decomposition $w=xyz$, which is clearly as desired.
  \item Suppose that $v_j\ne\varepsilon$ for every
    $j\in\{1,\ldots,M-1\}$. Then the $M$ positions in the
    decomposition $w=v_0\cdots v_M$ are pairwise distinct and we can
    apply \cref{pump-dcl}.  It clearly yields a decomposition of $w$
    as desired in our claim.
  \end{enumerate}

  Hence, the claim holds in any case. If we now choose $\ell$ high
  enough, then we find a run of $\cA_i$ on
  $(xy^\ell z,\bar{\mmap}_c)=(xy^\ell
  z,\mmap'+\ell\cdot\bar{\nmap}_2)$ where
  $\bar{\mmap}_c(c)\ge\mmap(c)$ and $\bar{\mmap}_c\ge_1\mmap'$. If we
  repeat this step for each $c\in\Lambda$ with $\mmap(c)>B$, we arrive
  at a word $\bar{w}$ and a multiset $\bar{\mmap}$ with
  $\mmap\le_1\bar{\mmap}$ and $w\le_\frakS \bar{w}$ and
  $(\bar{w},\bar{\mmap})\in S_i$. This implies
  $\mmap\in\wdcl{S_i}{\bar{w}}$.
\end{proof}

\subsection{Proof of Proposition~\ref{abstraction-consistency-rational}}\label{appendix-proof-abstraction-consistency-rational}
\abstractionConsistencyRational*
\begin{proof}
  Suppose $\frakV=(V_1,\ldots,V_k)$ is a tuple of subsets of
  $\multiset{\Lambda}$.  Clearly, if $\frakV$ is $\frakS$-consistent, then
  so is $\alpha_B(\frakV)$. The converse follows from
  \cref{extend-consistency}: Since the sets $V_1,\ldots,V_k$ are
  finite, we can start with $\alpha_B(\frakV)$ and successively add each
  multiset occurring in some $V_i$, without affecting
  $\frakS$-consistency. Then we arrive at a $\frakS$-consistent tuple
  $\frakV'=(V'_1,\ldots,V'_k)$ with $V_i\subseteq V'_i$ for
  $i\in[1,k]$. In particular, $\frakV$ is $\frakS$-consistent.
\end{proof}

\end{document}